\documentclass[11pt]{article}
\usepackage{graphicx}
\usepackage{graphics}
\usepackage{caption}
\usepackage{subcaption}
\usepackage{amsmath}
\usepackage{amsfonts}
\usepackage{amssymb}
\usepackage{latexsym}
\usepackage{setspace}
 \makeatletter \oddsidemargin 0 in \evensidemargin 0 in 
\newtheorem{theorem}{Theorem}

\newenvironment{proof}[1][Proof]{\textbf{#1.} }{\ \rule{0.5em}{0.5em}}

\newcommand{\bi}[1]{\mbox{\boldmath{$ #1 $}}}

\begin{document}
\title{Single step and multiple step forecasting in one dimensional single chirp signal using MCMC based Bayesian analysis}
 \author{ Satyaki Mazumder
 \footnote{e-mail: satyaki@iiserkol.ac.in.} 
\\Department of Mathematics and Statistics
  \\IISER Kolkata, Mohanpur Campus
 \\Mohanpur 741252, India
}
\maketitle

\begin{abstract}
Chirp signals are frequently used in different areas of science and engineering. MCMC based Bayesian inference is done here for purpose of one step and multiple step prediction in case of one dimensional single chirp signal with i.\ i.\ d.\ error structure as well as dependent error structure with exponentially decaying covariances. We use Gibbs sampling technique and random walk MCMC to update the parameters. We perform total five simulation studies for illustration purpose. We also do some real data analysis to show how the method is working in practice.
\end{abstract}
\section{\small Introduction}
\label{sec1: INT}
One dimensional single chirp signal, defined as
\begin{align}
\label{eq1:model 1}
y_t = A \cos (\alpha t + \beta t^2) + B \sin (\alpha t + \beta t^2) + \epsilon_t, 
\end{align}
for $t=1,\ldots, T$, is frequently used in different field of sciences, for example, sonar, radar, communications systems, as well as in oceanography and geology. One may see Abatzoglou (1986), Kumaresan and Verma (1987), Djuric and Kay(1990), Gini et al.\ (2000), Lin and Djuric (2000), Lahiri et al. (2012, 2014) and the references cited therein for details. Recently various types of parameter estimation techniques and their various properties have been studied for the signal (\ref{eq1:model 1}), for example see Kumaresan and Verma (1987), Djuric and Kay (1990), Gini et al.\ (2000), Nandi and Kundu (2004), Kundu and Nandi(2008), Lahiri et al. (2014), Saha and Kay (2002) and references cited therein. Kumerasan and Verma (1987) used rank reduction technique for estimating parameters of the model. Djuric and Kay (1990) proposed a linear regression technique after phase unwrapping. Gini et al.\ (2000) used maximum likelihood (ML) technique as one of their estimation technique. Saha and Kay (2002) used ML technique on superimposed chirp signals. They have used MCMC importance sampling for find maximum likelihood estimates. Lin et al.\ (2004) has found the maximum likelihood estimates of the parameters of chirp signal using simulated annealing technique. It is seen that most of the methods concentrated on ML technique in recent past. Recently some other techniques have drawn attention to the statistics community. For example, Nandi and Kundu (2004) first provided the asymptotic properties of least square estimates (LSE) of the parameters involved in one dimensional chirp signal with i.\ i.\ d.\ error structure. Kundu and Nandi (2008) extended those result in case of linear stationary errors with known auto covariance function. Lahiri et al.\ (2014) has used the least absolute deviation (LAD) technique to find the estimates of the parameters involved in the model. They also gave the asymptotic properties of LAD estimates under i.\ i.\ d.\ error structure. Although, similar to Kundu and Nandi (2008), Lahiri et al.\ (2014) assumed that the error variance is known. Therefore, it is seen that considerable amount of classical estimation techniques have been used for estimating the parameters of the chirp signal and their theoretical properties have studied in different circumstances for a while. Some Bayesian analysis of the chirp signal are found in the literature also. Lin and Djuric (2000) has done estimation of parameters of multiple of chirp signal using MCMC technique. However, they have only taken i.\ i.\ d.\ error structure into account. Moreover, it is important to mention that none of the methods, proposed so far, has taken the prediction issue into consideration. 
\par
Here we have analysed the one dimensional single chirp signal for forecasting in Bayesian paradigm. To be precise, our main aim, in this paper, is to predict a future observation through the Bayesian analysis of one dimensional single chirp signal. The advantage of using the Bayesian analysis for purpose of prediction is that it gives not only a single value or an interval, but also a complete density, which is known as posterior 
predictive density. It is also well known that the posterior predictive density is used for checking whether the model and the prior give a reasonable clarification of the uncertainty in a study (see Box E. P. George and Tiao C. George (1973) and Bickel J. Peter and Doksum A. Kjella (2007)). To achieve posterior predictive density we have used MCMC technique suitably and in the path of getting posterior predictive density, posterior densities of the parameters involved in the model have been found as by product.  Using these posterior densities one may perform the Bayesian inference of the parameters involved in the model, when required. The first part of the work mainly focuses on the i.\ i.\ d.\ error structure where we have simulated four different samples from the model (\ref{eq1:model 1}) and have illustrated the MCMC based Bayesian analysis of the model. Moreover, this MCMC based Bayesian method is applied on three different real data sets, obtained from http://archive.ics.uci.edu/ml, to see how our method is performing in practice, and in particular one of these three data set is used for multiple step forecasting. In second part, we deal with the dependent error structure though MCMC based methodologies with the same goal of forecasting. Kundu and Nandi (2008) has dealt with the model (1) assuming stationary error structure in great detail using classical inference, focusing estimation of the model parameters. However, in their numerical studies they have assumed that the auto covariance function (acf) is completely known. In our discussion it is assumed that the covariance structure of the error is exponentially decaying but unknown. In discrete time, it is known that exponentially decaying acf corresponds to stationary auto regressive process of order one (AR(1)) and Kundu and Nandi (2008) has presented the AR(1) example in their paper in numerical studies. With the same choice of the parameter values, as done in Kundu and Nandi (2008), a simulation study is done in our paper for purpose of 
illustration.
\par 
The remaining part of the paper is designed as follows. In Section 2 we describe the parameter spaces and give a overview of MCMC based Bayesian methodology. In Subsection (2.1.1) we provide the required details for Gibbs sampling, used in getting sample from joint posterior of the parameters. In the next Subsection (2.1.2) prior specifications for the parameters are made and the full conditional density functions, which are required for Gibbs sampling, are evaluated for the cases where the closed form of the conditional densities are available. In all other cases random walk MCMC is proposed (see Gamerman and Lopes (2006) and Liu, S. Jun (2008)) to update the parameters. 
In Section 3 we give the results of simulation studies based on our method, and in Section 4 we show the performance of our method when applied to real data. Section 5 deals with the dependent error structure where we assume an exponentially decaying covariance function with respect to time. Finally we give conclusion and future work in Section 6.
\section{\small Description of parameter spaces and methodology}
\label{sec2:DCP and MTH}
One dimensional single chirp signal (defined as equation (\ref{eq1:model 1})), assuming $\epsilon_t$ be random with $E(\epsilon_t)$ = 0, and var($\epsilon_t$) = $\sigma^2_{\mbox{\scriptsize $\epsilon$}}$ for all $t=1,\ldots,T$, has 5 parameters, namely, $A$, $B$, $\alpha$, $\beta$ and $\sigma_{\mbox{\scriptsize $\epsilon$}}$. Following Lahiri et. al (2012) we assume the following conditions on the parameters $A$, $B$, $\alpha$ and $\beta$:
\begin{enumerate}
\item $A^2$ + $B^2$ $< \,M^2$, for some known real number $M$.
\item $\alpha$, $\beta$ $\in \,(0,\pi)$.
\end{enumerate}
For purpose of ease in computation, we further reparametrize the above structure as follows. We take $A$ = $r\cos \theta$ and $B$ = $r\sin \theta$, with $0< r < M$ and $\theta$ $\in$ $[0,2\pi]$. Then obviously, $A^2+B^2$ = $r^2$ $<$ $M^2$. Therefore, we rewrite the model (\ref{eq1:model 1}) as, for $t$ = $1,\ldots, T$,
\begin{align}
\label{eq2:model2}
y_t = r\cos \theta \cos (\alpha t + \beta t^2) + r\sin \theta \sin (\alpha t + \beta t^2) + \epsilon_t,
\end{align}
with parameters $r$, $\theta$, $\alpha$, $\beta$ and $\sigma_{\mbox{\scriptsize $\epsilon$}}$ along with their parameter spaces\/:
\begin{itemize}
\item $r$ $\in$ $[0,M)$.
\item $\theta$ $\in$ $[0,2\pi]$.
\item $\alpha$, $\beta$ $\in \,(0,\pi)$.
\item $\sigma_{\mbox{\scriptsize $\epsilon$}}$ $\in$ $(0,\infty)$.
\end{itemize}
It needs to be noted that Lahiri et. al (2012, 2014) assumed $\sigma^2_{\mbox{\scriptsize $\epsilon$}}$ is known, unlike us. Although Nandi and Kundu (2004) provided an estimate of $\sigma^2_{\mbox{\scriptsize $\epsilon$}}$ in their theoretical study, but for numerical studies they took $\sigma^2_{\mbox{\scriptsize $\epsilon$}}$ to be known. 
For purpose of Bayesian analysis, we assume that the parameters are random and each having a prior distribution. Our main goal is to get the posterior predictive distribution of $y_{T+1}$ given the data $\bi{y}$ in Bayesian paradigm.
\subsection{\small MCMC based Bayesian methodology}
\label{2.1:MCMC MTH}
We assume that 
$[\epsilon_{t}|\sigma^2_{\mbox{\scriptsize $\epsilon$}}] \sim N(0,\sigma^2_{\mbox{\scriptsize $\epsilon$}})$, for all $t=1,\ldots,T$. 
Given $r$, $\theta$ (i.e., $A$ and $B$), $\alpha$, $\beta$ and $\sigma^2_{\mbox{\scriptsize $\epsilon$}}$, $\bi{y}$ = $(y_1,\ldots,y_T)'$, follows multivariate normal distribution with E($\bi{y})$ = $\bi{\mu}_{T}$ and cov($\bi{y})$ = $\sigma^2_{\mbox{\scriptsize $\epsilon$}}$ $I_{T\times T}$, where $\bi{\mu}_{T}$ = $(\mu_1,\ldots,\mu_T)'$, with $\mu_t$ = $r\cos \theta \cos (\alpha t + \beta t^2) + r\sin \theta \sin (\alpha t + \beta t^2)$, $t=1,\ldots, T$,  and $I_{T\times T}$ is an identity matrix of order $T\times T$. We want to find $[y_{T+1}|\bi{y}]$, the conditional distribution of $y_{T+1}$ given the data $\bi{y}$ = $(y_1,\ldots,y_T)'$. Using augmentation technique $[y_{T+1}|\bi{y}]$ can be written as 
\begin{align}
\label{eq3:posterior predictive}
[y_{T+1}|\bi{y}] = \int \int \int \int \int [y_{T+1}|\bi{y},r,\theta,\alpha,\beta,\sigma_{\mbox{\scriptsize $\epsilon$}}] [r,\theta,\alpha,\beta,\sigma_{\mbox{\scriptsize $\epsilon$}}|\bi{y}] dr\, d\theta\, d\alpha\, d\beta\, d\sigma_{\mbox{\scriptsize $\epsilon$}}
\end{align}
It is not possible to get an analytical form to the above integration. Therefore, simulation technique has to be implemented. Details follow. 
\par 
Given a sample from $[r,\theta,\alpha,\beta,\sigma_{\mbox{\scriptsize $\epsilon$}}|\bi{y}]$, $[y_{T+1}|\bi{y},r,\theta,\alpha,\beta,\sigma_{\mbox{\scriptsize $\epsilon$}}]$ will follow a normal distribution with mean $\mu_{T+1}$
and variance $\sigma^2_{\mbox{\scriptsize $\epsilon$}}$. Hence, once a sample is available from the posterior $[r,\theta,\alpha,\beta,\sigma_{\mbox{\scriptsize $\epsilon$}}|\bi{y}]$, the corresponding samples drawn from $[y_{T+1}|\bi{y},r,\theta,\alpha,\beta,\sigma_{\mbox{\scriptsize $\epsilon$}}]$ are from the posterior predictive (\ref{eq3:posterior predictive}), using which required posterior summaries can be obtained. To get the samples from $[r,\theta,\alpha,\beta,\sigma_{\mbox{\scriptsize $\epsilon$}}|\bi{y}]$, Gibbs sampler method is used. In the next subsection we give a brief description how we apply Gibbs sampler technique in the present situation.
\subsubsection{\small Required details for Gibbs sampling}
\label{2.1.1:Gibbs DTLS}
We denote the prior densities of $\alpha,\, \beta ,\, \theta$, $r$ and $\sigma^2_{\mbox{\scriptsize $\epsilon$}}$ as $[\alpha]$, $[\beta]$, $[\theta]$, $[r]$ and $[\sigma^2_{\mbox{\scriptsize $\epsilon$}}]$, respectively. Assuming the independence of prior distributions we get the joint posterior density of $r$, $\theta$, $\alpha$, $\beta$ and $\sigma_{\mbox{\scriptsize $\epsilon$}}^2$ given $\bi{y}$ as

\begin{align}
\label{eq6: jt post of all parameters}
[r, \theta, \alpha, \beta, \sigma_{\mbox{\scriptsize $\epsilon$}}^2|\bi{y}] \propto [r][\theta][\alpha][\beta][\sigma_{\mbox{\scriptsize $\epsilon$}}^2] [\bi{y}|r, \theta, \alpha, \beta, \sigma_{\mbox{\scriptsize $\epsilon$}}^2],
\end{align}
For Gibbs sampling the conditional distribution of each parameters given all others (commonly known as full conditional distribution), denoted by $[\cdot|\ldots]$, are needed and given by 
\begin{align}
\label{eq7:posterior of r}
[r|\ldots] \propto [r][\bi{y}|r, \theta, \alpha, \beta, \sigma_{\mbox{\scriptsize $\epsilon$}}^2],
\end{align}
\begin{align}
\label{eq8:posterior of theta}
[\theta|\ldots] \propto [\theta][\bi{y}|r, \theta, \alpha, \beta, \sigma_{\mbox{\scriptsize $\epsilon$}}^2],
\end{align}
\begin{align}
\label{eq9:posterior of alpha}
[\alpha|\ldots] \propto [\alpha][\bi{y}|r, \theta, \alpha, \beta, \sigma_{\mbox{\scriptsize $\epsilon$}}^2],
\end{align}
\begin{align}
\label{eq10:posterior of beta}
[\beta|\ldots] \propto [\beta][\bi{y}|r, \theta, \alpha, \beta, \sigma_{\mbox{\scriptsize $\epsilon$}}^2],
\end{align}
\begin{align}
\label{eq11:posterior of sigma}
[\sigma_{\mbox{\scriptsize $\epsilon$}}^2|\ldots] \propto [\sigma_{\mbox{\scriptsize $\epsilon$}}^2][\bi{y}|r, \theta, \alpha, \beta, \sigma_{\mbox{\scriptsize $\epsilon$}}^2].
\end{align}
\subsubsection{Prior specification and posterior densities of the parameters}
\label{2.1.2:PRIORS}
We assume the following prior distributions on the parameters
\begin{align}
\label{eq19:priors on params}
[r] &\sim \mbox{uniform}(0,M)\\[1ex]
[\theta] &\sim \mbox{uniform}(0,2\pi)\\[1ex]
[\alpha] &\sim \mbox{vonMises}(\alpha_0,\alpha_1)\\[1ex]
[\beta] &\sim \mbox{vonMises}(\beta_0,\beta_1)\\[1ex]
[\sigma^2_{\mbox{\scriptsize $\epsilon$}}] &\sim \mbox{inverse gamma}(\sigma_0,\sigma_1)
\end{align}
The closed form of the full conditional densities of $\theta$, $\alpha$, $\beta$ can not be obtained in closed form. So, we have used random walk MCMC to update these parameters. However, the conditional density of $r$ given all the others, i.e., $[r|\ldots]$ follows a truncated normal distribution with truncation between $(0,M)$ and with the mean parameter
\begin{equation}
\label{eq20:post mean of r}
E([r|\ldots]) = \frac{\sum_{t=1}^{T} y_t \left(\cos \theta \cos(\alpha t+\beta t^2) + \sin \theta \sin(\alpha t+\beta t^2)\right)}{\sum_{t=1}^{T} \left(\cos \theta \cos(\alpha t+\beta t^2) + \sin \theta \sin(\alpha t+\beta t^2)\right)^2}
\end{equation}
and variance parameter
\begin{equation}
\label{eq21: post var of r}
\mbox{var} ([r|\ldots])  = \frac{\sigma^2_{\mbox{\scriptsize $\epsilon$}}}{\sum_{t=1}^{T} \left(\cos \theta \cos(\alpha t+\beta t^2) + \sin \theta \sin(\alpha t+\beta t^2)\right)^2}.
\end{equation}
(The proof is given in the Appendix).
Moreover, it is straightforward to see that $\sigma^2_{\mbox{\scriptsize $\epsilon$}}$ given all the others, i.e., $[\sigma^2_{\mbox{\scriptsize $\epsilon$}}|\ldots]$ follows inverse gamma distribution with the parameters $\sigma_0+T/2$ and $\sigma_1+(\bi{y}-\bi{\mu}_{T})'(\bi{y}-\bi{\mu}_{T})/2$.
\section{\small Simulation Studies}
\label{3:SIMU STUDY}
In this section we have done four simulation studies to illustrate our method. We have given the true values of the parameters of simulated samples taken for our experiment in the Table \ref{table1:Simu details}. In each of four samples we keep the last observation for purpose of prediction. So, we have basically 100 observations for first three samples and 19 observation for last sample. We have applied the random walk MCMC algorithm for updating parameters $\theta$, $\alpha$, and $\beta$. For all practical purposes the true values of $M$ is not known so, we decide to take a sufficiently large value of $M$ to be in safe side. We choose $M$ to be equal to $100$. To run MCMC simulations it is needed to choose the prior parameters appropriately. For choosing mean directions in the prior distributions of $\alpha$, $\beta$ special technique is used. Loglikelihood function is maximized using simulated annealing technique with respect to the parameters $\alpha$ and $\beta$ (for details see Robert, P. and Casella, G. (2004) and Liu, S. (2008)), separately, and these values are used as initial values for MCMC iterations as well, for $\alpha$ and $\beta$, respectively. Below we discuss about the choice of prior parameters for $\sigma_{\mbox{
\scriptsize $\epsilon$}}$, $\alpha$ and $\beta$ in details.

\begin{table}[htp]
\centering
\begin{tabular}{c c c c c c c}
\hline
\hline 
Sample & No. of observations & $A$ & $B$ & $\alpha$ & $\beta$ & $\sigma_{\mbox{\scriptsize $\epsilon$}}$
\\[0.5ex]
\hline
1 & 101 & 2.0 & 1.25 & 1.75 & 1.05 & $\sqrt{0.5}$
\\
2 & 101 & 1.5 & 1.5 & 1.0 & 1.0 & $\sqrt{0.5}$
\\
3 & 101 & 2.0 & 2.0 & 1.75 & 1.75 & $\sqrt{2}$
\\
4 & 20 & 2.0 & 2.0 & 1.75 & 1.75 & $\sqrt{2}$
\\[1ex]
\hline
\end{tabular}
\caption{Description of parameters for simulated samples}
\label{table1:Simu details}
\end{table}

\par
First we decide about the choice of the prior parameters of  $\sigma_{\mbox{\scriptsize $\epsilon$}}$. The mean and variance of the inverse gamma distribution are 
$\sigma_1/(\sigma_0-1)$ = $a$, say, for $\sigma_0>1$, and $\frac{\sigma_1^2}{(\sigma_0-1)^2(\sigma_0-2)}$ = $a^2/(\sigma_0-2)$, for $\sigma_0>2$, respectively. $\sigma_0$ is taken to be $4$ for all simulated samples, so that the variance becomes $0.5a^2$. We choose $a$ to be $1$ for first and second simulated samples, and choose $2$ for third and fourth simulated samples, respectively.  Accordingly we got the values of and $\sigma_1$ for each simulated sample. The parameter values for inverse gamma distributions are listed in the Table \ref{table2:Hyperparm details}.
\par 
In case of vonMises distributions (priors of $\alpha$ and $\beta$) we choose the values of mean directions which maximises the loglikelihood function using simulated annealing, as mentioned earlier. We run $50$ iterations in each simulated annealing, separately for $\alpha$ and $\beta$, for each sample. For $\alpha$, the values are obtained as $1.81$, $0.85$, $1.72$ and $1.68$ four simulated samples, respectively. These values are chosen to be the mean directions in vonMises distributions as well as the initial values for $\alpha$ in MCMC iteration, for four samples. We get $1.00$, $1.18$, $1.78$ and $1.54$ as the values for $\beta$ in simulated annealing which maximize the loglikelihoods for four samples. $1.00$, $1.18$, $1.78$ and $1.54$ are used as the mean directions in vonMises distributions as well as the initial values for $\beta$ in MCMC iteration, for four samples, respectively. The scale parameters are chosen to be $3$ for vonMises in each of the samples for $\alpha$ and $\beta$. 
\par
The choices of the hyper parameters for $\alpha$, $\beta$ and $\sigma{\mbox{\scriptsize $\epsilon$}}$ for different simulated samples are summarized in the table \ref{table2:Hyperparm details}.
\begin{table}[htp]
\centering
\begin{tabular}{c c c c c c c c c}
\hline
\hline 
Sample & $\alpha_0$ & $\alpha_1$ & $\beta_0$ & $\beta_1$ & $\sigma_{0}$ & $\sigma_{1}$
\\[0.5ex]
\hline
1 & 1.81 & 3.0 & 1.0 & 3.0 & 4.0 & $3.0$
\\
2 & 0.85 & 3.0 & 1.18 & 3.0 & 4.0 & $3.0$
\\
3 & 1.72 & 3.0 & 1.78 & 3.0 & 4.0 & $6.0$
\\
4 & 1.68 & 3.0 & 1.54 & 3.0 & 4.0 & $6.0$
\\[1ex]
\hline
\end{tabular}
\caption{Description of choices of hyper parameters for $\alpha$, $\beta$ and $\sigma{\mbox{\scriptsize $\epsilon$}}$}
\label{table2:Hyperparm details}
\end{table}

\noindent
With the above choice of hyper parameters, $500000$ MCMC iteration have been done with burning period $50000$. We use the normal random walk proposal with variance $0.5$ to update $\theta$, $\alpha$, and $\beta$, for all these four simulated samples. The choice of this variance is set based on a pilot run of MCMC iteration. We mention here that once the sample observations are obtained from $r$ and $\theta$, we transform the sample values to that of $A$ = $r\cos(\theta)$ and $B$ = $r\sin(\theta)$. Details about the results of MCMC iteration for each of the sample are discussed here. Posterior densities along with the true values are provided in the Figures (\ref{Fig:Post of A,B,alpha,beta, sigma for sample1}), (\ref{Fig:Post of A,B,alpha,beta, sigma for sample2}), (\ref{Fig:Post of A,B,alpha,beta, for sample3}) and (\ref{Fig:Post of A,B,alpha,beta, for sample4}), for sample 1, 2, 3 and 4, respectively. Except for $\beta$ in the figure (\ref{Fig:Post of A,B,alpha,beta, for sample4}), all other true values are well within the high probability region. We have taken only 20 observations for sample 4. So, it is not unusual to notice such an incident, specially when the posteriors are 
not unimodal. It can also be noted that as soon as the number of observations are increased to $100$, the problem of $\beta$ is solved (figure (\ref{Fig:Post of A,B,alpha,beta, for sample3})). True signals along with 95\% credible intervals, obtained based on MCMC simulations, are provided in the figure (\ref{Fig:Fit for sample1,2,3,4}). It is seen that in all the cases the true signal falls well within the 95\% credible intervals. Finally, posterior predictive densities for $101$th observations of samples 1, 2, 3 and $20$th observation of sample 4, are given in the Figure (\ref{Fig:Post predictive for sample1,2,3,4}). It is seen that true future values are well within the 95\% credible interval in each of the cases, which is our main aim for this paper. 

\section{\small Real Data Analysis}
\label{4: REAL ANA}
Three real data sets have been taken from http://archive.ics.uci.edu/ml of which two are of type sonar rocks and one is of type sonar mines. Each signal contains 60 observations. Bache, K. and Lichman, M. (2013) mainly used the data for classification purpose. They got sonar signals from two different substances one is mine, other is rocks. Here we use these data sets for showing performance of our method for purpose of one step and multiple step forecasting. First we consider two different signals, the one from sonar mine and one of the two signals from sonar rock and keep the last observations for purpose of prediction. Therefore, we use $59$ observations for our analysis for the two above mentioned signals. We analyse another sonar rock signal in a different mode, in the sense that we keep last 5 observations for purpose of multiple step forecasting. That means that for this data set we only use $55$ observations for analysis. For first two signals (the sonar mine and one of the sonar rock, for which $59$ 
observations are considered for analysis), We give the 95\% credible interval based on sample observations obtained from MCMC simulations for purpose of fitting and the posterior predictive densities for purpose of prediction. For the last sonar rock signal five posterior predictive densities are given to show how more than one true future values being captured by 95\% credible intervals. 
\par 
We follow the same path for choosing the prior parameter values for $\sigma_{\mbox{\scriptsize $\epsilon$}}$ as we have done in the case of simulated samples. Here in particular we choose the value of $a$ (prior mean) to be $0.40$, $0.65$ and $0.50$ for the three data sets respectively. $\sigma_0$ has set to be $4$ as earlier, so that the variance becomes half of the square of the mean, for all the data sets considered here. Accordingly we find the values of $\sigma_1$ for each cases. The above choice of means have been done after running a pilot MCMC iterations. 
\par 
For $\alpha$ as well as for $\beta$, the scale parameters for vonMises distributions have been chosen to be $3$ for each of the data sets. The mean directions of vonMises for $\alpha$ have been set to be $2.83$ and $2.91$ for the sonar mine signal and the first sonar rock signal, respectively. Similarly, for $\beta$, we choose the mean directions of vonMises distributions to be $1.21$ and $1.32$ for the sonar mine signal and the first sonar rock signal, respectively. These values are obtained based on a small iteration ($50$ iterations) of simulated annealing technique, separately, on $\alpha$ and $\beta$ for each of the data sets. Finally, for the second sonar rock signal (in which case $55$ observations are considered for analysis), the choice of the mean directions for $\alpha$ and $\beta$ are taken to be $2.43$ and  $2.81$, obtained as a result of small number of iterations ($50$ iterations) of simulated annealing. For $r$, the value of $M$ needs to be given however, the true value of $M$ is not known here so, we choose a large value of $M$, $100$, for all these real data sets. 
\par
With these choices of the prior parameters we run $500000$ MCMC iterations with burning 
period $50000$, and the following results are noted.
Figure (\ref{Fig:Sonar_data_mines_1_2_rocks_3_fit}) provides the 95\% fit for the sonar mine signal, and the first sonar rock signal based on 59 observations. There are 60 observations for each of these signals. We have taken 59 observations for purpose of fitting and have kept 60th observation for prediction purpose. Figure (\ref{Fig:Sonar_data_1_2_posterior_predictive_rocks_3_predictive}) gives the posterior predictive densities of 60th observations for the sonar mine signal and the first sonar rock signal. We have noted from figure (\ref{Fig:Sonar_data_mines_1_2_rocks_3_fit}) that 95\% credible intervals mostly contain the true signals in both the two cases. 95\% credible completely contains the true sonar rock signal. However, three true observations (5th, 29th and 54th) fall outside the 95\% credible interval for the sonar mine signal (first graph of figure (\ref{Fig:Sonar_data_mines_1_2_rocks_3_fit})). At the same time it is noticed that the pattern of the signal has been best captured for the sonar 
mine signal. On the other hand, from Figure (\ref{Fig:Sonar_data_1_2_posterior_predictive_rocks_3_predictive}) it is observed that the true values of 60th observation fall well within the credible intervals for each of the two signals, the sonar mine signal and the first sonar rock signal.
\par
The second sonar rock signal, consisting of 60 observations, is analysed as follows. We keep first 55 observations as the known data and last 5 observations for purpose of multiple step prediction, as discussed earlier. Now, in Figures (\ref{Fig:Sonar rocks predictions_1st three}) and (\ref{Fig: Sonar rocks predictions_last two}) the five posterior predictive densities are given for last five observations, respectively. It is interesting to observe that true values of 56th, 57th, 58th, 59th and 60th observations fall well within the 95\% credible region. It is notable to see that even with only $55$ observations we can predict next $5$ observations in a reasonable way.


\section{\small Dependent error structure with exponentially decaying covariances}
\label{6:Depend error}
In this section we assume that $\bi{\epsilon}$ = $(\epsilon_1,\ldots,\epsilon_T)'$, given $\sigma_{\mbox{\scriptsize $\epsilon$}}^2$ and $\rho$ has a multivariate normal distribution with mean $(0\ldots,0)'$ and covariance matrix
\[
\mbox{cov}(\bi{\epsilon}) = \sigma_{\mbox{\scriptsize $\epsilon$}}^2 \Delta_{T},
\]
where $\Delta_{T}$ = $(a_{i,j})$ is the correlation matrix of order $T\times T$, with the following structure
\begin{equation}
a_{i,j} = \begin{cases}
1 & \mbox{ if } i=j\\[1ex]
\exp{(-\rho |i-j|)} & \mbox{ otherwise},
\end{cases}
\label{eq30:correlation of epsilon}
\end{equation}
with $\rho$ $\in$ $(0,\infty)$.
Under the above assumptions $[\bi{y}|r,\theta,\alpha,\beta,\sigma_{\mbox{\scriptsize $\epsilon$}}^2,\rho]$ follows a multivariate normal distribution with the mean parameter
$\bi{\mu}_{T}$ = $(\mu_1,\ldots,\mu_T)$ ($\mu_t$ is equal to $r\cos \theta \cos (\alpha t + \beta t^2) + r\sin \theta \sin (\alpha t + \beta t^2)$, for $t=1,\ldots, T$) and the covariance matrix $\sigma_{\mbox{\scriptsize $\epsilon$}}^2 \Delta_{T}$ of order $T\times T$. Given a data $\bi{y}$, our main aim is to get a posterior predictive density of $y_{T+1}$ for one step forecasting. The density of $[y_{T+1}|\bi{y}]$ can be written as 
\begin{equation}
\label{eq22:post pred depend}
[y_{T+1}|\bi{y}] = \int \int \int \int \int \int [y_{T+1}|\bi{y},r,\theta,\alpha,\beta,\sigma_{\mbox{\scriptsize $\epsilon$}},\rho] [r,\theta,\alpha,\beta,\sigma_{\mbox{\scriptsize $\epsilon$}},\rho|\bi{y}]\, dr\, d\theta\, d\alpha\, d\beta\, d\sigma_{\mbox{\scriptsize $\epsilon$}}\, d\rho,
\end{equation}
as done in equation (\ref{eq3:posterior predictive}) for independent error structure. It has to be noted that now the number of parameter increases to $6$ from $5$ (the number of parameters present in the i.\ i.\ d.\ case). As mentioned earlier in section (\ref{2.1:MCMC MTH}), it is not possible to get an analytical form of the above integration. The same simulation technique, as done in (\ref{2.1:MCMC MTH}) is implemented here. It is easily seen that $[y_{T+1}|\bi{y},r,\theta,\alpha,\beta,\sigma_{\mbox{\scriptsize $\epsilon$}}^2,\rho]$ follows a normal distribution with mean 
\begin{equation}
\label{eq23:cond mean of y_T+1}
E(y_{T+1}|\ldots) = \mu_{T+1} + \bi{c}' (\sigma_{\mbox{\scriptsize $\epsilon$}}^2 \Delta_{T})^{-1} (\bi{y} - \bi{\mu}_{T})
\end{equation}
and variance
\begin{equation}
\label{eq24:cond var of y_T+1}
\mbox{var}(y_{T+1}|\ldots) = \sigma_{\mbox{\scriptsize $\epsilon$}}^2 - \bi{c}'(\sigma_{\mbox{\scriptsize $\epsilon$}}^2 \Delta_{T})^{-1}\bi{c},
\end{equation}
where $\mu_{T+1}$ = $r\cos(\theta)\cos(\alpha (T+1) +\beta (T+1)^2)$ + $r\sin(\theta) \sin(\alpha (T+1)+\beta (T+1)^2)$ and 
\begin{equation}
\label{eq25:cov b/w T+1 and others}
\bi{c} = \sigma_{\mbox{\scriptsize $\epsilon$}}^2 (\exp(-T\rho),\ldots,\exp(-\rho))'
\end{equation}
is a vector of order $T\times 1$, containing the covariances between $y_{T+1}$ and $(y_{1},\ldots,y_{T})$. Therefore, once a sample is available from the posterior $[r,\theta,\alpha,\beta,\sigma_{\mbox{\scriptsize $\epsilon$}},\rho|\bi{y}]$, the corresponding samples drawn from $[y_{T+1}|\bi{y},r,\theta,\alpha,\beta,\sigma_{\mbox{\scriptsize $\epsilon$}}]$ are from the posterior predictive (\ref{eq22:post pred depend}). It is good to mention here that samples can be generated from $[y_{T+k}|\bi{y}]$ as well for multiple step forecasting, with a little generalization of augmentation technique, adding each simulated $y_{T+j}$ to the previous set of data $(y_{1},\ldots,y_{T+j-1})'$ to get $(y_{1},\ldots,y_{T+j-1},y_{T+j})'$, denoted by $\bi{y}_{T+j}$, for $j$ = $1,\ldots, k-1$. Then the above MCMC technique can be used. To be precise, at each augmentation stage a single MCMC sample is required to draw from $[r,\theta,\alpha,\beta,\sigma_{\mbox{\scriptsize $\epsilon$}},\rho|\bi{y}_{T+j}']$ and once this sample is generated, it is easy to obtain a sample from $[y_{T+j+1}|r,\theta,\alpha,\beta,\sigma_{\mbox{\scriptsize $\epsilon$}},\rho,\bi{y}_{T+j}]$, denoted as $[y_{T+j+1}|\ldots]$, because
$[y_{T+j+1}|\ldots]$ will follow a normal distribution with mean
\[
E(y_{T+j+1}|\ldots) = \mu_{T+j+1} + \bi{c}_{T+j}'(\sigma_{\mbox{\scriptsize $\epsilon$}}^2\Delta_{T+j})^{-1}(\bi{y}_{T+j}-\bi{\mu}_{T+j}) 
\]
and variance 
\[
\mbox{var}(y_{T+j+1}|\ldots) = \sigma_{\mbox{\scriptsize $\epsilon$}}^2 - \bi{c}_{T+j}'(\sigma_{\mbox{\scriptsize $\epsilon$}}^2 \Delta_{T+j})^{-1}\bi{c}_{T+j},
\]
where $\mu_{T+j+1}$ is the mean at time $T+j+1$, $\bi{\mu}_{T+j}$ is the expectation of $\bi{y}_{T+j}$, $\Delta_{T+j}$ is the correlation matrix of $\bi{y}_{T+j}$, and $\bi{c}_{T+j}$ is a vector of order $(T+j)\times 1$ containing the covariances between $y_{T+j+1}$ and $\bi{y}_{T+j}$.
\par 
We give details of simulations in one step forecasting here which can be easily generalized for multiple step forecasting. To get the sample from $[r,\theta,\alpha,\beta,\sigma_{\mbox{\scriptsize $\epsilon$}},\rho|\bi{y}]$ we use Gibbs sampler technique as earlier. Assuming the independence of the prior distributions, $[r,\theta,\alpha,\beta,\sigma_{\mbox{\scriptsize $\epsilon$}},\rho|\bi{y}]$ can be written as $[r,\theta,\alpha,\beta,\sigma_{\mbox{\scriptsize $\epsilon$}},\rho|\bi{y}]$ $\propto$ $[r]$ $[\theta]$ $[\alpha]$ $[\beta]$ $[\sigma_{\mbox{\scriptsize $\epsilon$}}]$ $[\rho]$ $[\bi{y}|r,\theta,\alpha,\beta,\sigma_{\mbox{\scriptsize $\epsilon$}},\rho]$. The choice of the prior distributions for $r,\theta,\alpha,\beta$ and $\sigma_{\mbox{\scriptsize $\epsilon$}}$ is taken to be the same as done in section (\ref{2.1.2:PRIORS}). For $\rho$, we assume that 
\begin{equation}
\label{eq26:prior of rho}
[\rho] \sim \mbox{gamma}(\rho_0,\rho_1)
\end{equation}
and obtain the full conditional density of $\rho$ as 
\begin{equation}
\label{eq27:posterior of rho}
[\rho|\ldots] \propto [\rho] [\bi{y}|\ldots].
\end{equation}
The forms of the full conditional distributions of $r$, $\theta$, $\alpha$, $\beta$ and $\sigma_{\mbox{\scriptsize $\epsilon$}}$ remain the same as equations (\ref{eq7:posterior of r}), (\ref{eq8:posterior of theta}), (\ref{eq9:posterior of alpha}), (\ref{eq10:posterior of beta}), (\ref{eq11:posterior of sigma}), respectively. In the current scenario also, the closed form of the full conditional densities are available only for $\sigma_{\mbox{\scriptsize $\epsilon$}}^2$ and $r$. For rest of the parameters we use the normal random walk MCMC with variance $0.5$, as earlier, for updating. The full conditional distribution of $r$, $[r|\ldots]$, turns out to be a truncated normal distribution, truncation between $(0,M)$, with mean 
\begin{equation}
\label{eq28:post mean of r in dependent}
E(r|\ldots) = \frac{\bi{y}'\Delta_{T}^{-1}\bi{b}_{T}}{\bi{b}_{T}'\Delta_{T}^{-1}\bi{b}_{T}}
\end{equation}
and variance
\begin{equation}
\label{eq29:post var of r in dependent}
\mbox{var}(r|\ldots) = \frac{\sigma_{\mbox{\scriptsize $\epsilon$}}^2}{\bi{b}_{T}'\Delta_{T}^{-1}\bi{b}_{T}},
\end{equation}
where $\bi{b}_{T}$ is such that $\bi{\mu}_{T}$ = $r$ $\bi{b}_{T}$ (the proof is given in Appendix). It is easy to seen that the full conditional distribution of $\sigma_{\mbox{\scriptsize $\epsilon$}}^2$ is the inverse gamma distribution with the parameters $\sigma_0+T/2$ and $\sigma_1 + (\bi{y}-\bi{\mu}_{T})'\Delta_{T}^{-1}(\bi{y}-\bi{\mu}_{T})/2$.
\par 
With the above discussion a simulation study has been done here. For simulation of the data we choose the values of the parameters as $A$ = $2.93$, $B$ = $1.91$, $\alpha$ = $2.5$, $\beta$ = $0.1$, $\sigma_{\mbox{\scriptsize $\epsilon$}}^2$  = $0.5$ and $\rho$ = $\ln (2)$. The choice of the above parameter values are motivated from Kundu and Nandi (2008). The choice of the prior parameters are decided as before, that is, for $\alpha$ and $\beta$ we run a small number of iterations ($50$) of simulated annealing technique to maximize the loglikelihood, separately with respect to $\alpha$ and $\beta$ and the values are found to be $2.41$ and $0.18$ for $\alpha$ and $\beta$, respectively. We take these values of $\alpha$ and $\beta$ as initial values for MCMC iterations as well.
\par 
For $\sigma_{\mbox{\scriptsize $\epsilon$}}^2$ we set the mean of the prior density, denoted by $a$, to be $1.0$ and $\sigma_0$ to be $4.0$ as done earlier (in section (\ref{3:SIMU STUDY})) and accordingly we evaluate the values of the $\sigma_0$ and $\sigma_1$. In case of $\rho$ we set the mean of prior density, $\rho_0 / \rho_1$, 1 and the $\rho_1$ is chosen to be $2$, so that the variance, $\rho_0/\rho_1^2$, becomes one half of the mean. Accordingly we get the values of $\rho_0$ is calculated.
\par 
Using the above choice of prior parameters $500000$ MCMC iterations has been done with $50000$ burning period. Here we discuss the outcomes of the experiment. Figure (\ref{Fig:Post of A,B,alpha,beta,sigma,rho for dependent sample}) provides the posterior densities of $A$, $B$, $\alpha$, $\beta$, $\sigma_{\mbox{\scriptsize $\epsilon$}}$ and $\rho$. It is observed that the true values of these parameters fall in high probability region in each of the cases. Finally, the posterior predictive density for $101$th observation and the 95\% credible interval obtained from simulated sample for the true signal are provided in the Figure (\ref{Fig:posterior predictive and fit for dependent error}). It is seen that the true future value is well within the 95\% credible interval of the posterior density, indicating the usefulness of our MCMC based method for forecasting. It is also seen that the true signal is almost always contained in the 95\% credible interval obtained based on samples from MCMC simulations, except for three values, namely, $1$st, $37$th and $92$nd observations. $1$st and $37$th true values fall below the 2.5\% interval and $92$nd observation fall above the 97.5\% interval. However, the pattern of the true signal is very nicely described by the intervals.  

\section{\small Conclusion and future work}
\label{5:CONC}
In this paper we have shown that using appropriate MCMC simulation technique one can successfully forecast one or more future observations based on the available data on one dimensional single chirp signal in Bayesian paradigm. (In this regard we have considered both independent error covariance structure as well as dependent error covariance structure. For independent error covariance structure,) we have seen that for simulated as well as for real data our method has performed very well for purpose of forecasting. In simulation studies we have considered four different samples with different values of parameters. We have kept the last observation for purpose of forecasting for each of these samples. It is observed that true future values for different samples fall within the 95\% credible intervals in all these cases. Moreover, in these simulation studies we have shown that true values of the parameters fall in high probability regions of the posterior densities in most of the cases. For the case where we 
have used only $20$ observations, true value of $\beta$ has fallen in a very low probability region in posterior density of $\beta$ (between the two modes of the posterior density of $\beta$), which is clearly because of small sample size (note that as soon as we increase the sample size to $100$ with the same set of values of parameters, posterior density of $\beta$ has rightly captured the true value of $\beta$). Here we once again emphasise that we take $\sigma_{\mbox{\scriptsize $\epsilon$}}$ to be unknown unlike Lahiri et al. (2012, 2014). This gives more freedom for using our method in practice.
\par 
MCMC simulation technique has been applied on three real data sets, taken from
\newline 
[http://archive.ics.uci.edu/ml], to see how the method is performing in practice. In this website data on two different types of signal are available, e.g., sonar mine signal and sonar rock signal. There are $60$ observations available corresponding to each signal. For our experiment we choose one sonar mine signal and two sonar rock signal.  We apply MCMC iteration on $59$ observations for the sonar mine and one of the sonar rock signals, keeping the last observation for purpose of forecasting. $55$ observations are taken for purpose of analysis in case other sonar rock signal. We keep $5$ observations for purpose of prediction to show how the method is working for more than one future observations. We have observed that in case of the sonar mine signal and the first sonar rock signal, posterior predictive densities have nicely captured the true values of $60$th observations and the fitting to the signals are extremely well, in the sense that 95\% credible intervals, obtained from MCMC simulations, have 
nicely captured the true signal. For the second sonar rock data, we have observed that the all five true future observations have fallen within the $95\%$ credible region of the posterior predictive densities. It is encouraging to note that with only $55$ observations, using MCMC iteration technique suitably, one can predict more than one observations in a significant manner, for one dimensional single chirp signal. Through out these experiments, simulated annealing technique is used to get the initial values and mean directions of the prior distribution for the parameters $\alpha$ and $\beta$. 
 
For dependent error covariance, we consider the covariance being exponentially decaying proportional to lag difference in this paper. This special covariance structure corresponds to the case of auto regressive process with lag one on errors in discrete time domain (Chatfiled, 2003). Kundu and Nandi (2008) have done a theoretical and numerical study on the stationary error structure. However, they have assumed that the auto covariance function is fully known, that is, they have taken the error variance and covariance to be known. We, on the contrary, take $\sigma_{\mbox{\scriptsize $\epsilon$}}^2$ and $\rho$ as unknowns and assuming prior densities on these, come up with posterior densities of $\sigma{\mbox{\scriptsize $\epsilon$}}^2$ and $\rho$. In the numerical example, we have shown that the true values fall in high probability region for both $\sigma_{\mbox{\scriptsize $\epsilon$}}^2$ and $\rho$, respectively. This is clearly an improvement over the previous work. 
\par 
For future work we will consider multiple chirp signal for purpose of Bayesian analysis and forecasting. One dimensional multiple chirp signal is defined as
\begin{equation}
\label{eq14:multiple chirp}
y_{t} = \sum_{k=1}^{p} \{A_{k} \cos(\alpha_{k} t +\beta_{k} t^2) + B_{k} \sin(\alpha_{k} t + \beta_{k} t^2)\} + \epsilon_{t}, t=1,\ldots, T.
\end{equation}
For details of multiple signal one may see Saha and Kay (2002), Kundu and Nandi (2008) etc. In most of the cases $p$ is assumed to be known. We will consider $p$ as unknown and will perform TTMCMC (Das and Bhattachariya, 2014) for purpose of estimation and forecasting in Bayesian paradigm for our future work.

\section*{Acknowledgement}
The author is very grateful to Moumita Das for her insightful comments and suggestions which improve the paper in every aspect. The author is also grateful to the assistance obtained from uci machine learning data repository (http://archive.ics.uci.edu/ml) for the real data sets.

\begin{center}
\section*{Appendix I}
\end{center}
\begin{theorem}
\label{Thm1:full conditional of r}
Let $\bi{y}|(r,\theta,\alpha,\beta,\sigma^2_{\mbox{\scriptsize $\epsilon$}})$ $\sim$ $N_{T}(\bi{\mu}_{T},\sigma^2_{\mbox{\scriptsize $\epsilon$}} I_{T\times T})$, where $\bi{\mu}_{T}$ = $(\mu_1,\ldots,\mu_T)'$, with $\mu_t$ = $r\cos \theta \cos (\alpha t + \beta t^2) + r\sin \theta \sin (\alpha t + \beta t^2)$, $t=1,\ldots, T$, and $I_{T\times T}$ be the identity matrix of order $T\times T$. Also we assume that $r$ $\in$ $(0,M)$, for some known real $M$ and $[r]\sim \mbox{ uniform }(0,M)$. Then $[r|\bi{y},\theta, \alpha,\beta,\sigma^2_{\mbox{\scriptsize $\epsilon$}}]$, denoted as $[r|\ldots]$, follows a truncated normal distribution with truncation between $(0,M)$ and with the  mean parameter and variance parameter as specified in equation $(\ref{eq20:post mean of r})$ and $(\ref{eq21: post var of r})$, respectively.
\end{theorem}
\begin{proof}
We note that
\begin{align*}
[r|\ldots] &\propto [r][\bi{y}|\ldots]
\\
& \propto \exp \left[-\frac{1}{2\sigma^2_{\mbox{\scriptsize $\epsilon$}}}\sum_{t=1}^{T}\left\{y_t-r(\cos \theta \cos(\alpha t+\beta t^2)+\sin \theta \sin(\alpha t+\beta t^2))\right\}^2\right]\chi_{(0,M)}(r)
\end{align*}
We simplify the exponent term (without $-\frac{1}{2\sigma^2_{\mbox{\scriptsize $\epsilon$}}}$) below.

\begin{align*}
&\sum_{t=1}^{T}\left\{y_t-r(\cos \theta \cos(\alpha t+\beta t^2)+\sin \theta \sin(\alpha t+\beta t^2))\right\}^2
\\
&
=r^2 \sum_{t=1}^T \left(\cos \theta \cos(\alpha t+\beta t^2)+\sin \theta \sin(\alpha t+\beta t^2)\right)^2 
\\
& -2r \sum_{t=1}^{T} y_t \left(\cos \theta \cos(\alpha t+\beta t^2)+\sin \theta \sin(\alpha t+\beta t^2)\right) + \mbox{constant with respect to } r
\\
&
= \sum_{t=1}^T \left(\cos \theta \cos(\alpha t+\beta t^2)+\sin \theta \sin(\alpha t+\beta t^2)\right)^2
\\
& 
\left\{ r-\frac{\sum_{t=1}^{T} y_t \left(\cos \theta \cos(\alpha t+\beta t^2)+\sin \theta \sin(\alpha t+\beta t^2)\right)}{\sum_{t=1}^{T^{•}}\left(\cos \theta \cos(\alpha t+\beta t^2)+\sin \theta \sin(\alpha t+\beta t^2)\right)^2}\right\}^2 + \mbox{constant with respect to } r
\end{align*}
Therefore,
\begin{align*}
[r|\ldots] &\propto \exp \left[-\frac{1}{\sigma_r^2} \left\{ r-\frac{\sum_{t=1}^{T} y_t \left(\cos \theta \cos(\alpha t+\beta t^2)+\sin \theta \sin(\alpha t+\beta t^2)\right)}{\sum_{t=1}^{T} \left(\cos \theta \cos(\alpha t+\beta t^2)+\sin \theta \sin(\alpha t+\beta t^2)\right)^2}\right\}^2\right]\chi_{0,M}(r),
\end{align*}
where 
\[
\sigma^2_{r} = \frac{\sigma^2_{\mbox{\scriptsize $\epsilon$}}}{\sum_{t=1}^{T}\left(\cos \theta \cos(\alpha t+\beta t^2)+\sin \theta \sin(\alpha t+\beta t^2)\right)^2}.
\]
Hence the proof follows.  
\end{proof}


\begin{theorem}
\label{Thm2:full conditional of r in dependent case}
Let $\bi{y}|(r,\theta,\alpha,\beta,\sigma^2_{\mbox{\scriptsize $\epsilon$}},\rho)$ $\sim$ $N_{T}(\bi{\mu}_{T},\sigma^2_{\mbox{\scriptsize $\epsilon$}} \Delta_{T})$, where $\bi{\mu}_{T}$ = $(\mu_1,\ldots,\mu_T)'$, with $\mu_t$ = $r\cos \theta \cos (\alpha t + \beta t^2) + r\sin \theta \sin (\alpha t + \beta t^2)$, $t=1,\ldots, T$, and $\Delta_{T}$ be the correlation matrix or order $T\times T$, with the elements as specified in equation $(\ref{eq30:correlation of epsilon})$. Also we assume that $r$ $\in$ $(0,M)$, for some known real $M$ and $[r]\sim \mbox{ uniform }(0,M)$. Then $[r|\bi{y},\theta, \alpha,\beta,\sigma^2_{\mbox{\scriptsize $\epsilon$}},\rho]$, denoted as $[r|\ldots]$, follows a truncated normal distribution with truncation between $(0,M)$ and with the  mean parameter and variance parameter as specified in equation $(\ref{eq28:post mean of r in dependent})$ and $(\ref{eq29:post var of r in dependent})$, respectively.
\end{theorem}
\begin{proof}
We follow similar line of proof as done in Theorem (\ref{Thm1:full conditional of r}). Writing $\bi{\mu}_{T}$ = $r\bi{b}_{T}$, we clearly observe that
\begin{align*}
[r|\ldots] \propto \exp\left(-\frac{1}{2\sigma^2_{\mbox{\scriptsize $\epsilon$}}}(\bi{y}-r\bi{b}_{T})'\Delta_{T}^{-1}(\bi{y}-r\bi{b}_{T})\right) \chi_{(0,M)}(r) 
\end{align*}
After simplifying the exponent term it is readily seen that
\begin{align*}
[r|\ldots] \propto \exp\left[-\frac{\bi{b}_{T}'\Delta_{T}^{-1}\bi{b}_{T}}{2\sigma^2_{\mbox{\scriptsize $\epsilon$}}}\left(r-\frac{\bi{y}'\Delta_{T}^{-1}\bi{b}_{T}}{\bi{b}_{T}'\Delta_{T}^{-1}\bi{b}_{T}}\right)^2\right]\chi_{(0,M)}(r)
\end{align*}
Hence the proof follows.
\end{proof}

\begin{figure}[htp]
\centering
\includegraphics[height=1.5in,width=1.5in]{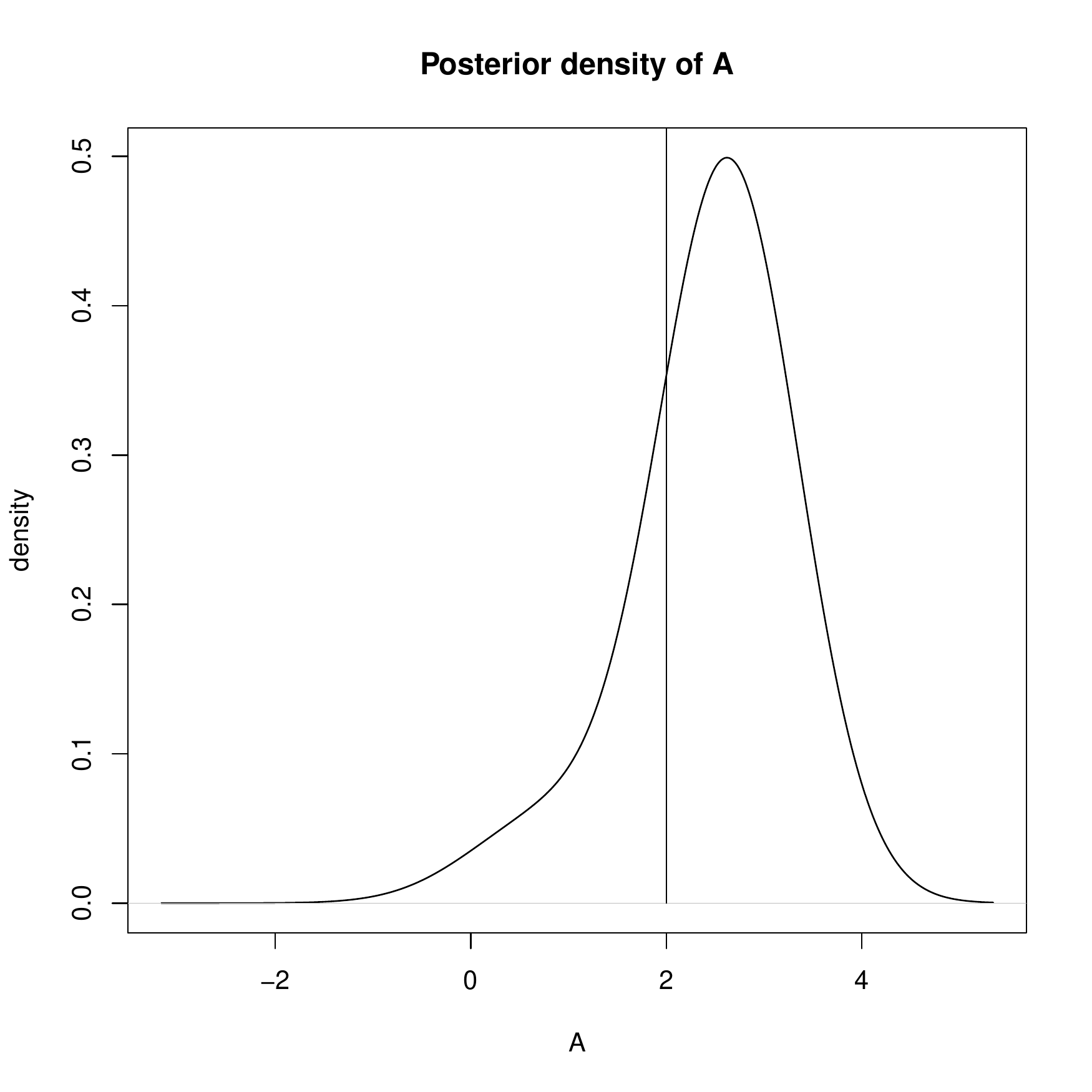}
\includegraphics[height=1.5in,width=1.5in]{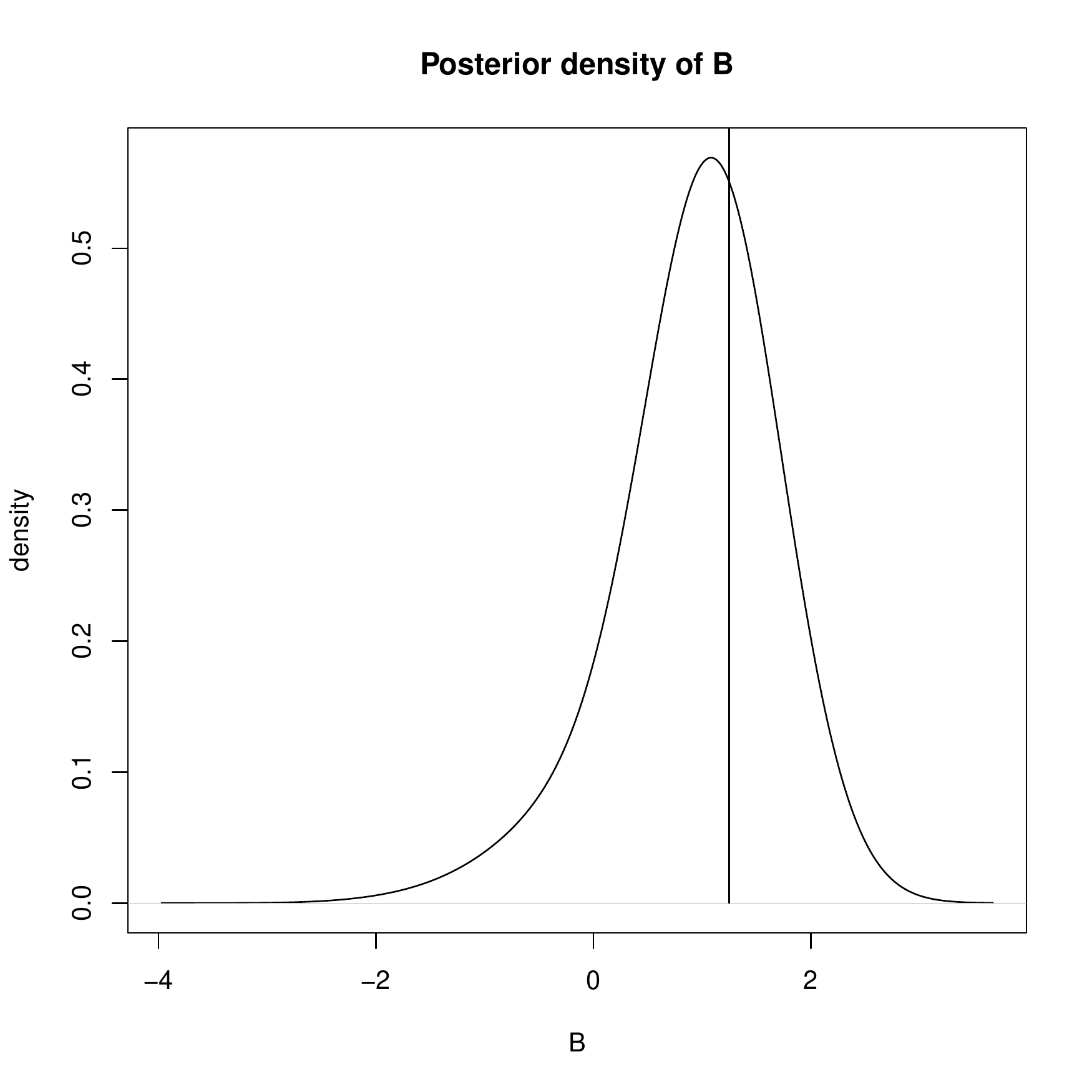}
\includegraphics[height=1.5in,width=1.5in]{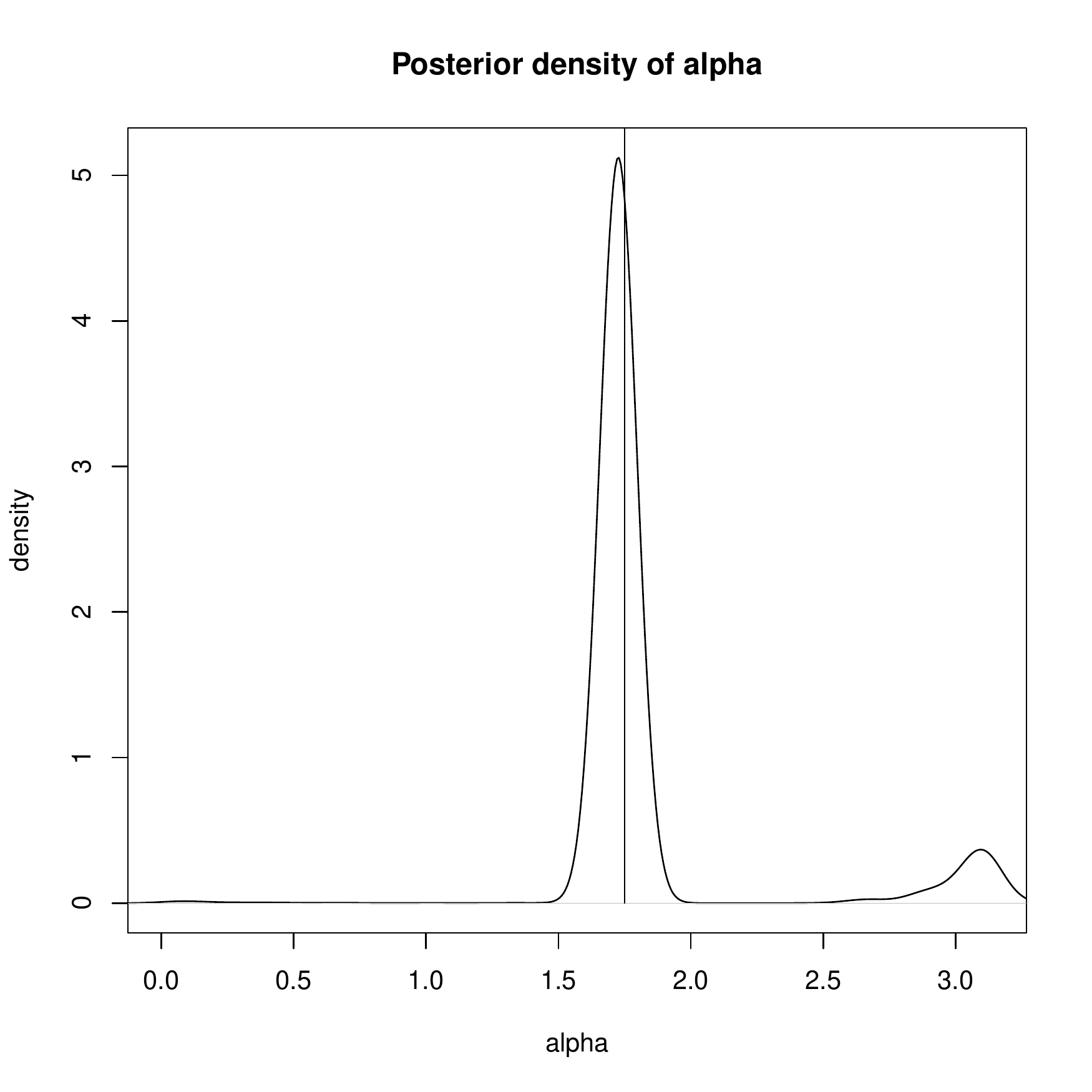}
\includegraphics[height=1.5in,width=1.5in]{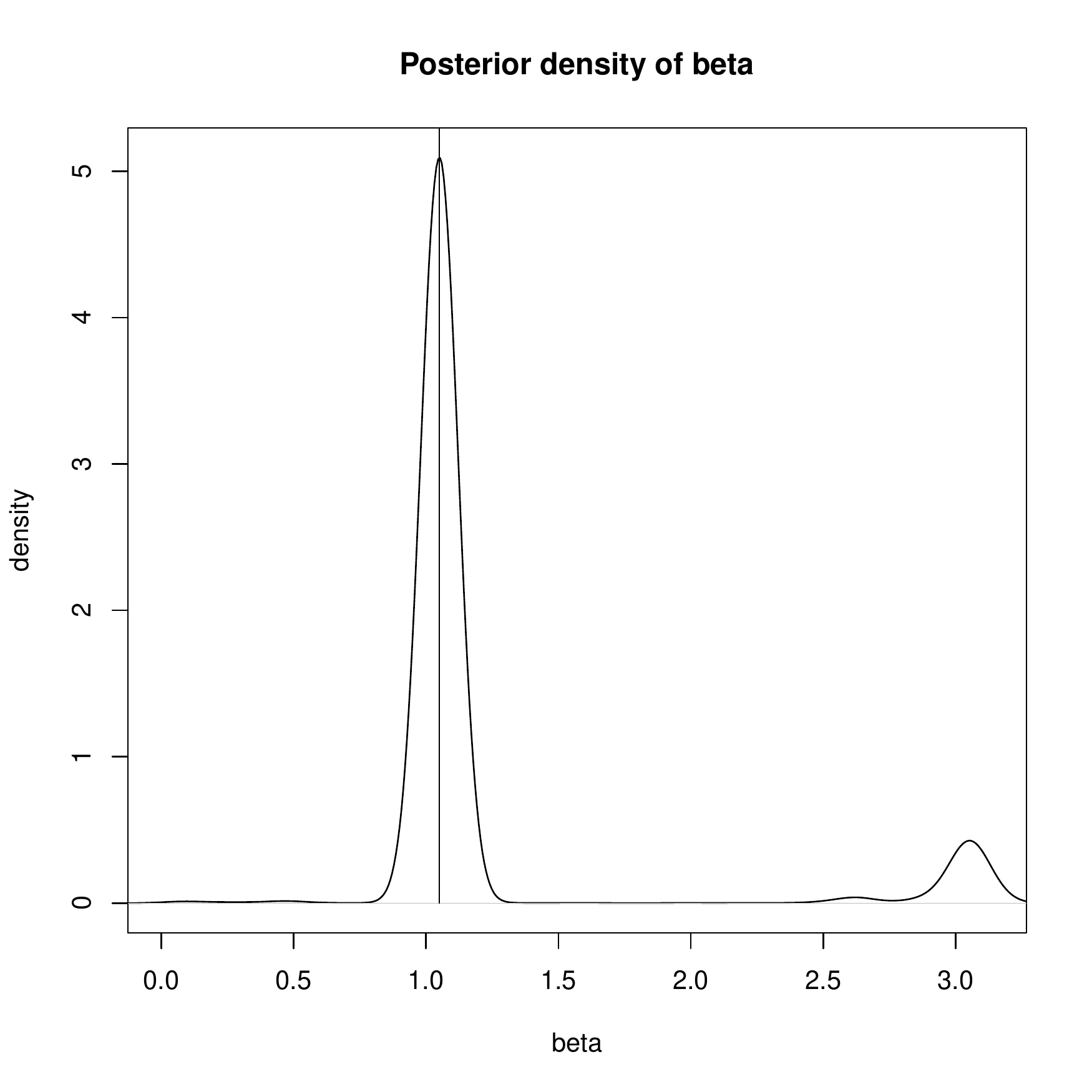}
\includegraphics[height=1.5in,width=1.5in]{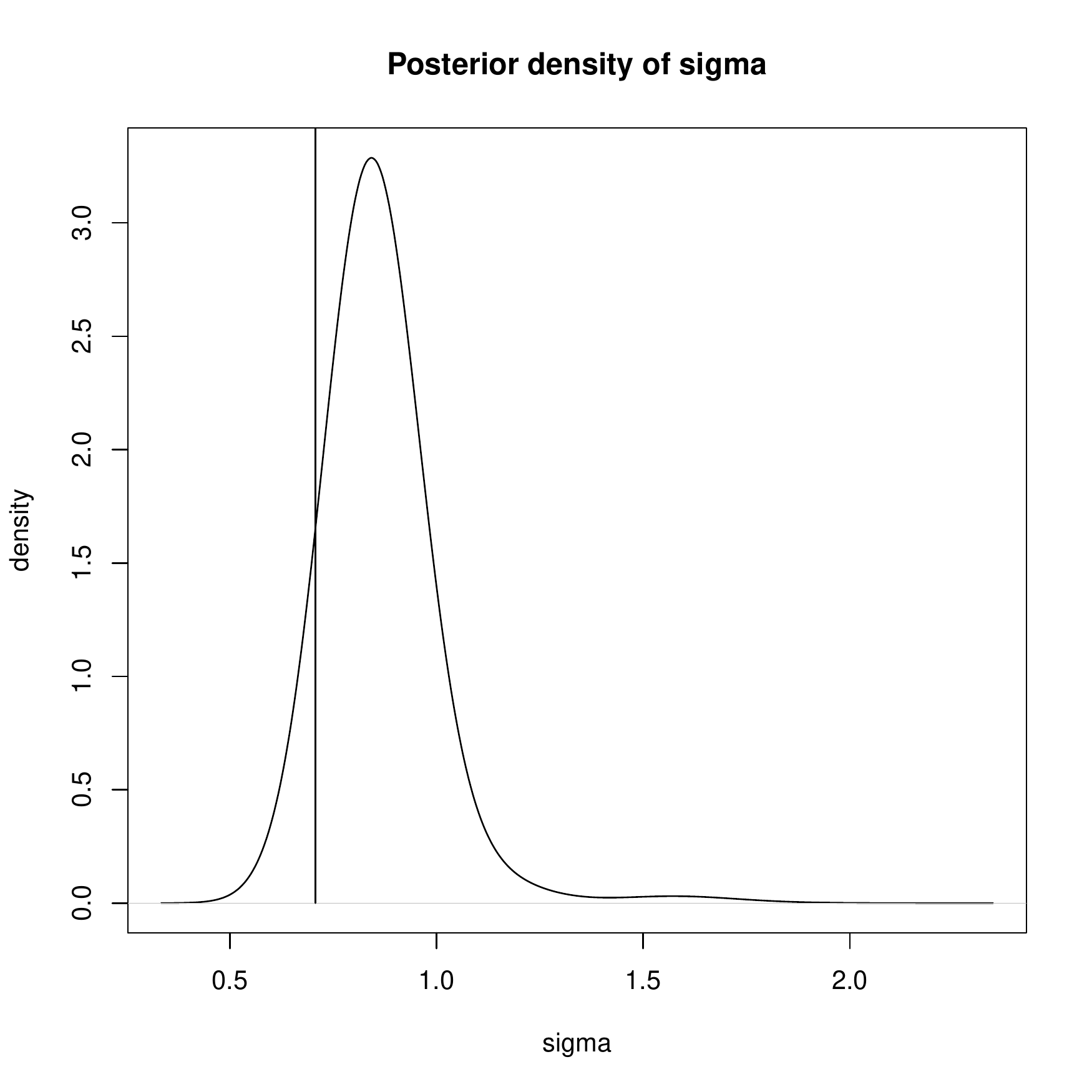}
\caption{Posterior densities of $A$, $B$, $\alpha$, $\beta$ and $\sigma$ for sample 1 of Table 1, where true values are indicated with vertical lines.}
\label{Fig:Post of A,B,alpha,beta, sigma for sample1}
\end{figure}

\begin{figure}[htp]
\centering
\includegraphics[height=1.5in,width=1.5in]{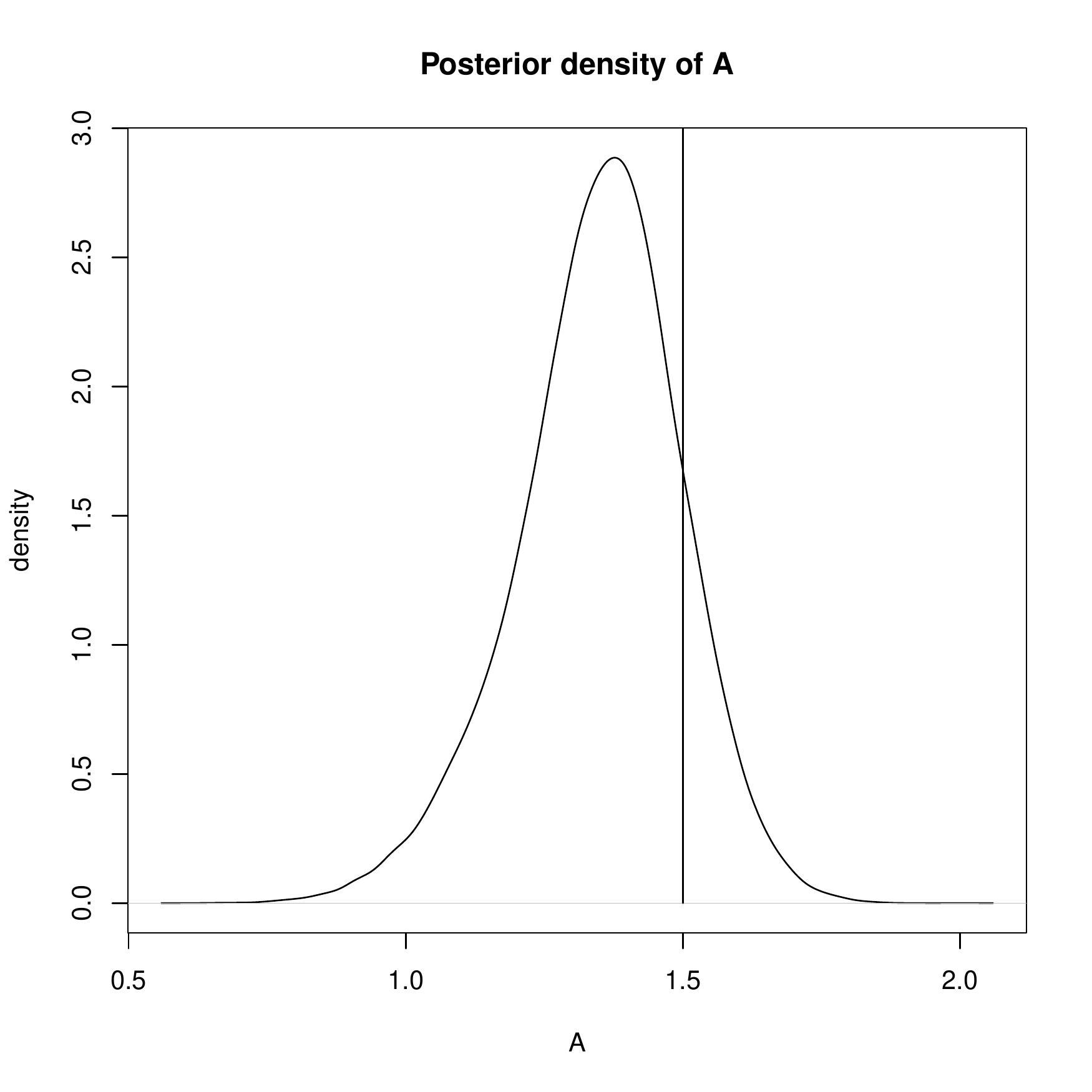}
\includegraphics[height=1.5in,width=1.5in]{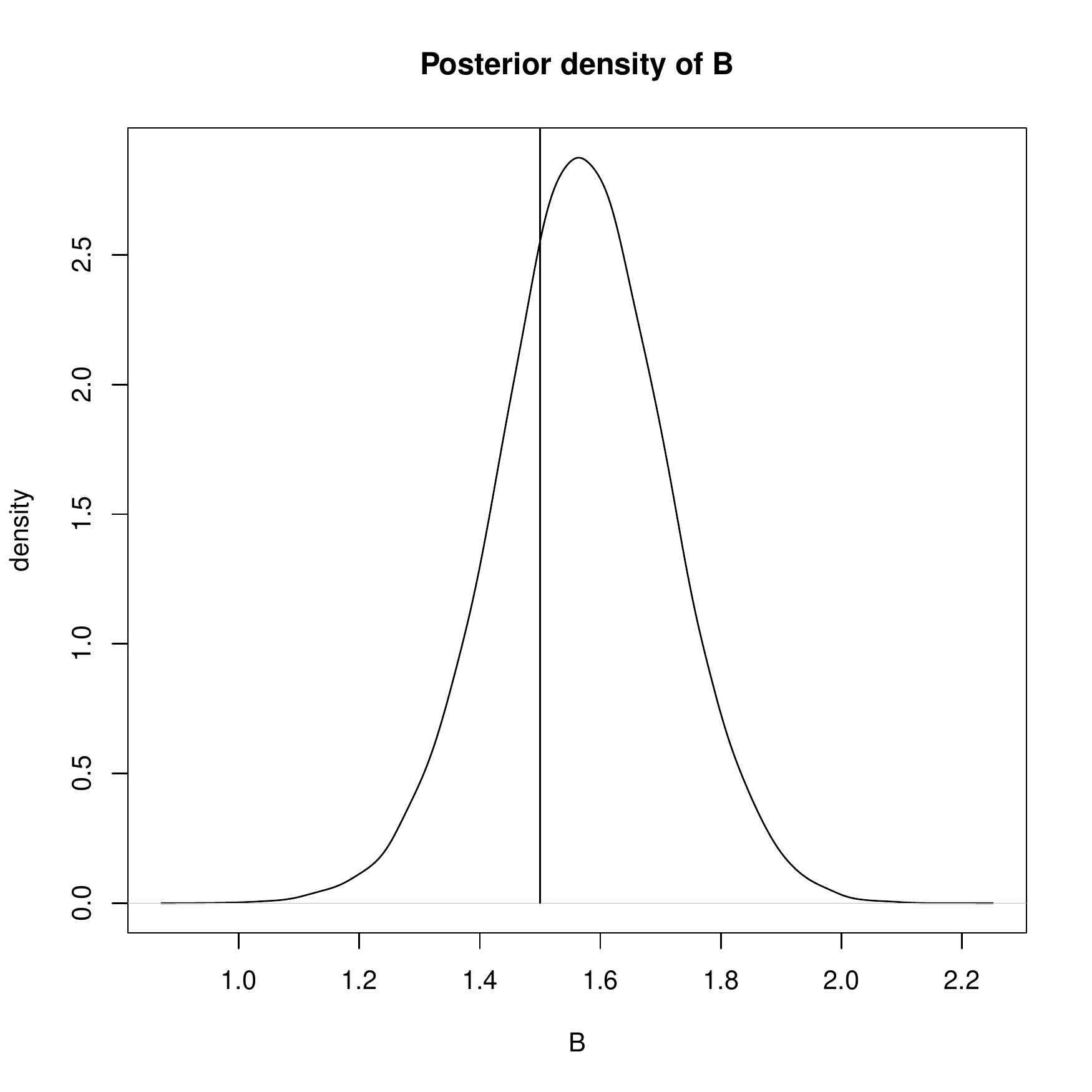}
\includegraphics[height=1.5in,width=1.5in]{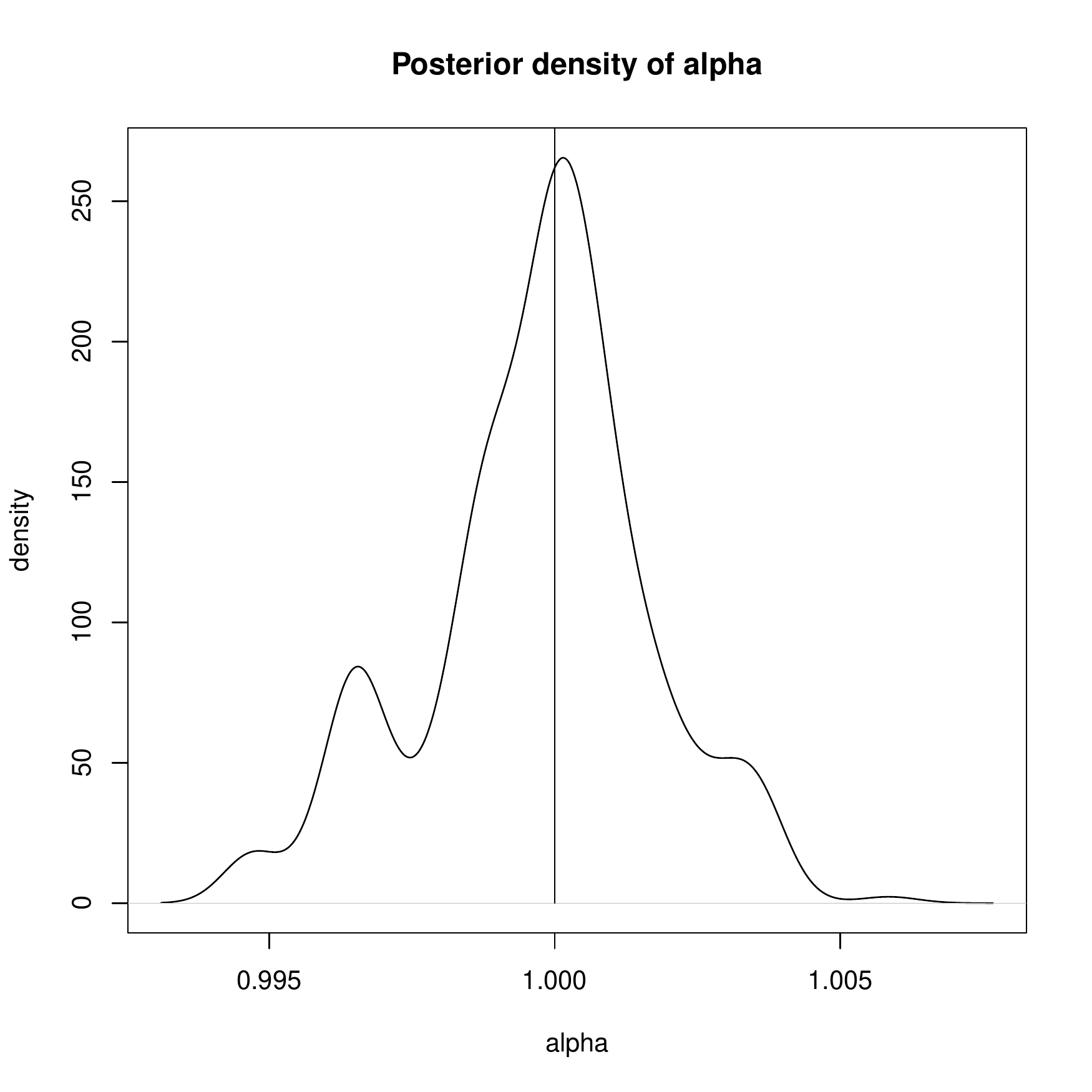}
\includegraphics[height=1.5in,width=1.5in]{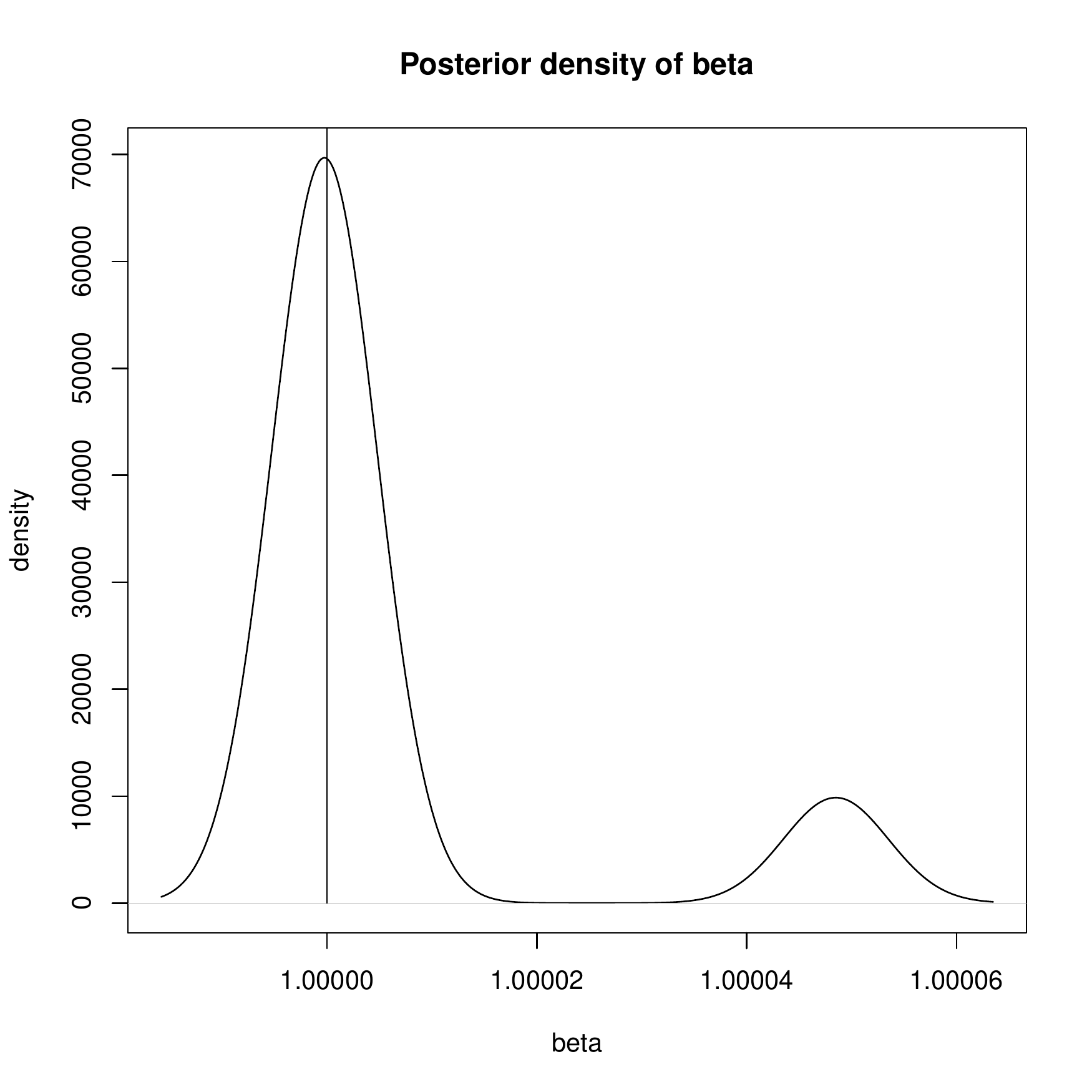}
\includegraphics[height=1.5in,width=1.5in]{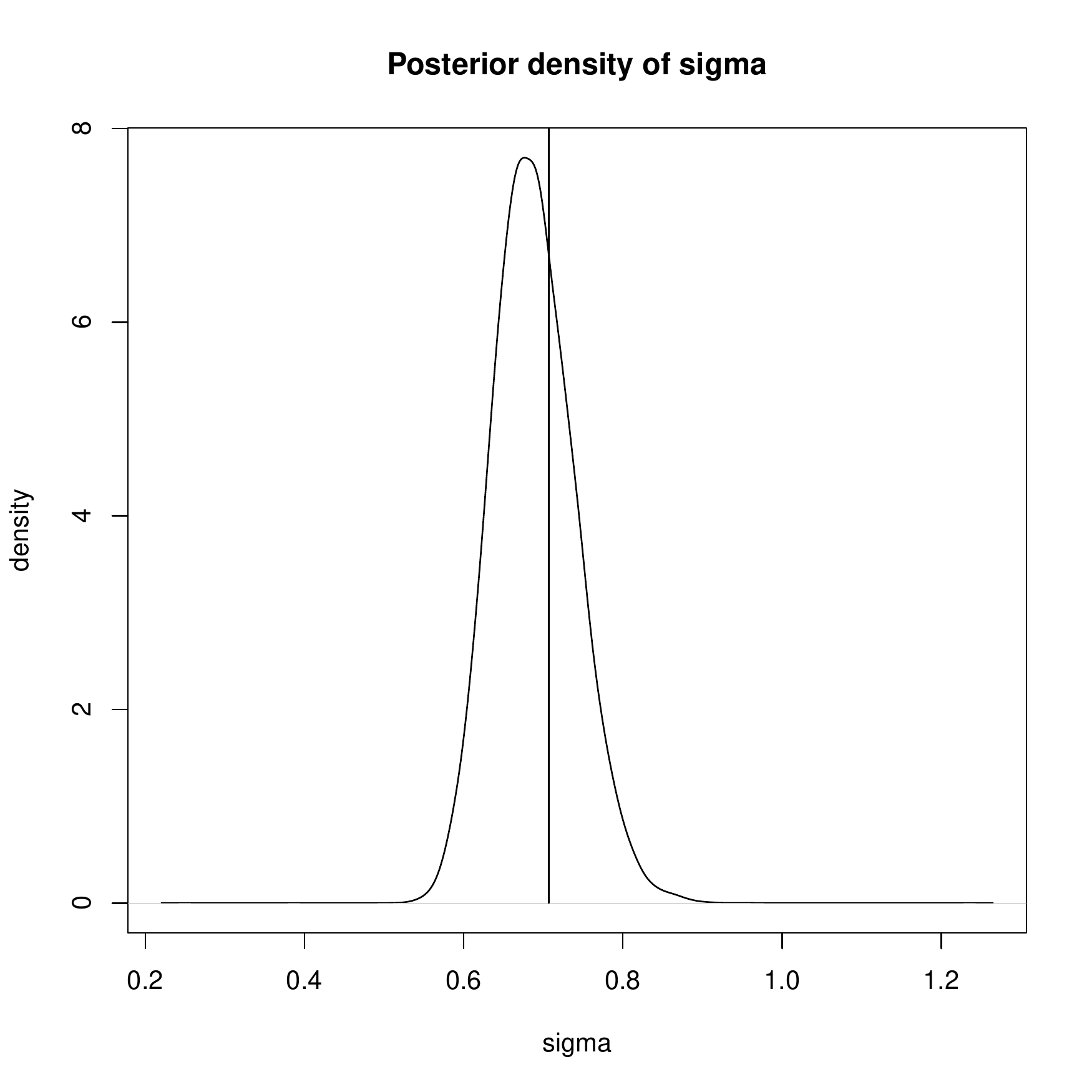}
\caption{Posterior densities of $A$, $B$, $\alpha$, $\beta$ and $\sigma$ for sample 2 of Table 1, where true values are indicated with vertical lines.}
\label{Fig:Post of A,B,alpha,beta, sigma for sample2}
\end{figure}

\begin{figure}[htp]
\centering
\includegraphics[height=1.5in,width=1.5in]{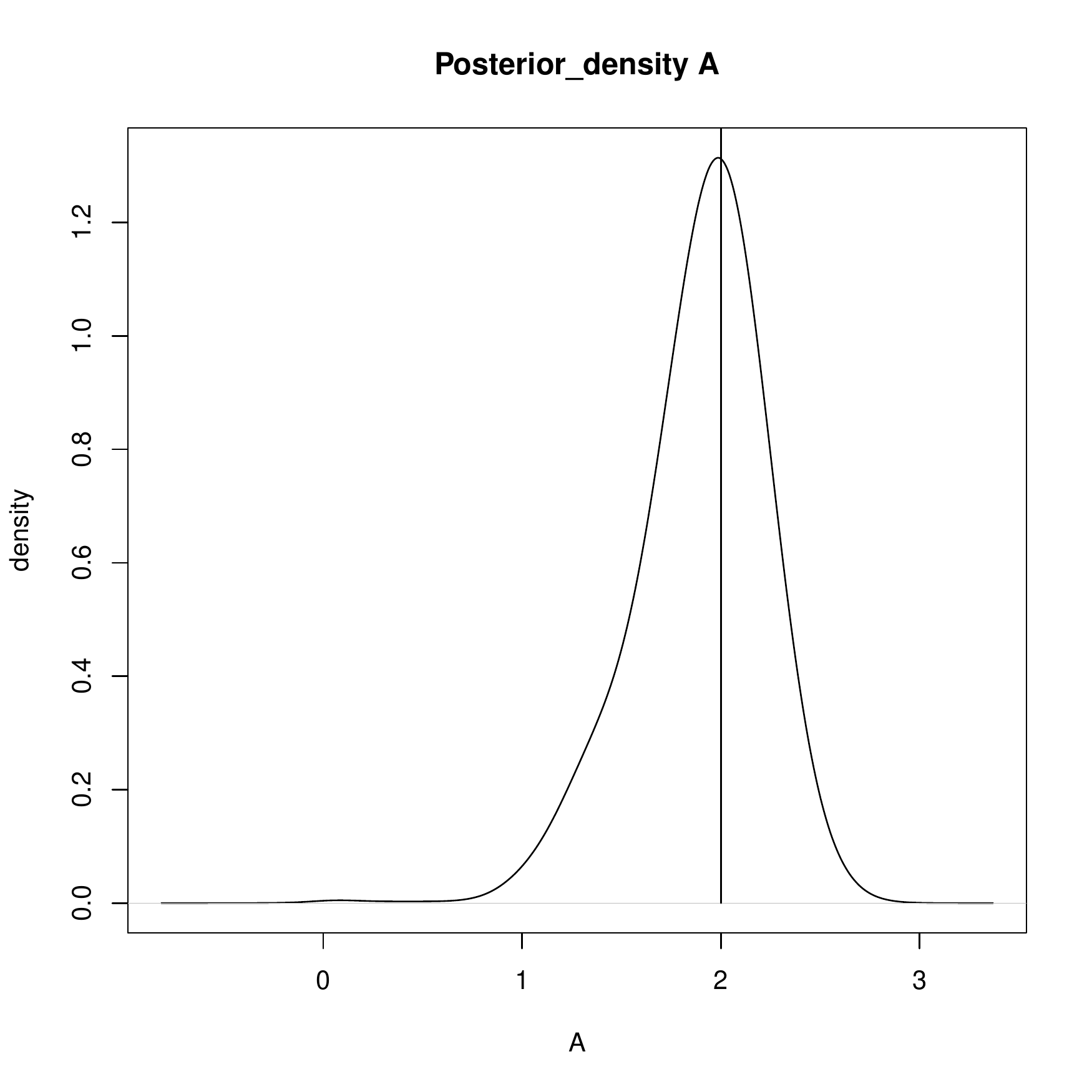}
\includegraphics[height=1.5in,width=1.5in]{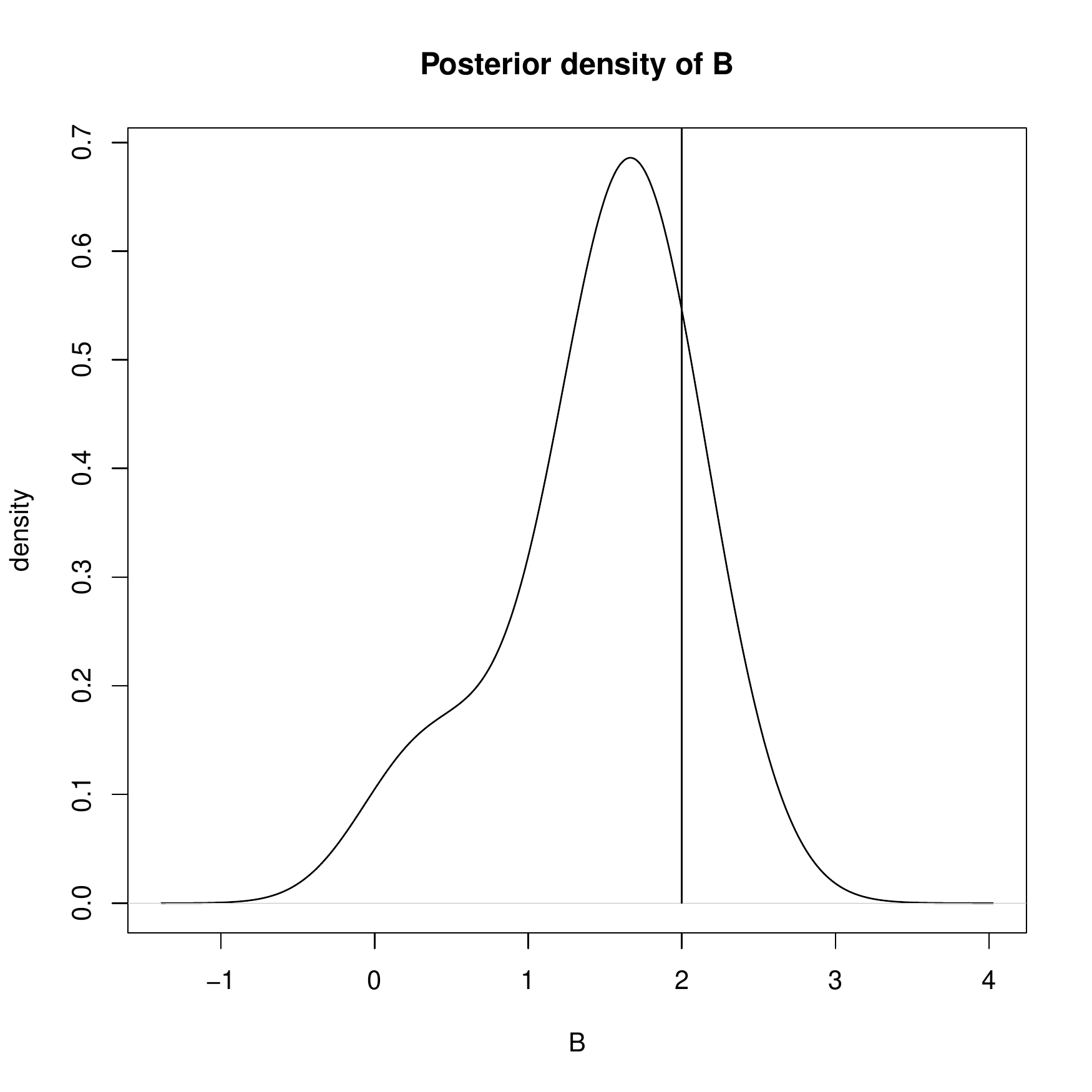}
\includegraphics[height=1.5in,width=1.5in]{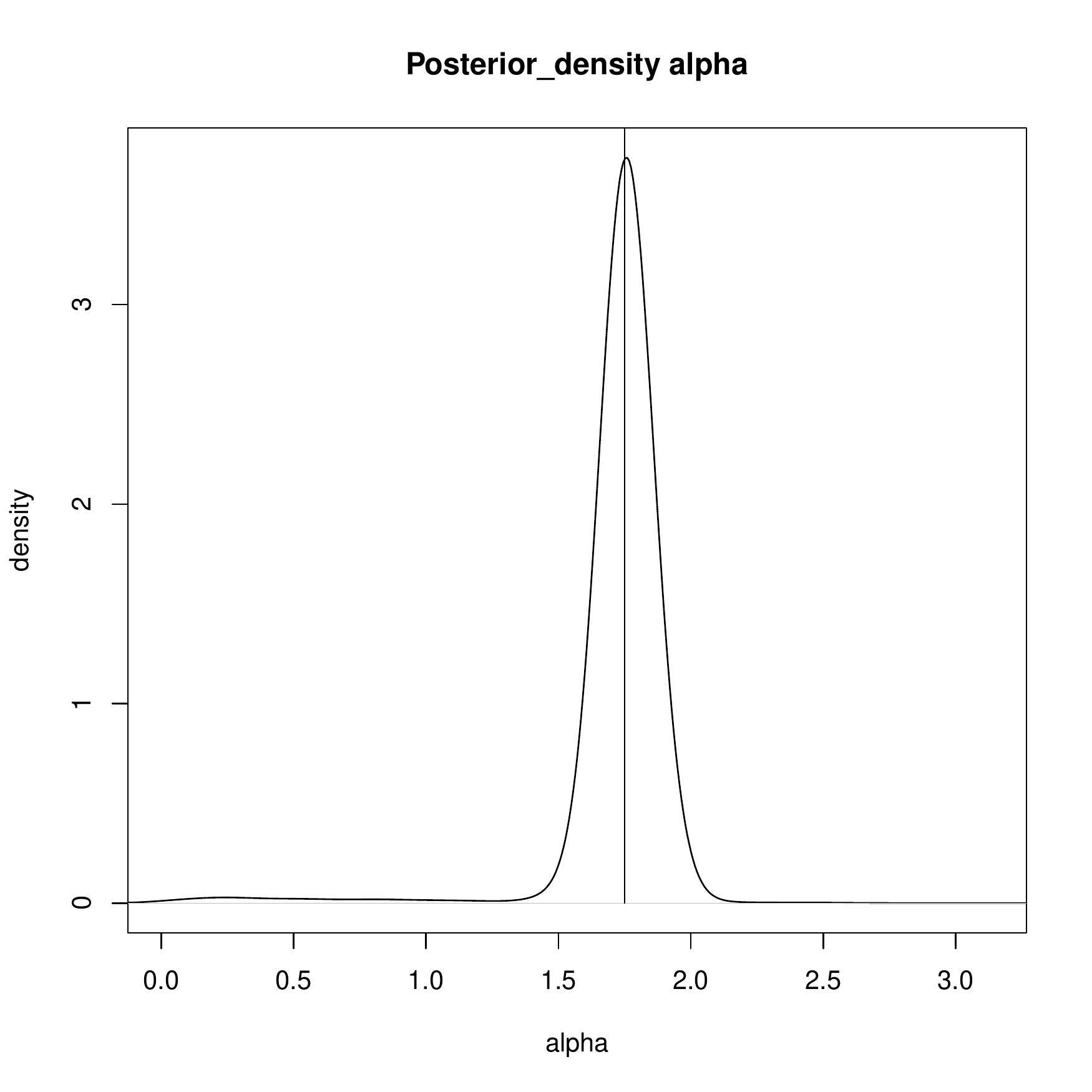}
\includegraphics[height=1.5in,width=1.5in]{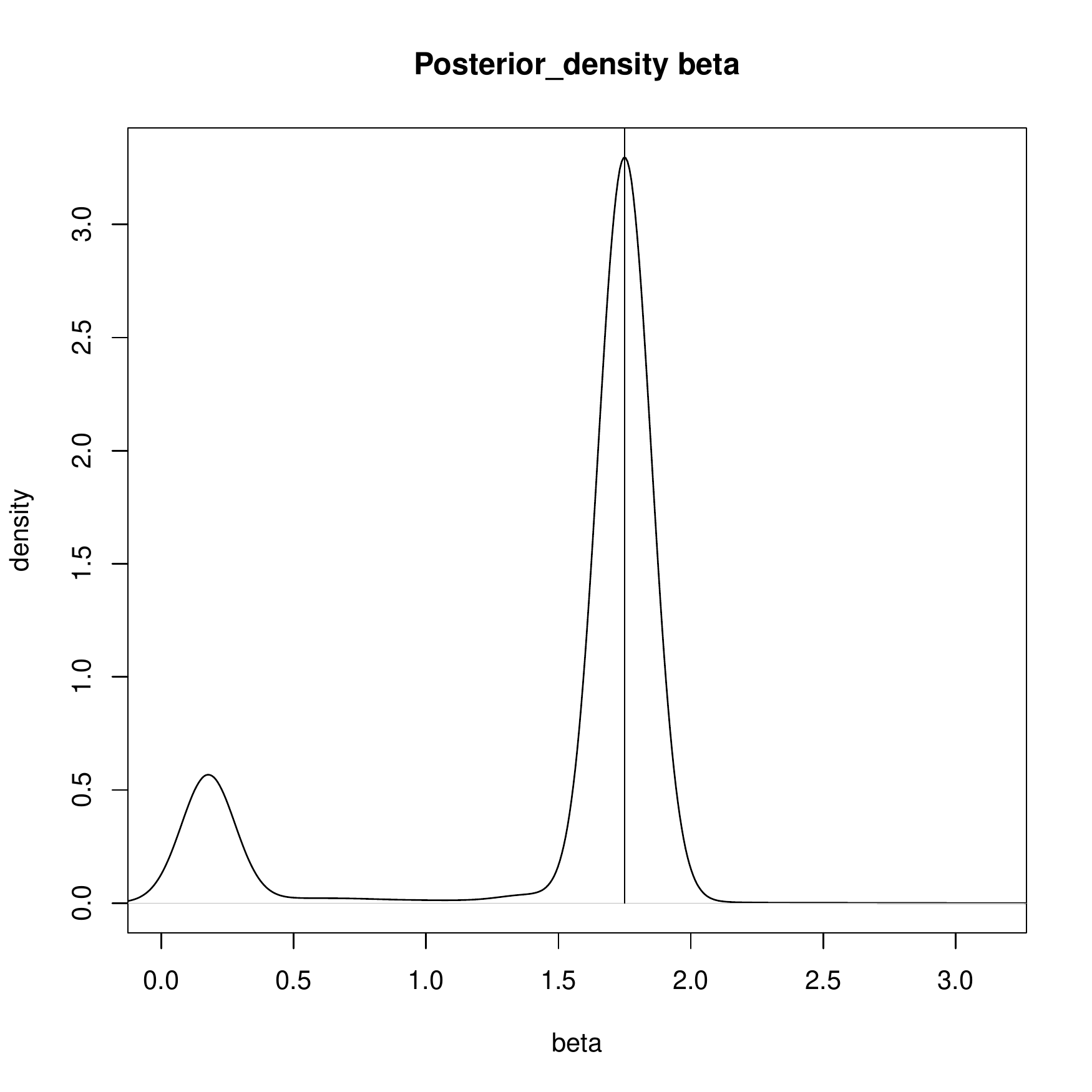}
\includegraphics[height=1.5in,width=1.5in]{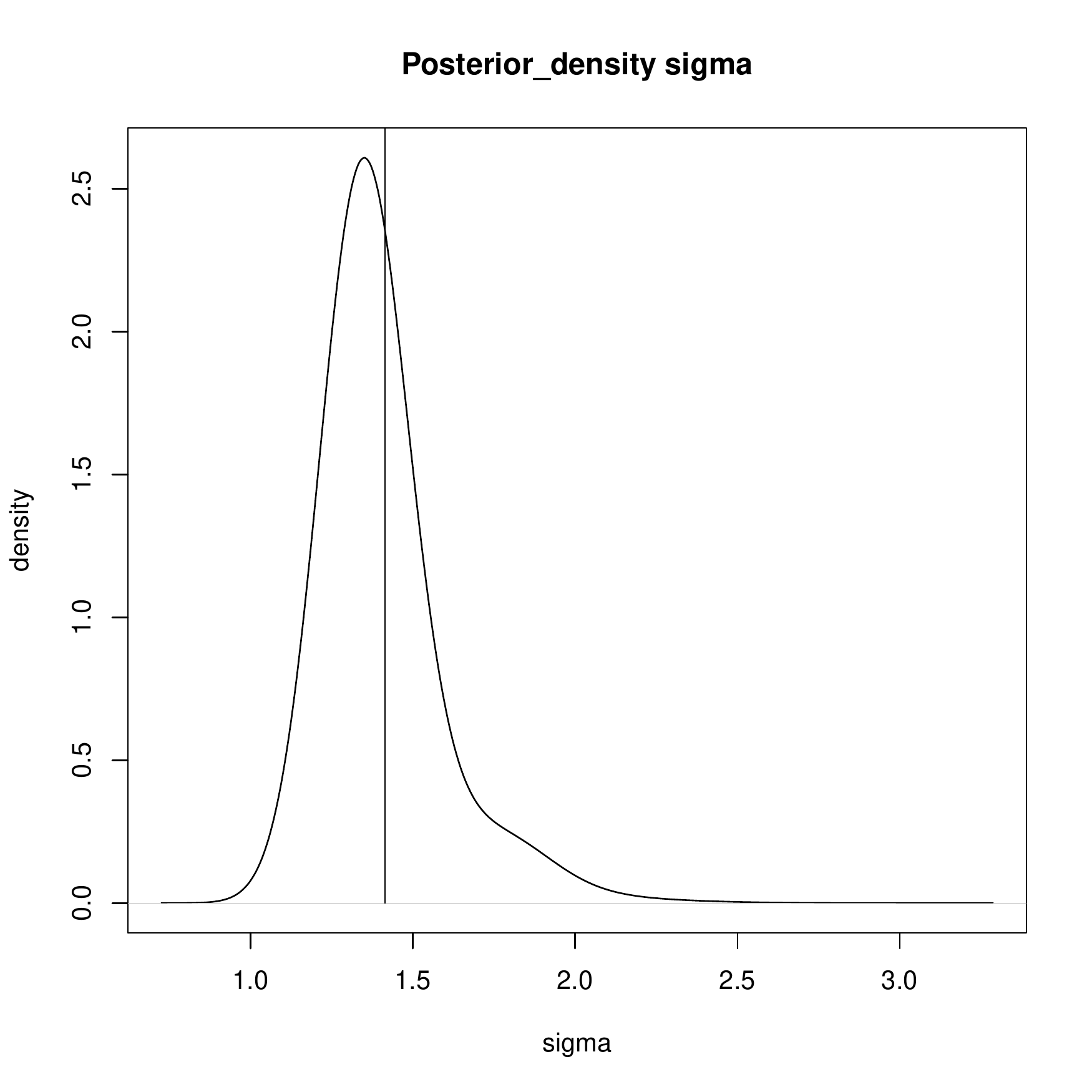}
\caption{Posterior densities of $A$, $B$, $\alpha$, $\beta$ and $\sigma$ for sample 3 of Table 1, where true values are indicated with vertical lines.}
\label{Fig:Post of A,B,alpha,beta, for sample3}
\end{figure}

\begin{figure}[htp]
\centering
\includegraphics[height=1.5in,width=1.5in]{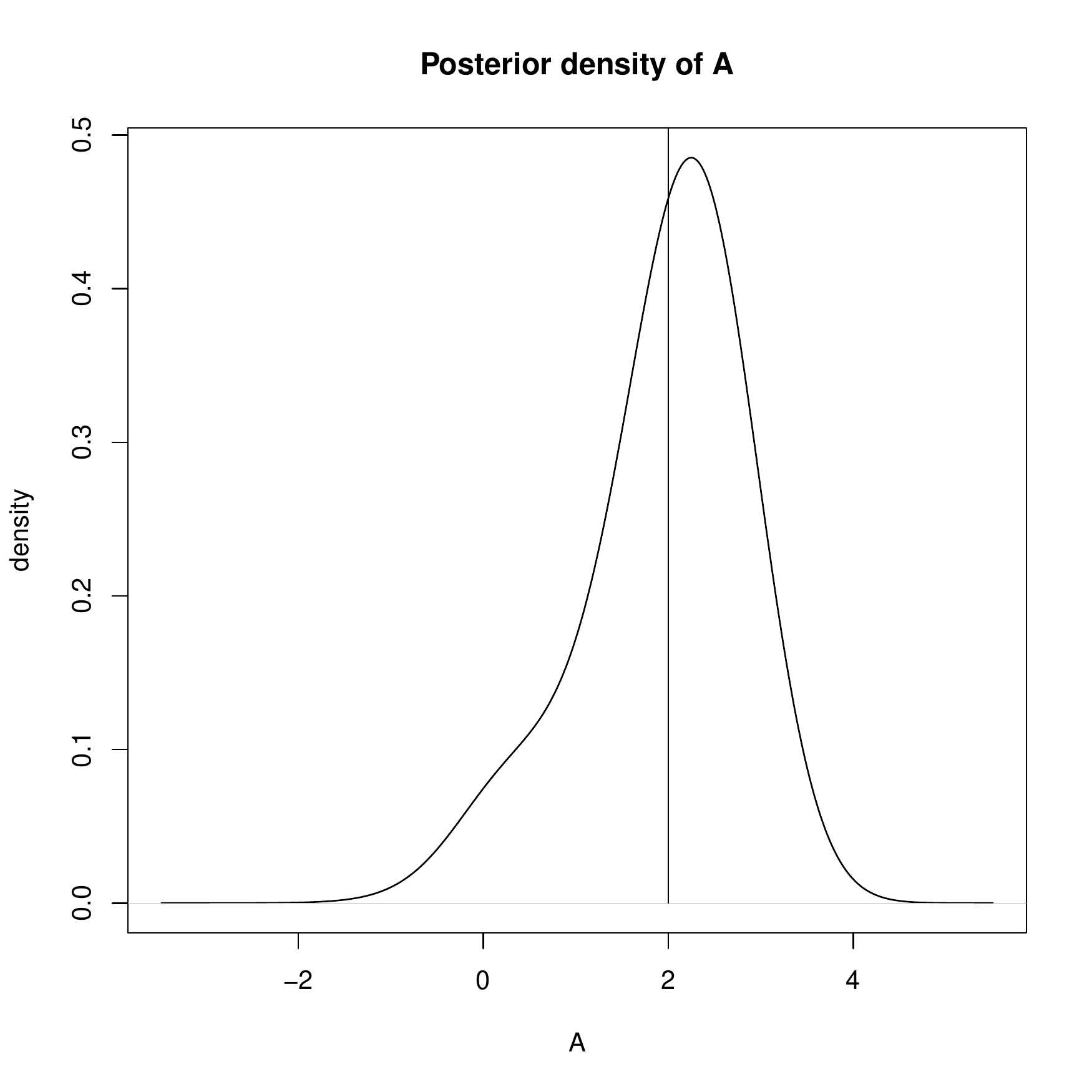}
\includegraphics[height=1.5in,width=1.5in]{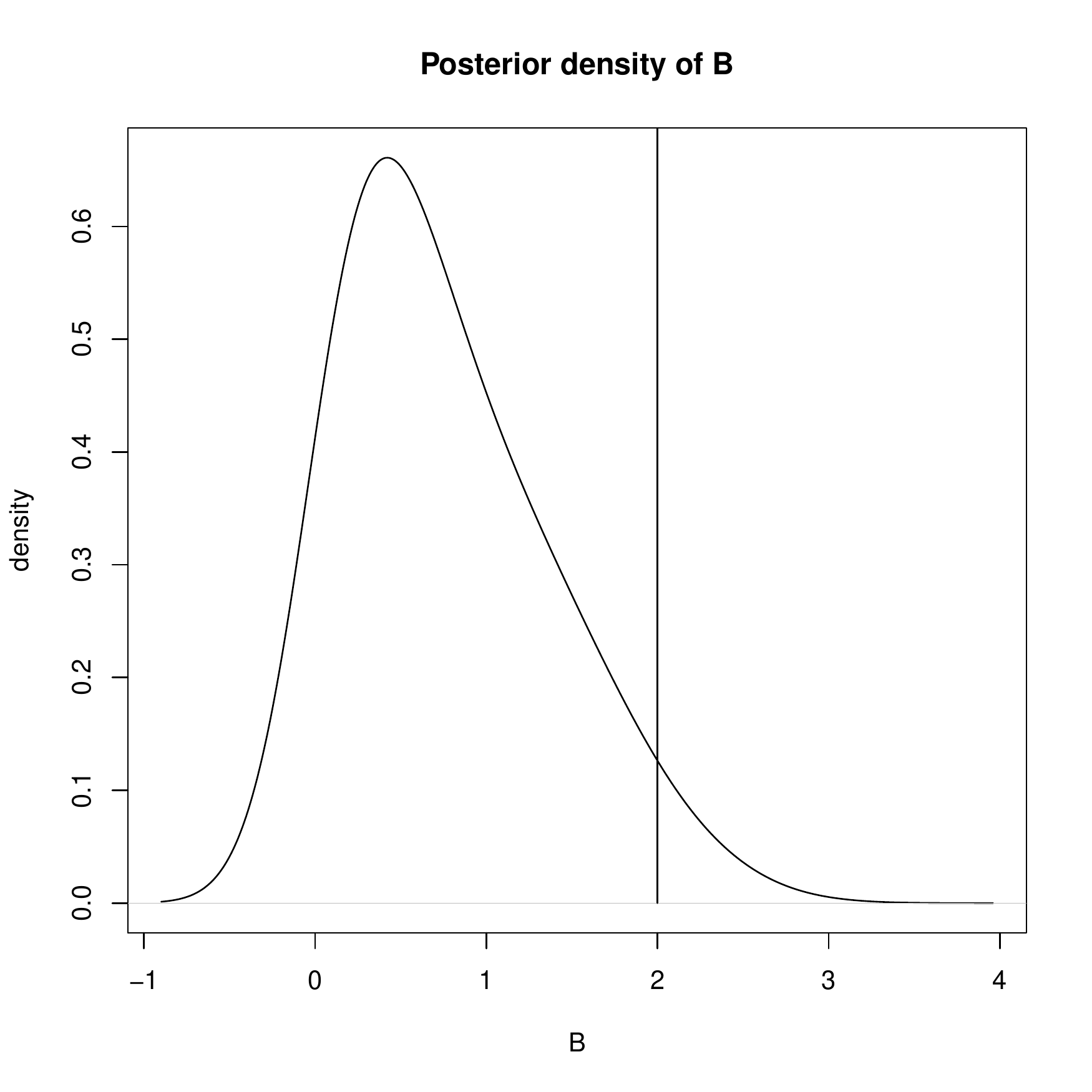}
\includegraphics[height=1.5in,width=1.5in]{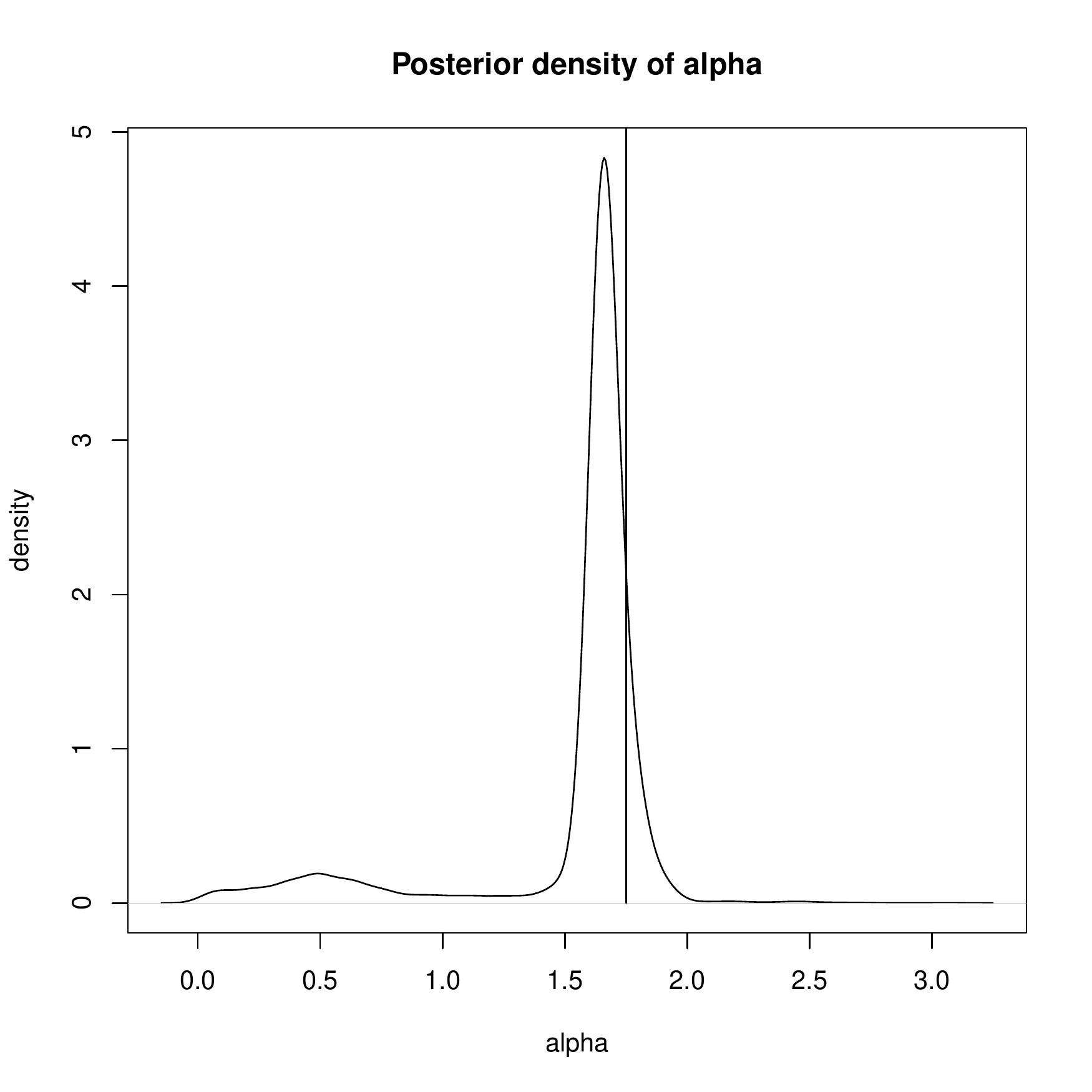}
\includegraphics[height=1.5in,width=1.5in]{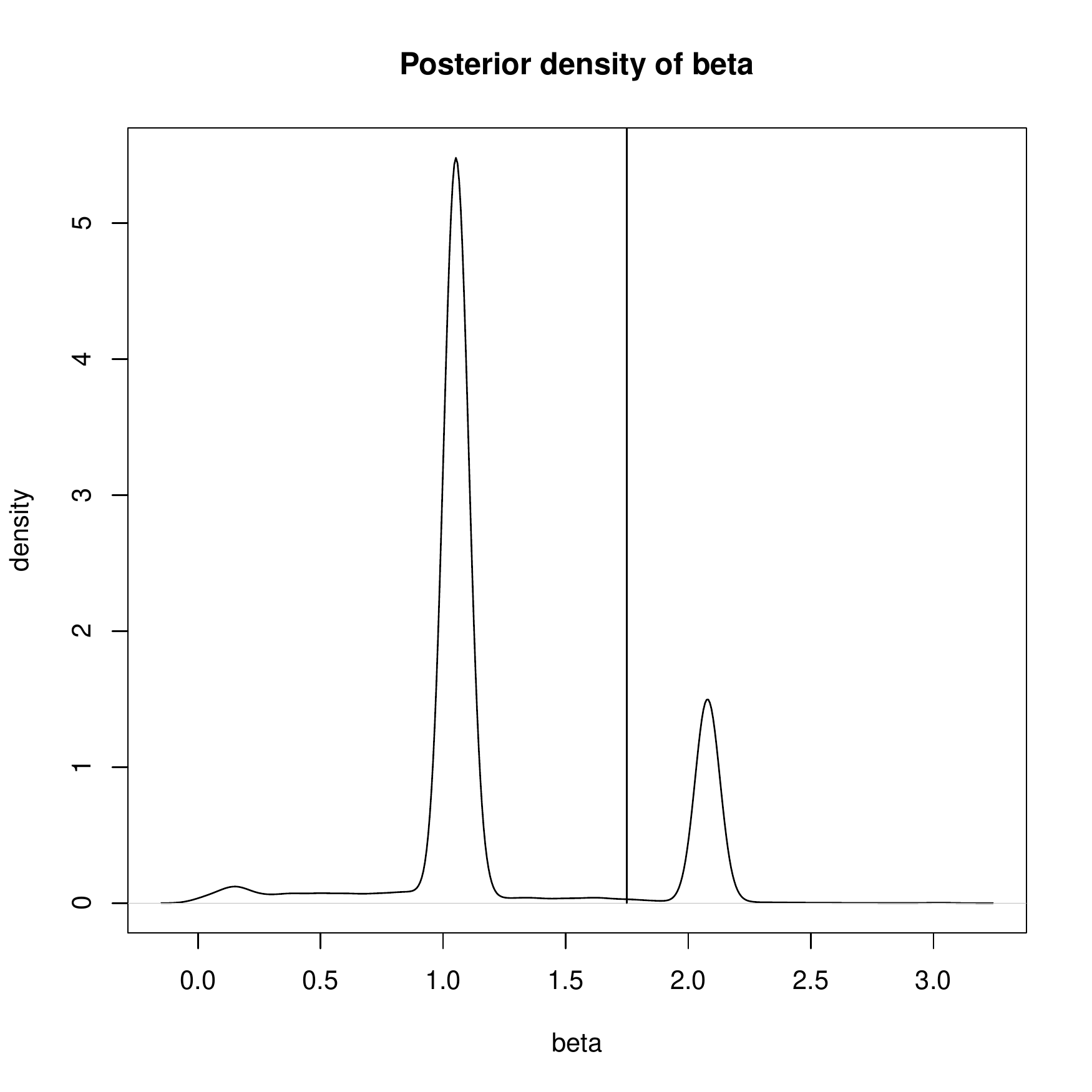}
\includegraphics[height=1.5in,width=1.5in]{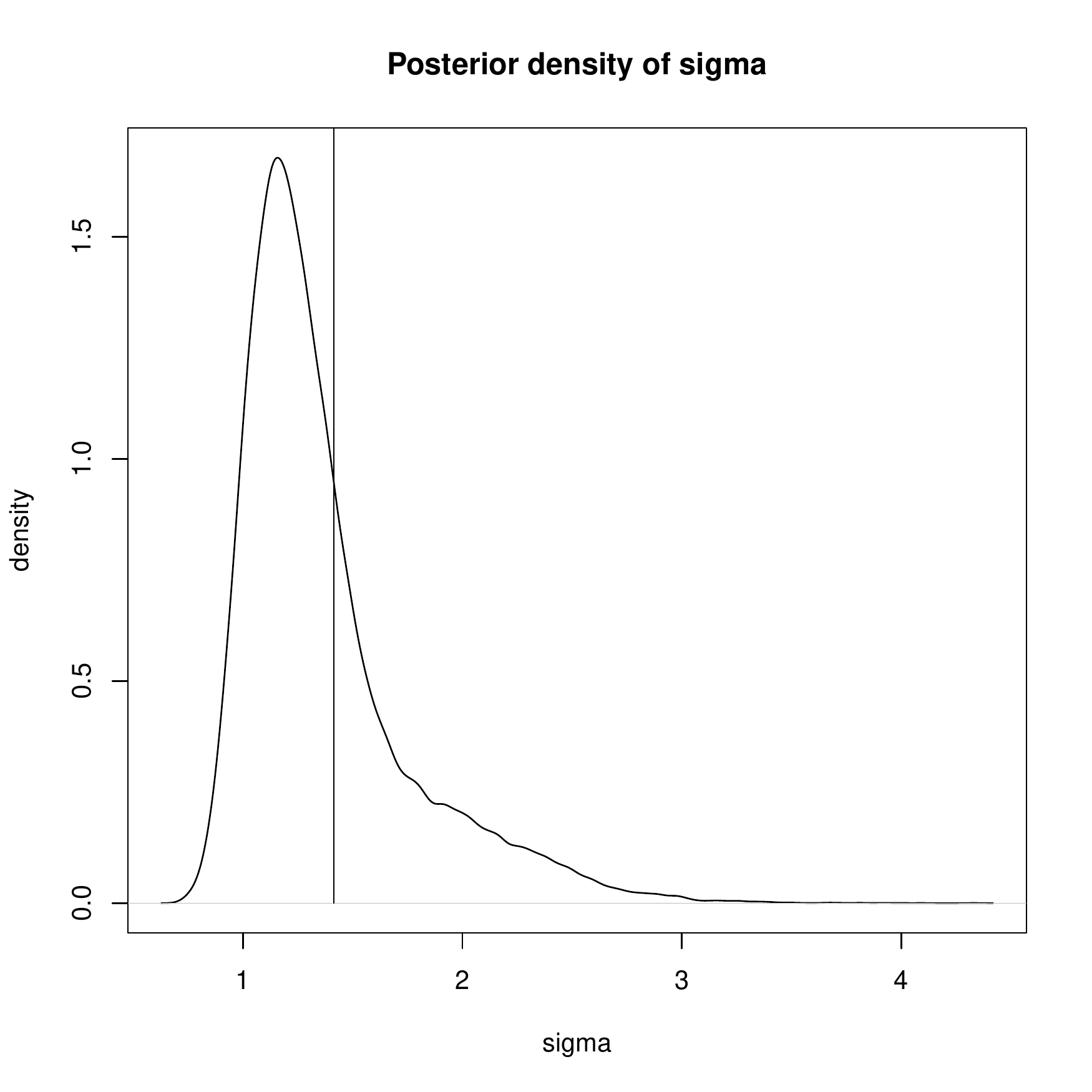}
\caption{Posterior densities of $A$, $B$, $\alpha$, $\beta$ and $\sigma$ for sample 4 of Table 1, where true values are indicated with vertical lines.}
\label{Fig:Post of A,B,alpha,beta, for sample4}
\end{figure}

\begin{figure}[htp]
\centering
\includegraphics[height=1.5in,width=1.5in]{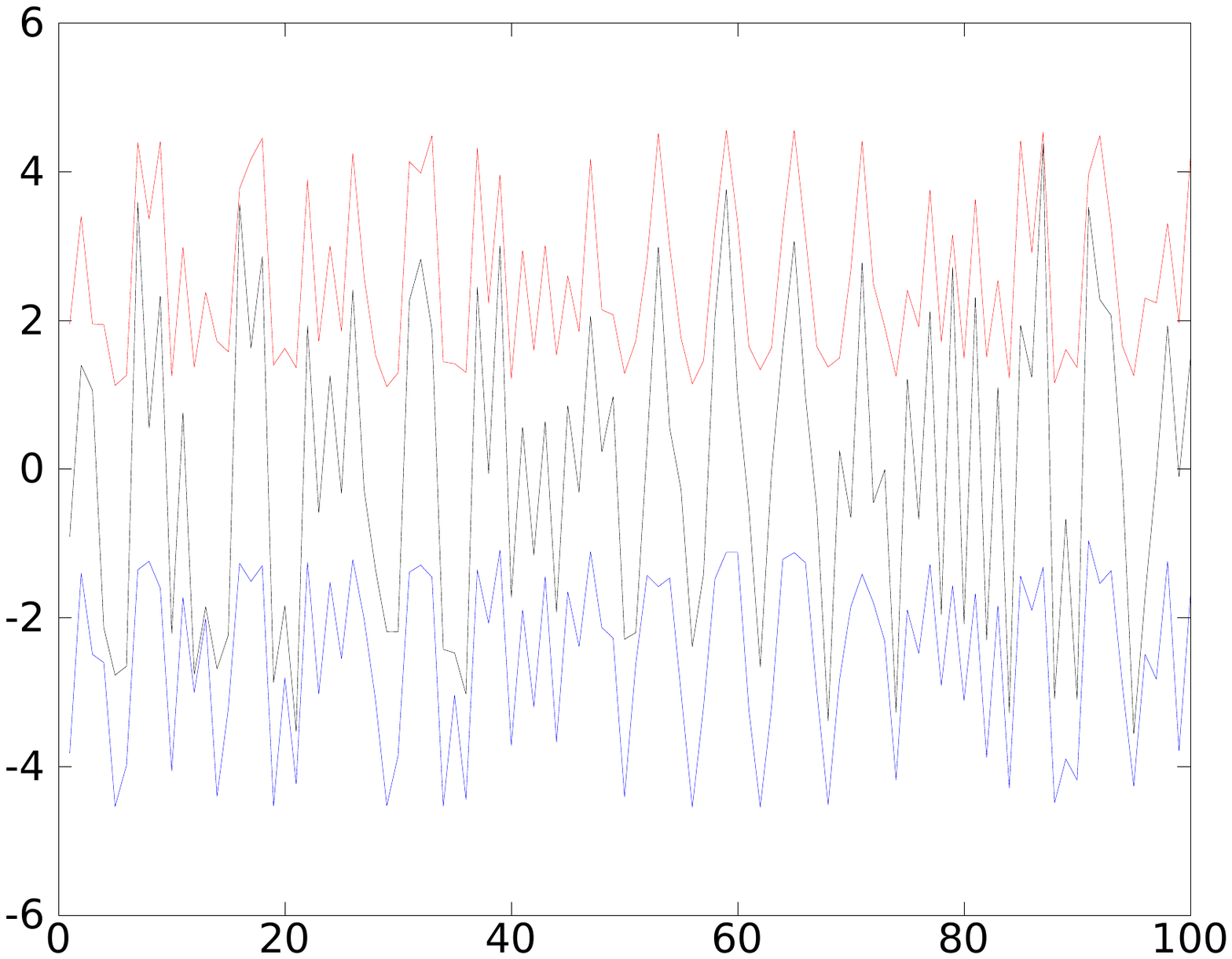}
\includegraphics[height=1.5in,width=1.5in]{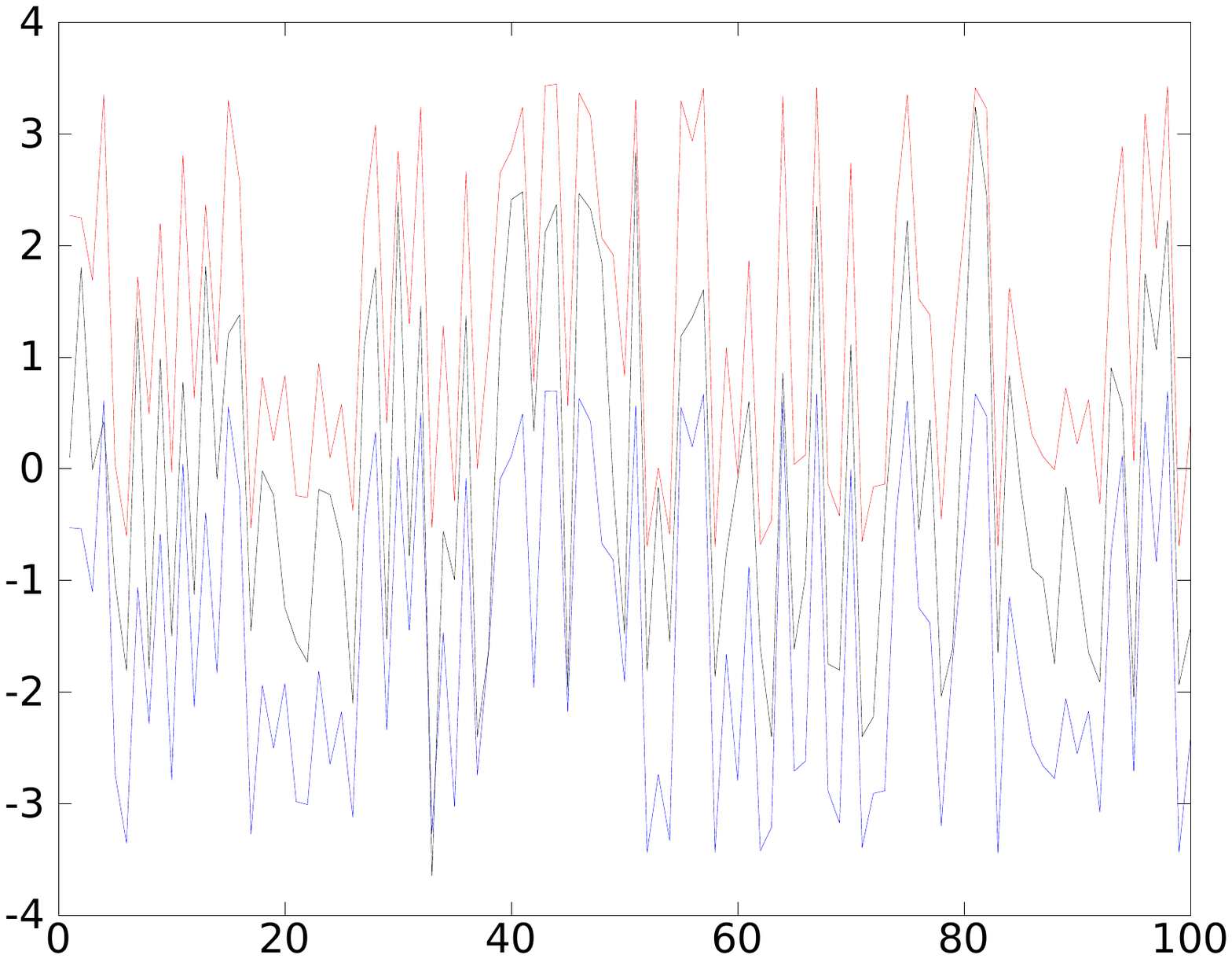}
\includegraphics[height=1.5in,width=1.5in]{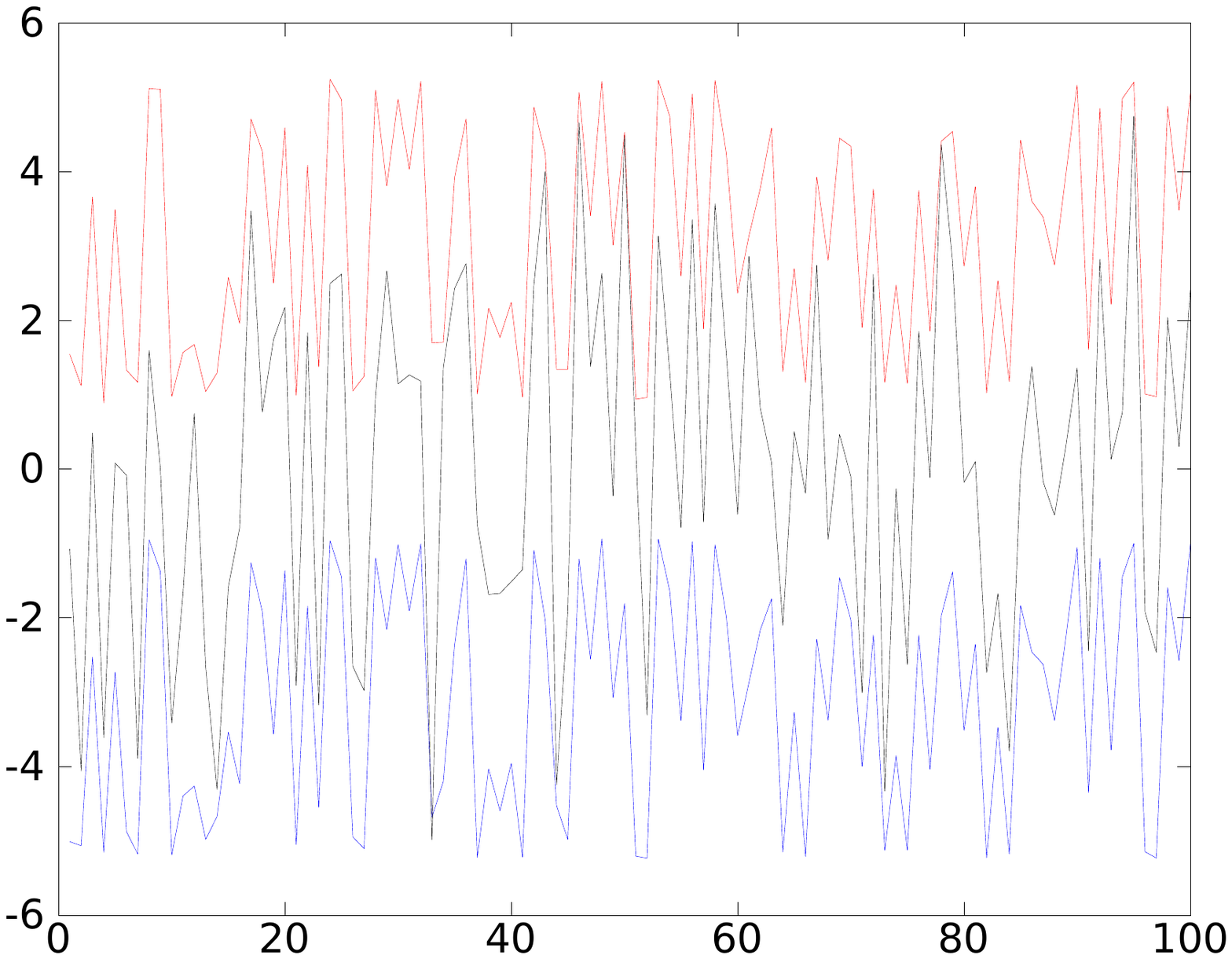}
\includegraphics[height=1.5in,width=1.5in]{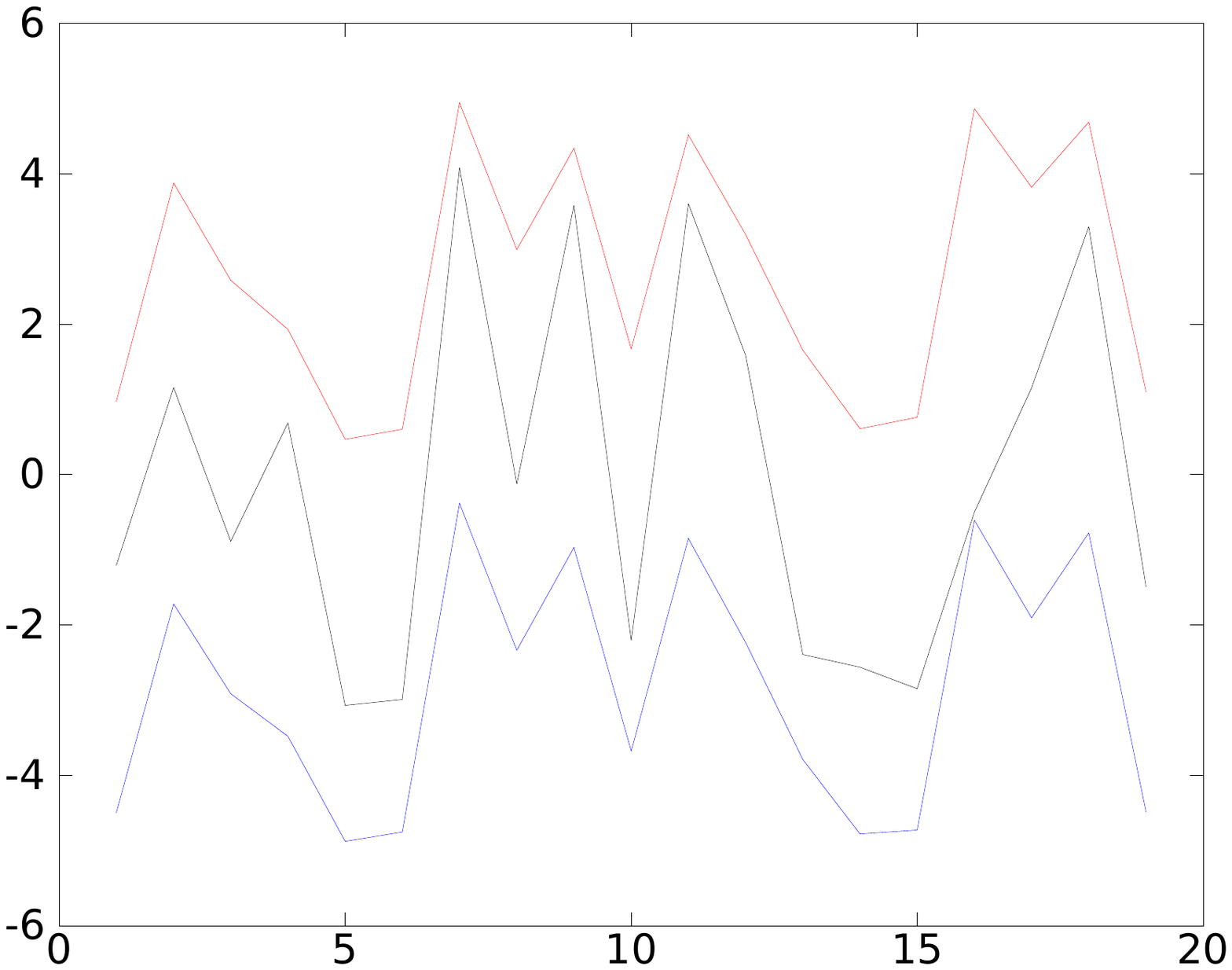}
\caption{95\% credible intervals for observed signal for sample 1, 2, 3 and 4 of table 1, where black line shows the observed signal, blue lines indicate lower 2.5\% and red lines indicate upper 97.5\% signals obtained based on MCMC simulations.}
\label{Fig:Fit for sample1,2,3,4}
\end{figure}

\begin{figure}[htp]
\centering
\includegraphics[height=1.5in,width=1.5in]{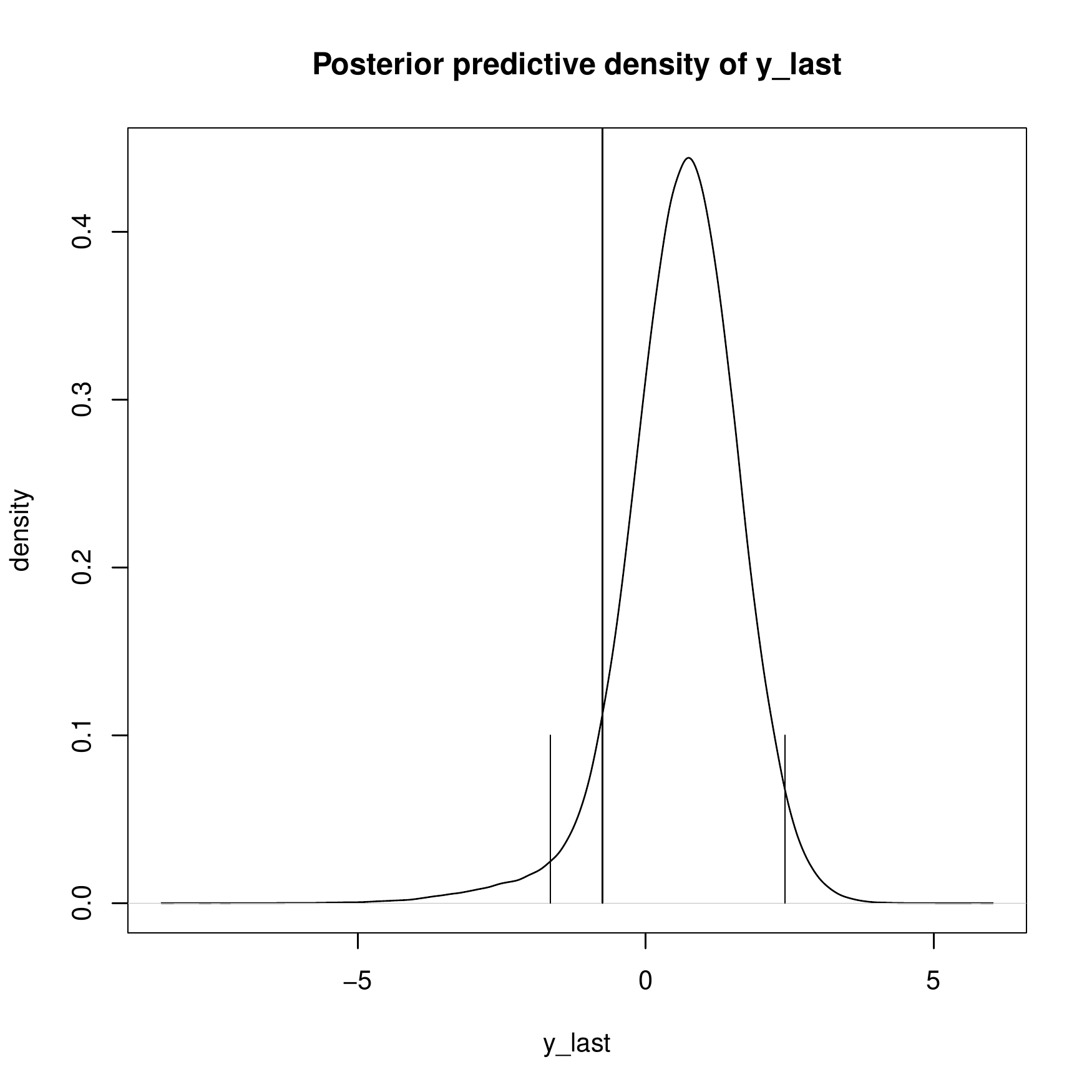}
\includegraphics[height=1.5in,width=1.5in]{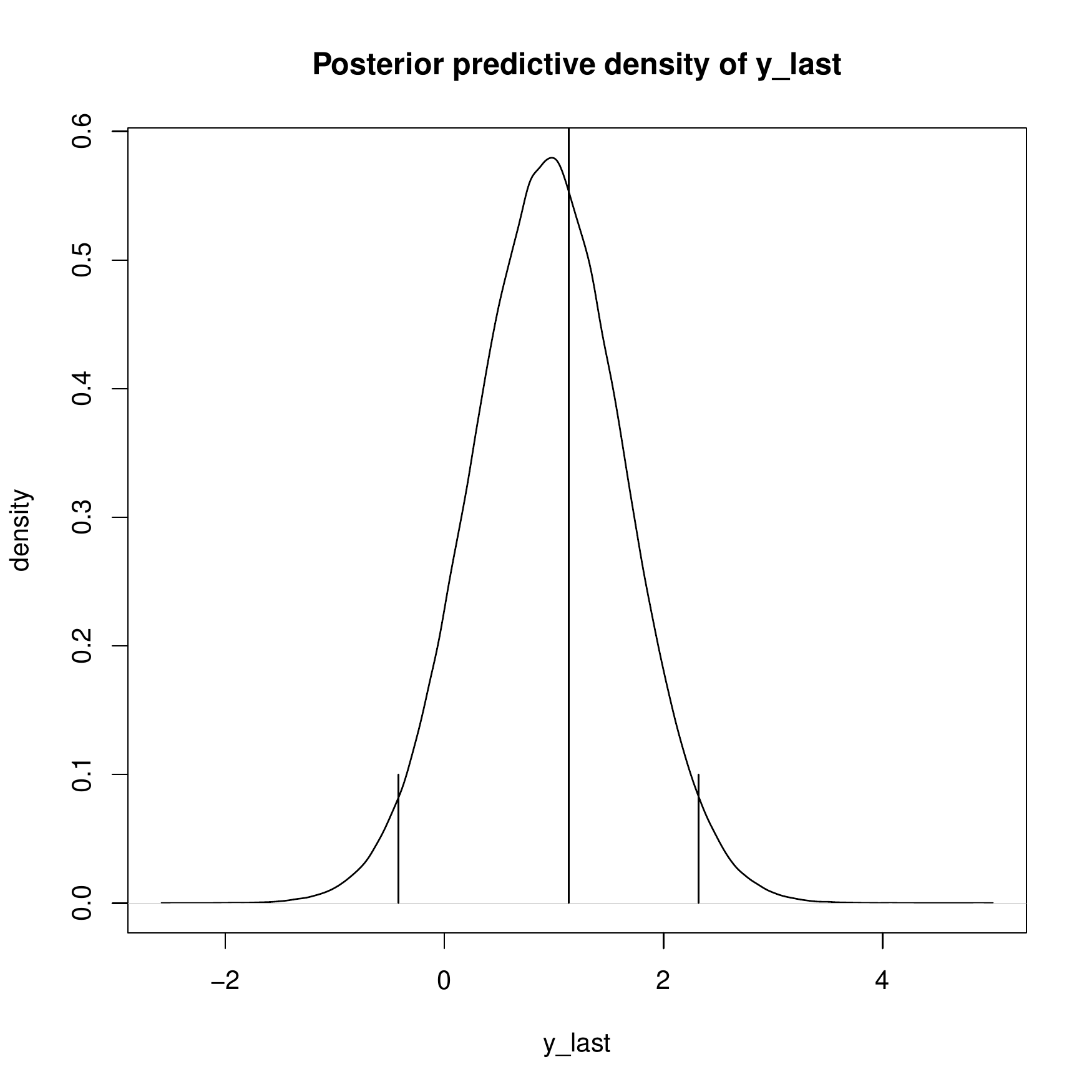}
\includegraphics[height=1.5in,width=1.5in]{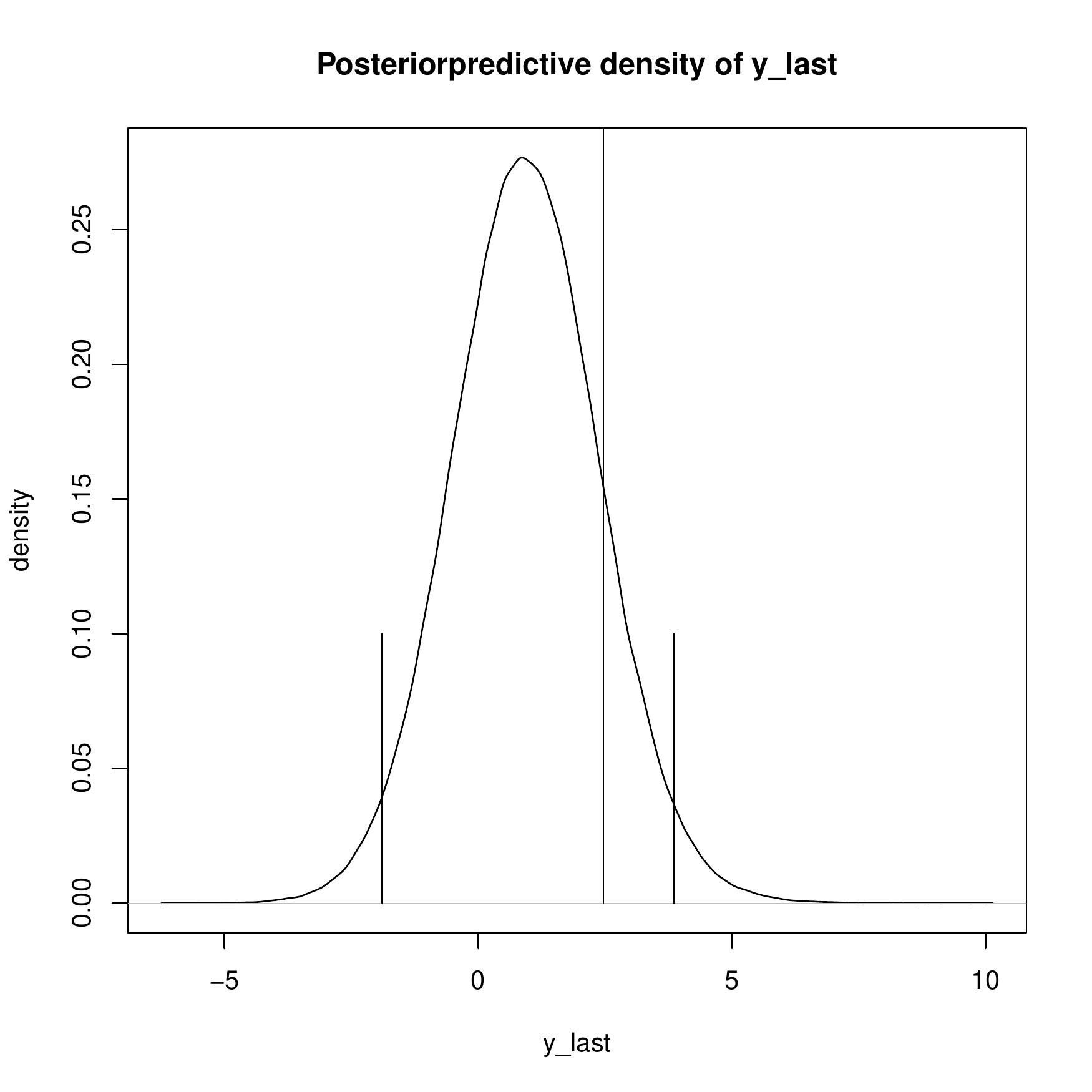}
\includegraphics[height=1.5in,width=1.5in]{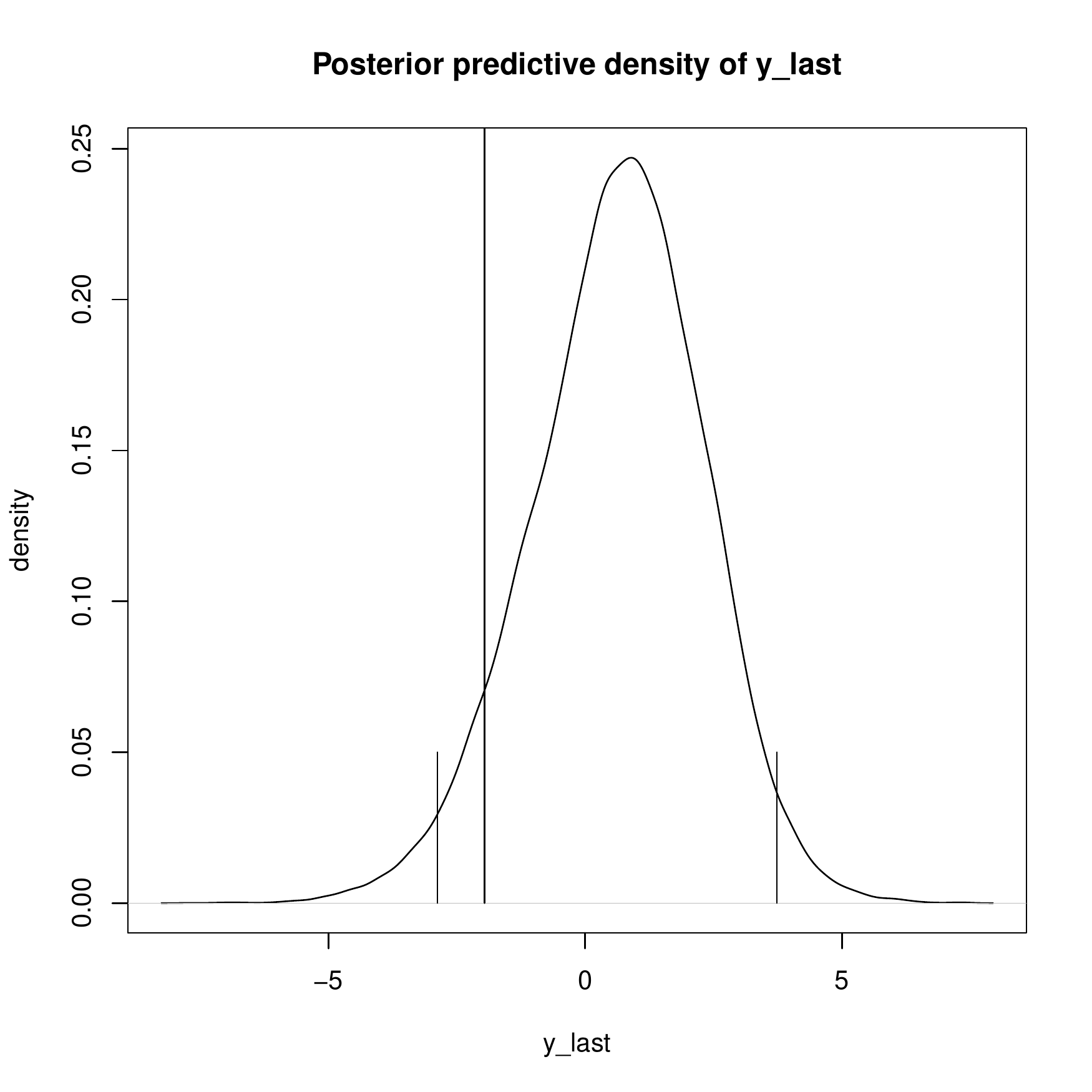}
\caption{Posterior predictive densities of $101$th observations for sample 1,2,3 and that of $20$th observation for sample 4 of Table 1, where true values are indicated with long vertical lines and 95\% credible intervals are shown with short vertical lines.}
\label{Fig:Post predictive for sample1,2,3,4}
\end{figure}

\begin{figure}[htp]
\centering
\includegraphics[height=2.5in,width=2.5in]{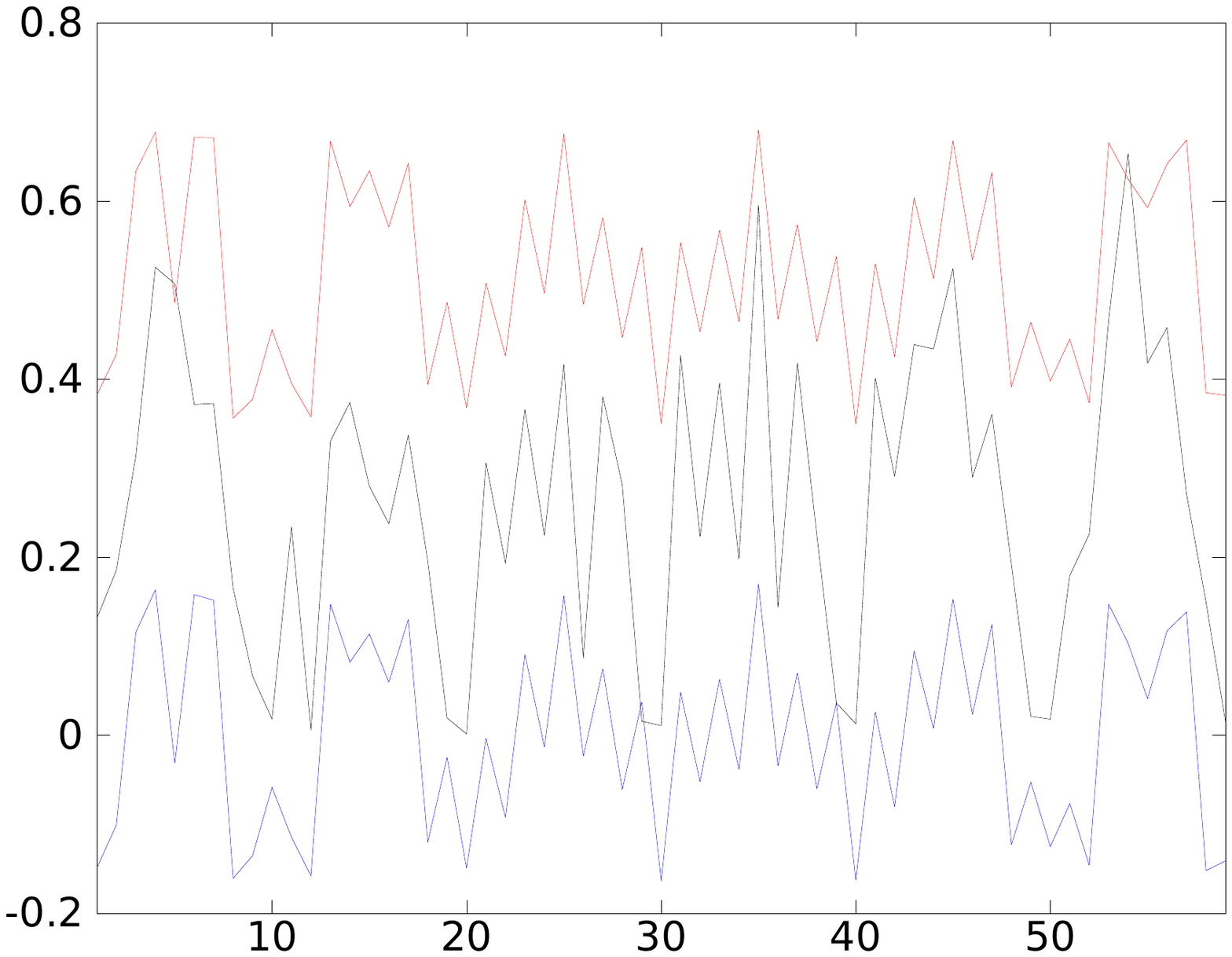}
\includegraphics[height=2.5in,width=2.5in]{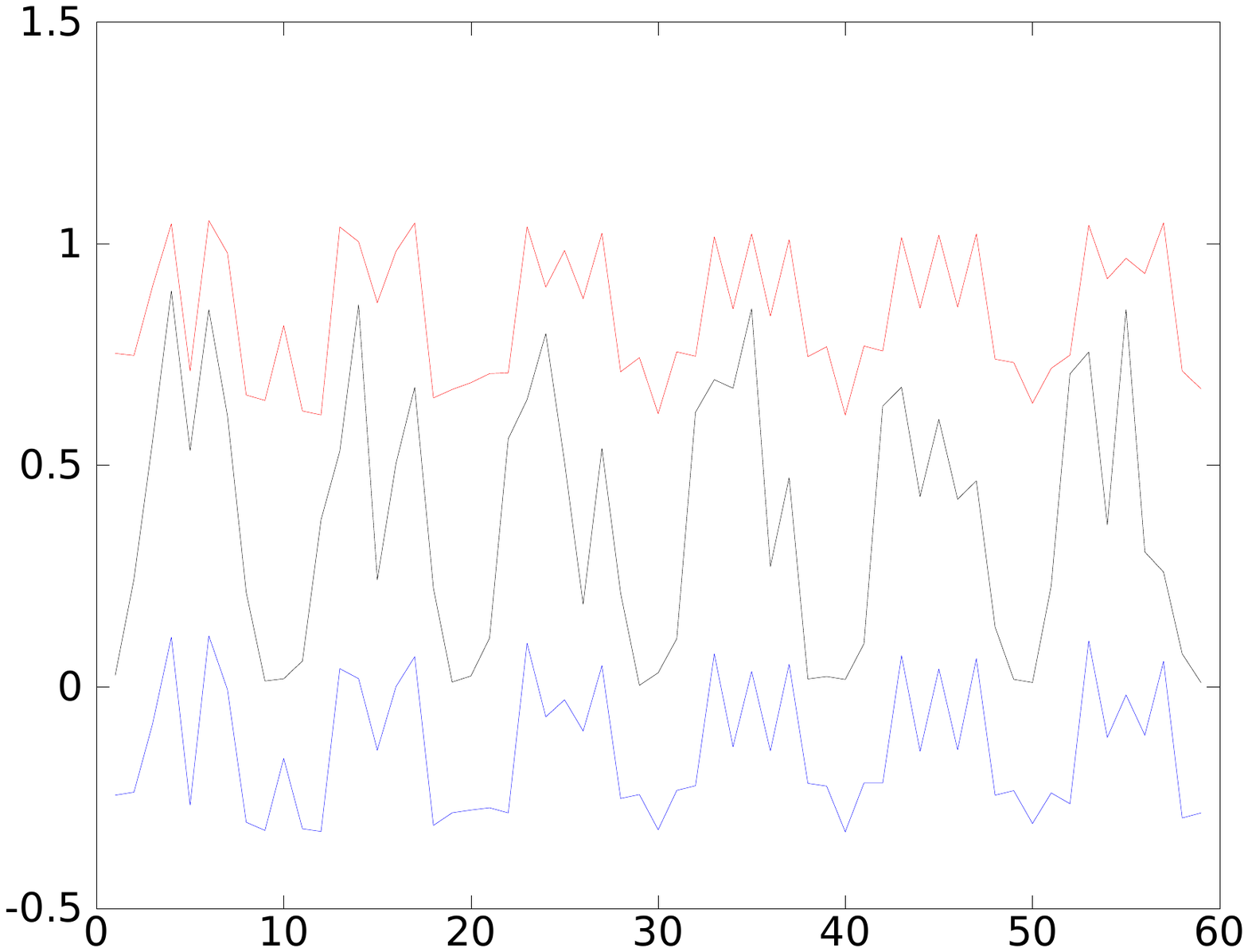}
\caption{95\% Credible interval for true signals of the sonar mine signal and the first sonar rock data, where black line shows the observed signal, blue lines indicate lower 2.5\% and red lines indicate upper 97.5\% signals obtained based on MCMC simulations.}
\label{Fig:Sonar_data_mines_1_2_rocks_3_fit}
\end{figure}

\begin{figure}[htp]
\centering
\includegraphics[height=2.5in,width=2.5in]{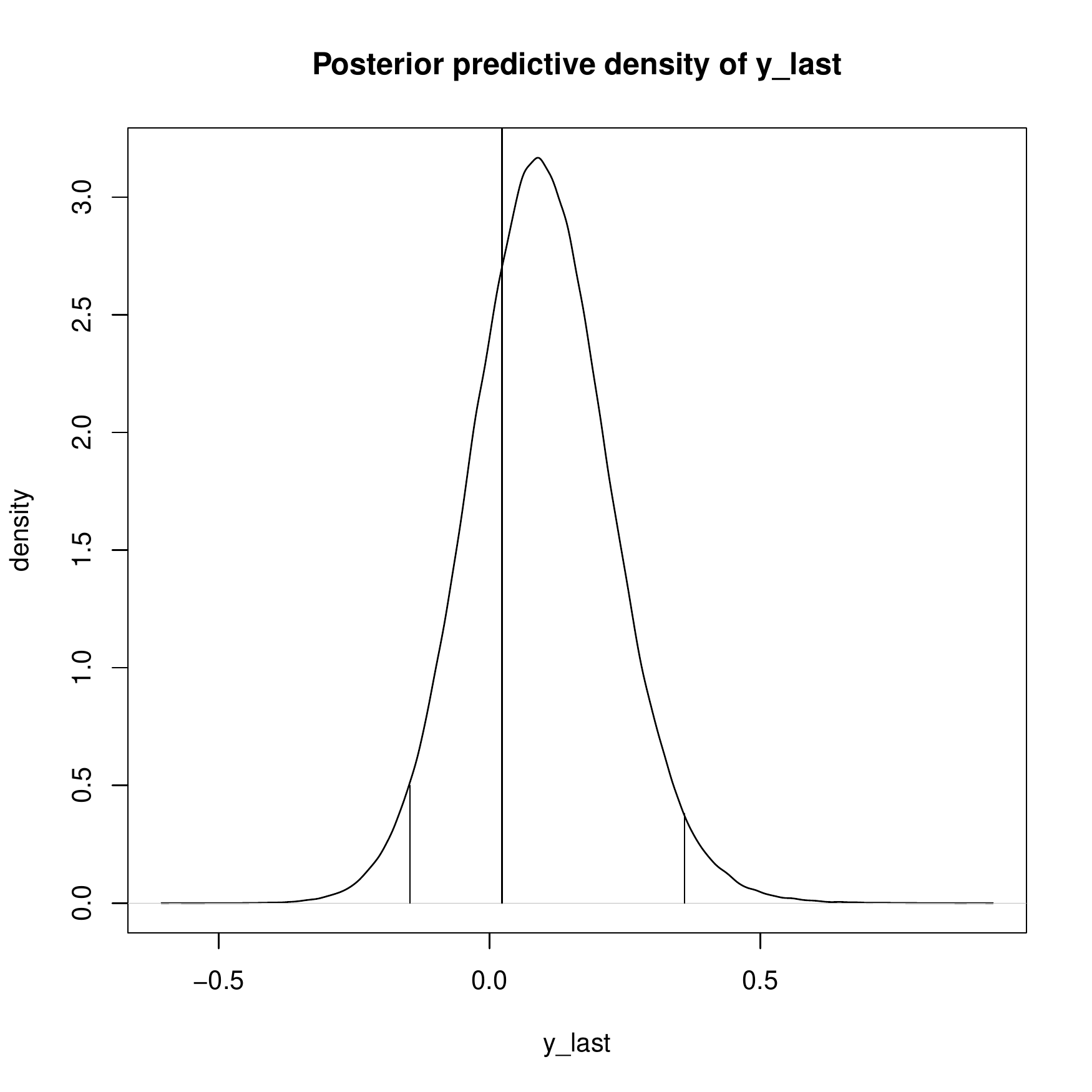}
\includegraphics[height=2.5in,width=2.5in]{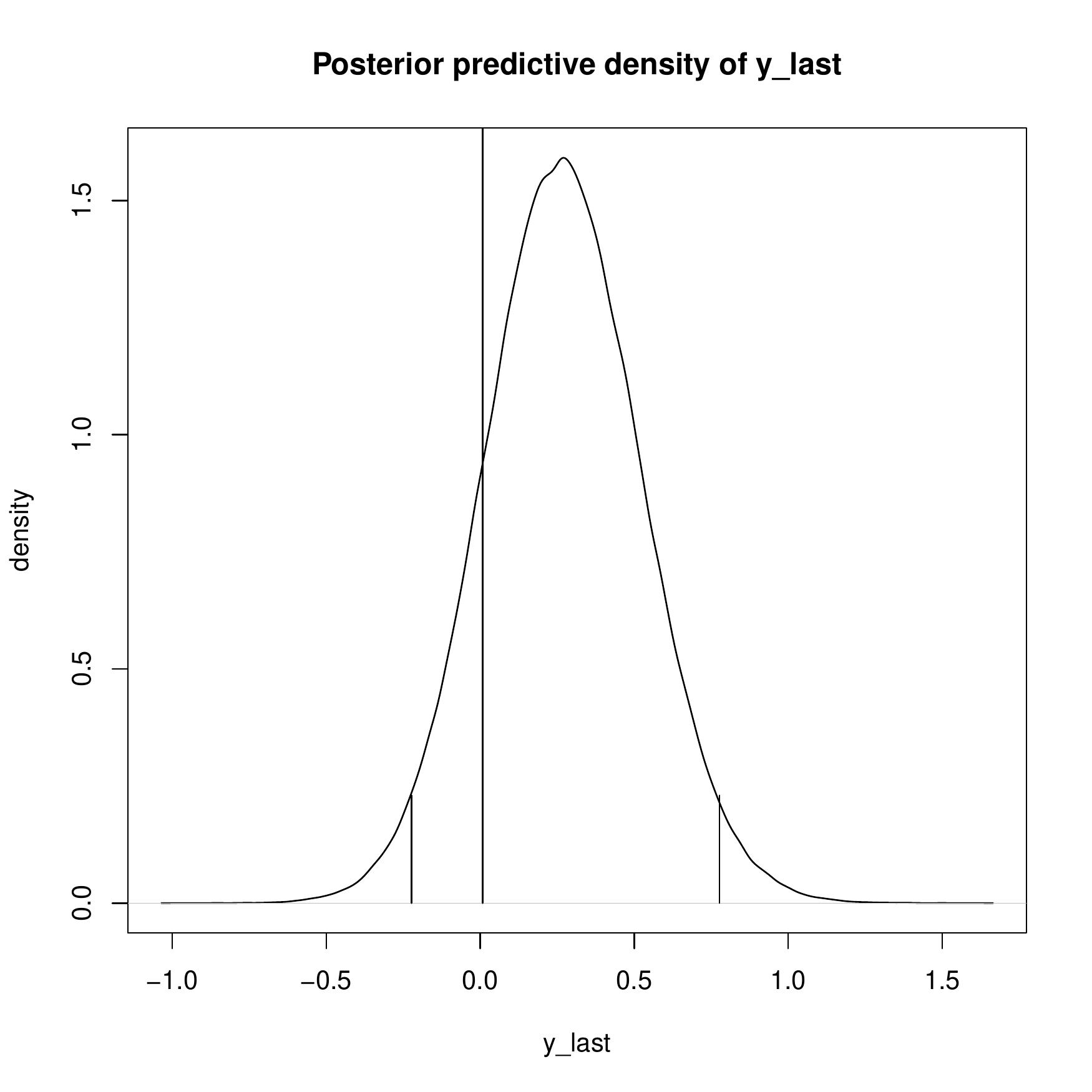}
\caption{Posterior predictive 60th observations for the sonar mine signal and the first sonar rock signal, where true values are indicated with long vertical lines and 95\% credible intervals are shown with short vertical lines.}
\label{Fig:Sonar_data_1_2_posterior_predictive_rocks_3_predictive}
\end{figure}

\begin{figure}[htp]
\centering
\includegraphics[height=1.5in,width=1.5in]{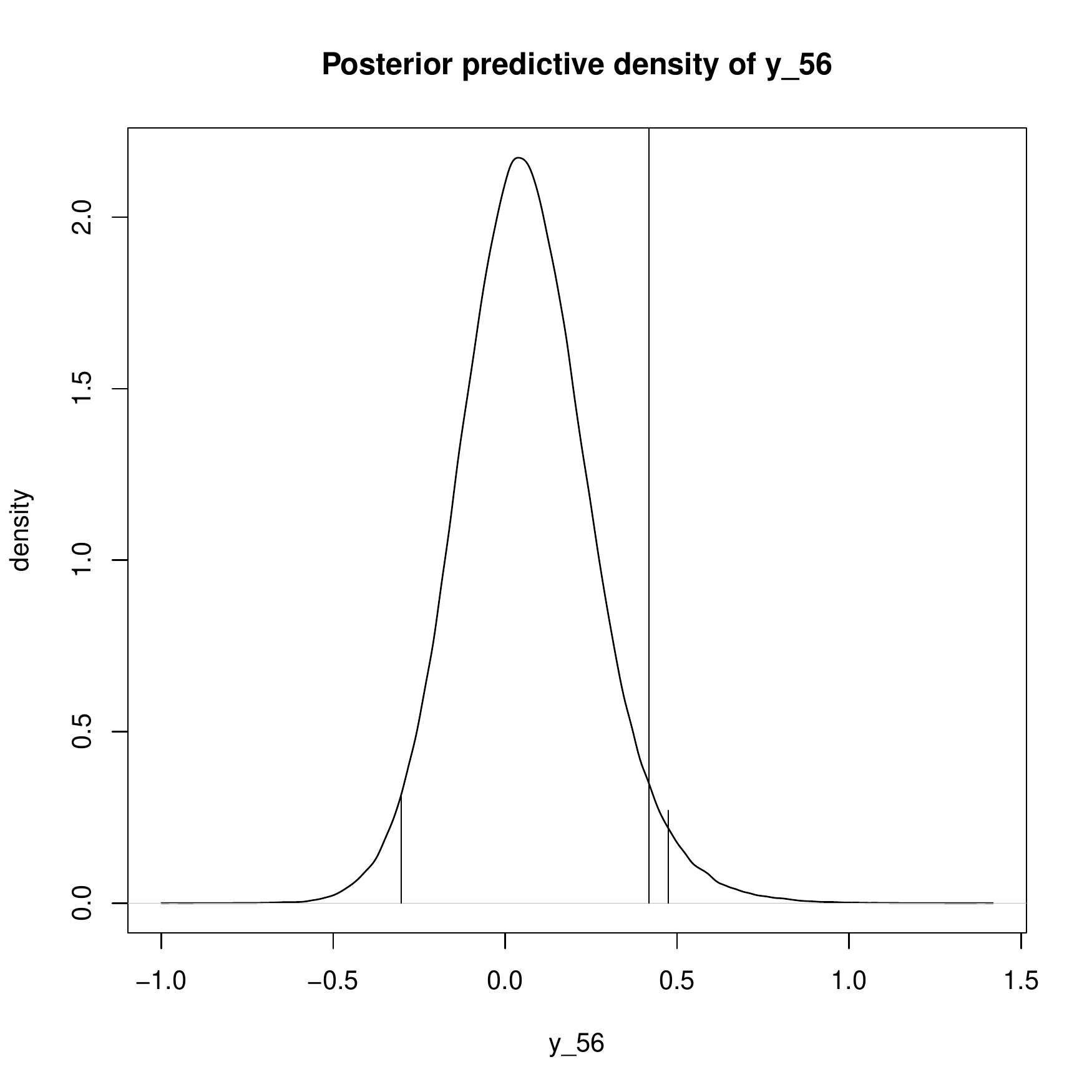}
\includegraphics[height=1.5in,width=1.5in]{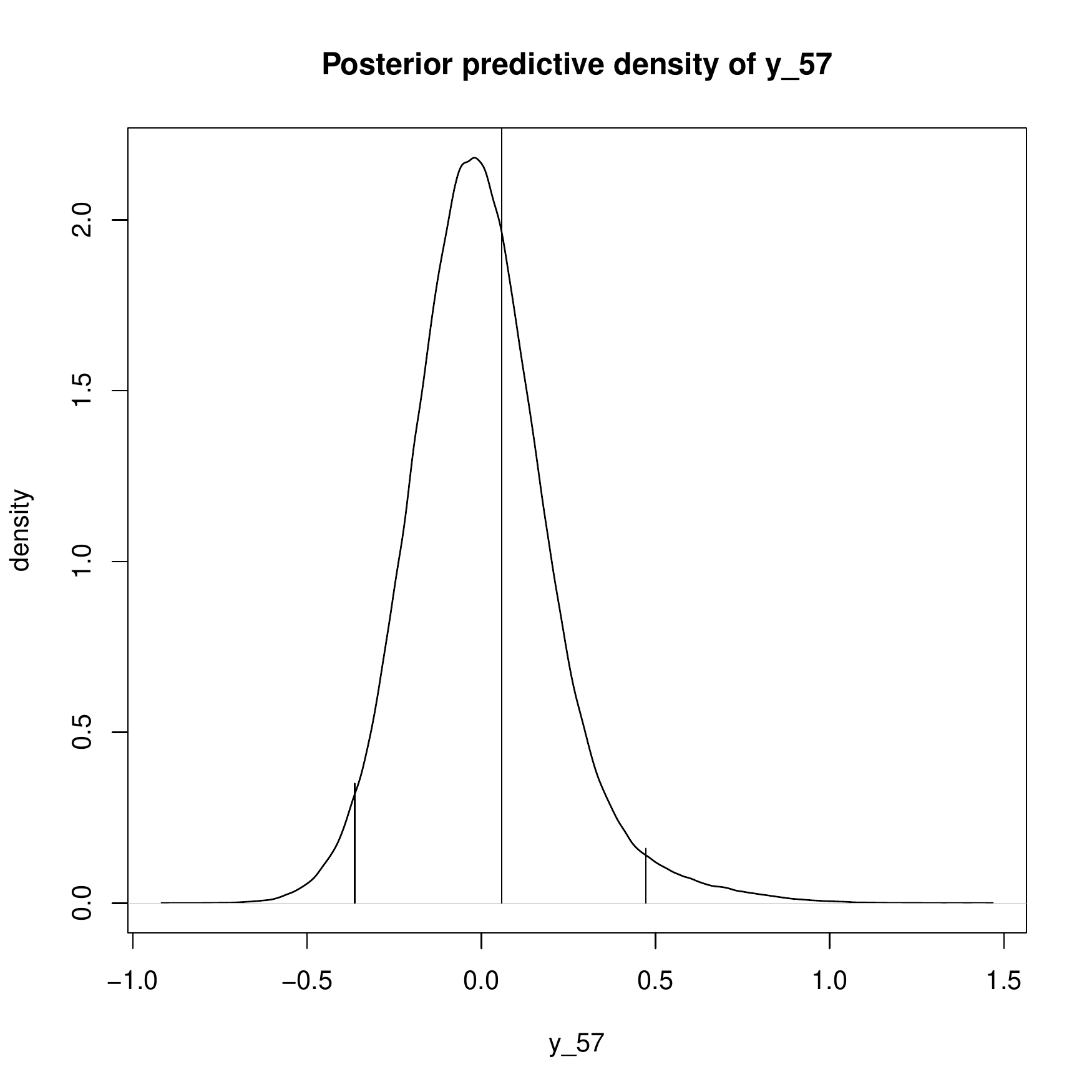}
\includegraphics[height=1.5in,width=1.5in]{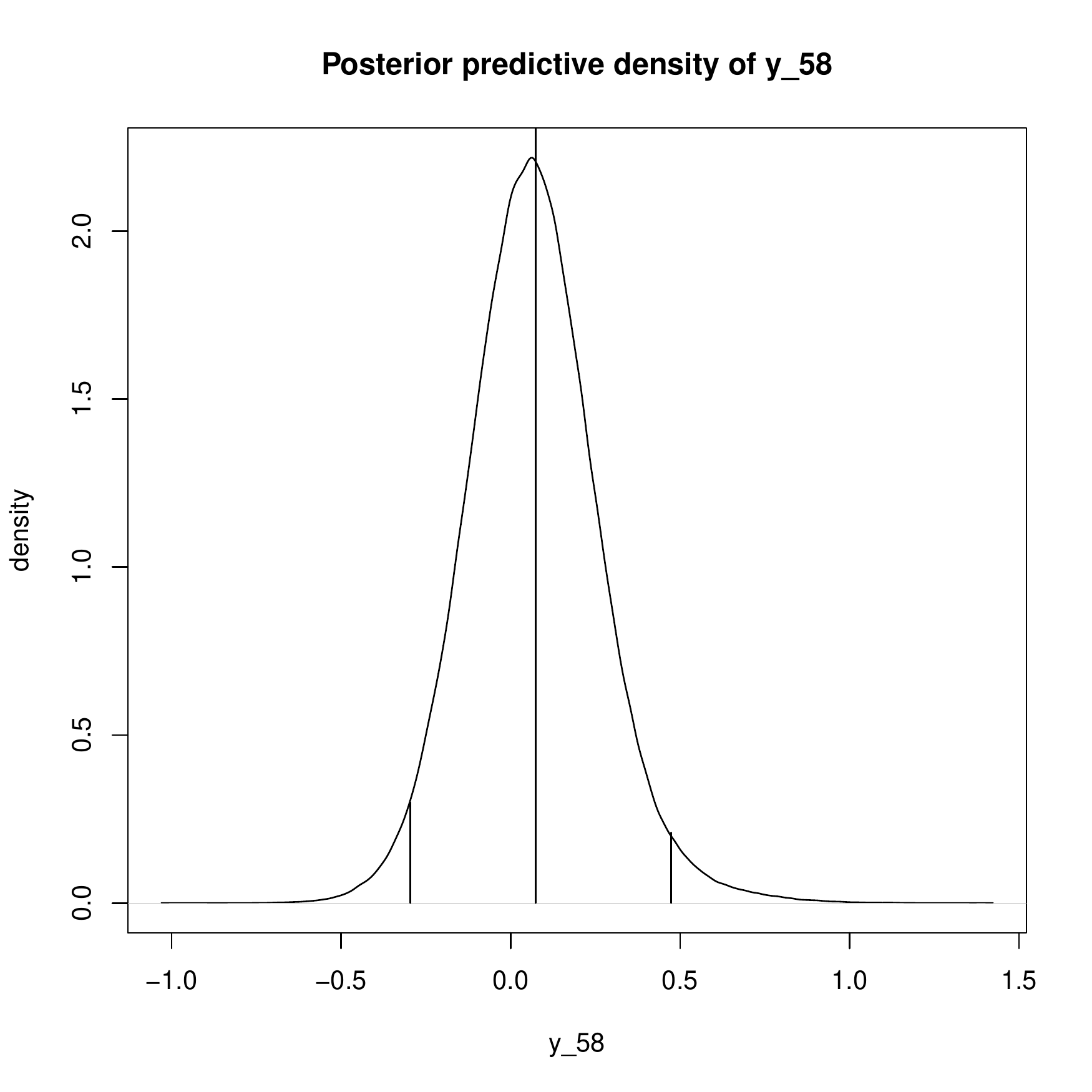}
\caption{Posterior predictive of 56th, 57th and 58th observations for the second sonar rocks signal, where true values are indicated with long vertical lines and 95\% credible intervals are shown with short vertical lines.}
\label{Fig:Sonar rocks predictions_1st three}
\end{figure}

\begin{figure}[htp]
\centering
\includegraphics[height=1.5in,width=1.5in]{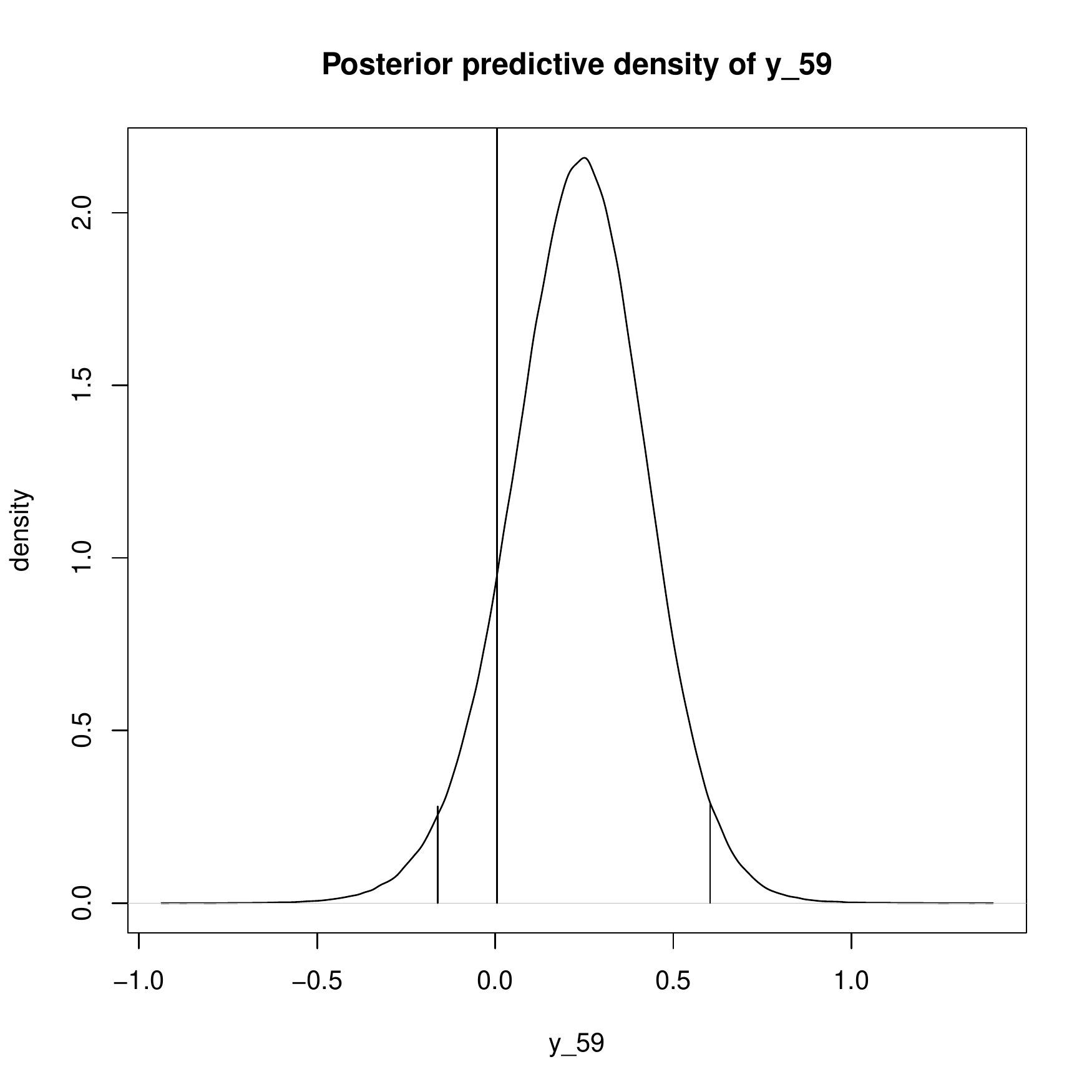}
\includegraphics[height=1.5in,width=1.5in]{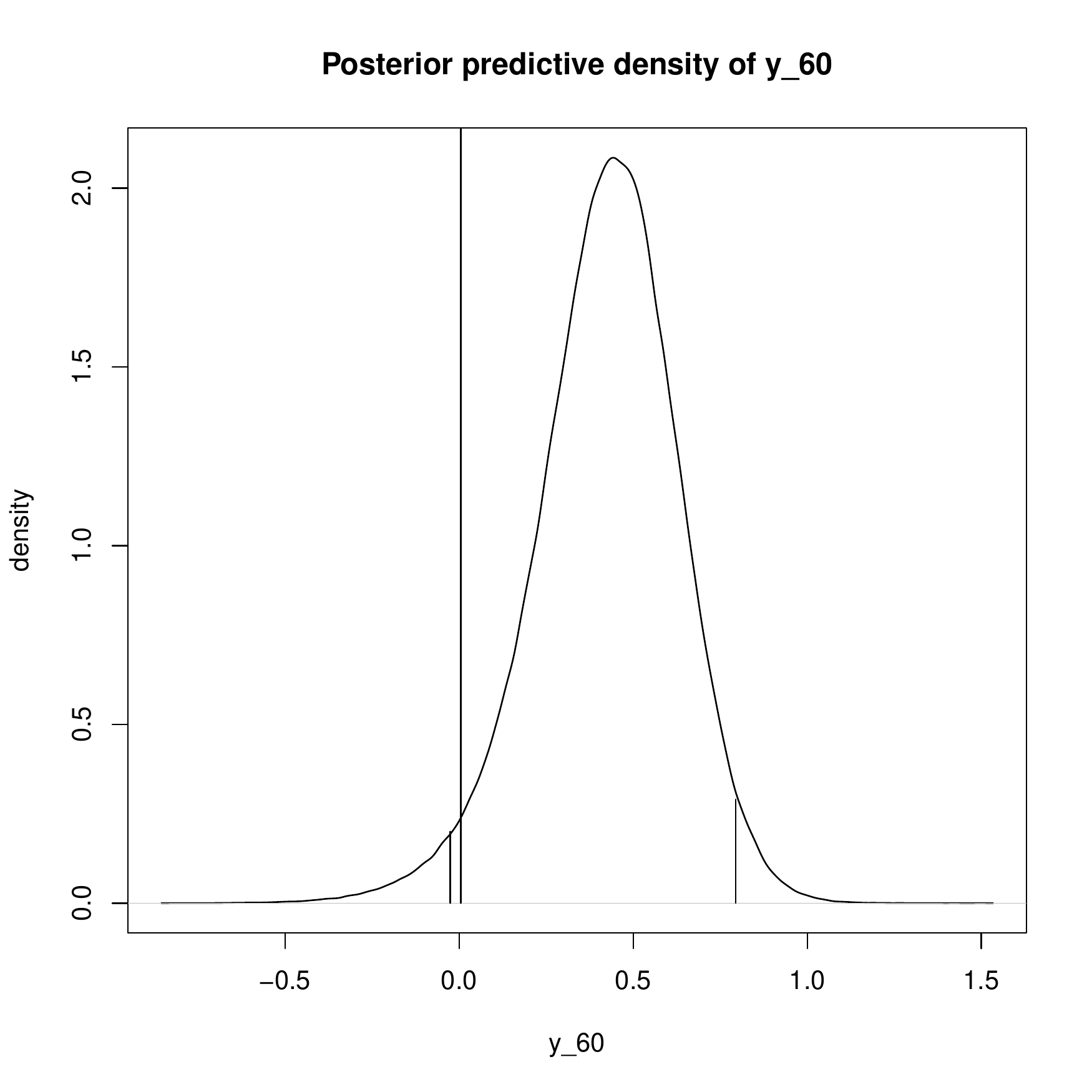}
\caption{Posterior predictive of 59th and 60th observations for the second sonar rock signal, true values are indicated with long vertical lines and 95\% credible intervals are shown with short vertical lines.}
\label{Fig: Sonar rocks predictions_last two}
\end{figure}

\begin{figure}[htp]
\centering
\includegraphics[height=1.5in,width=1.5in]{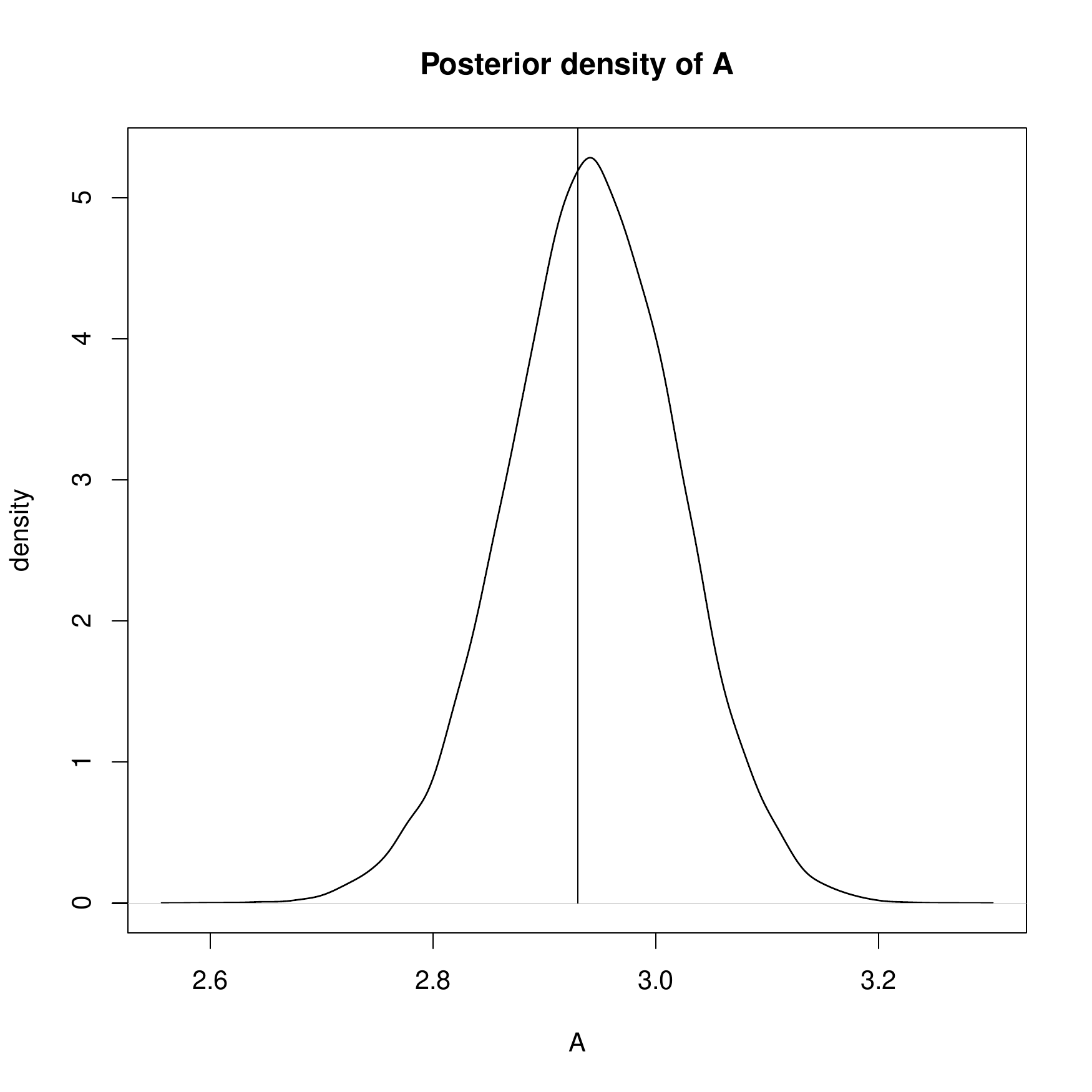}
\includegraphics[height=1.5in,width=1.5in]{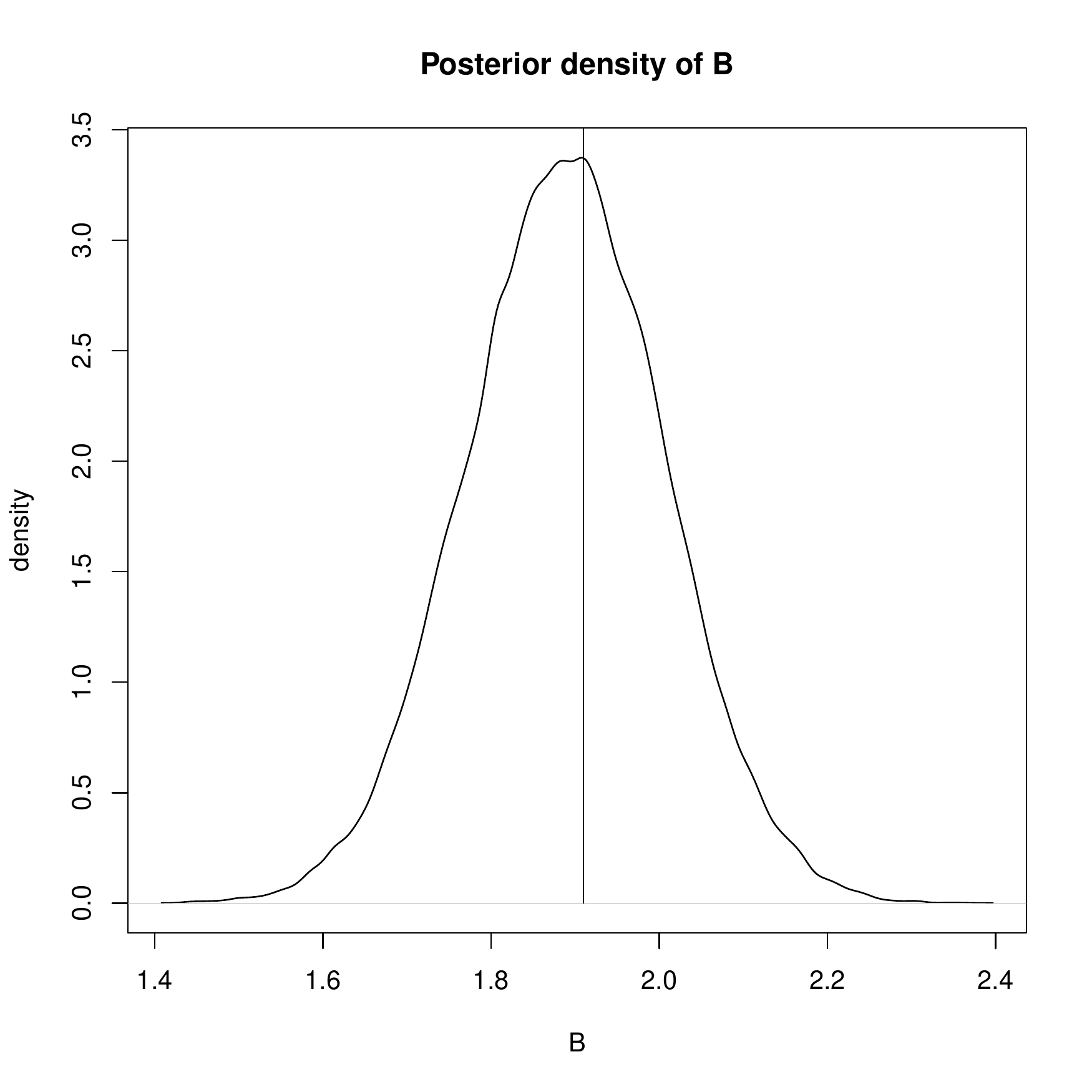}
\includegraphics[height=1.5in,width=1.5in]{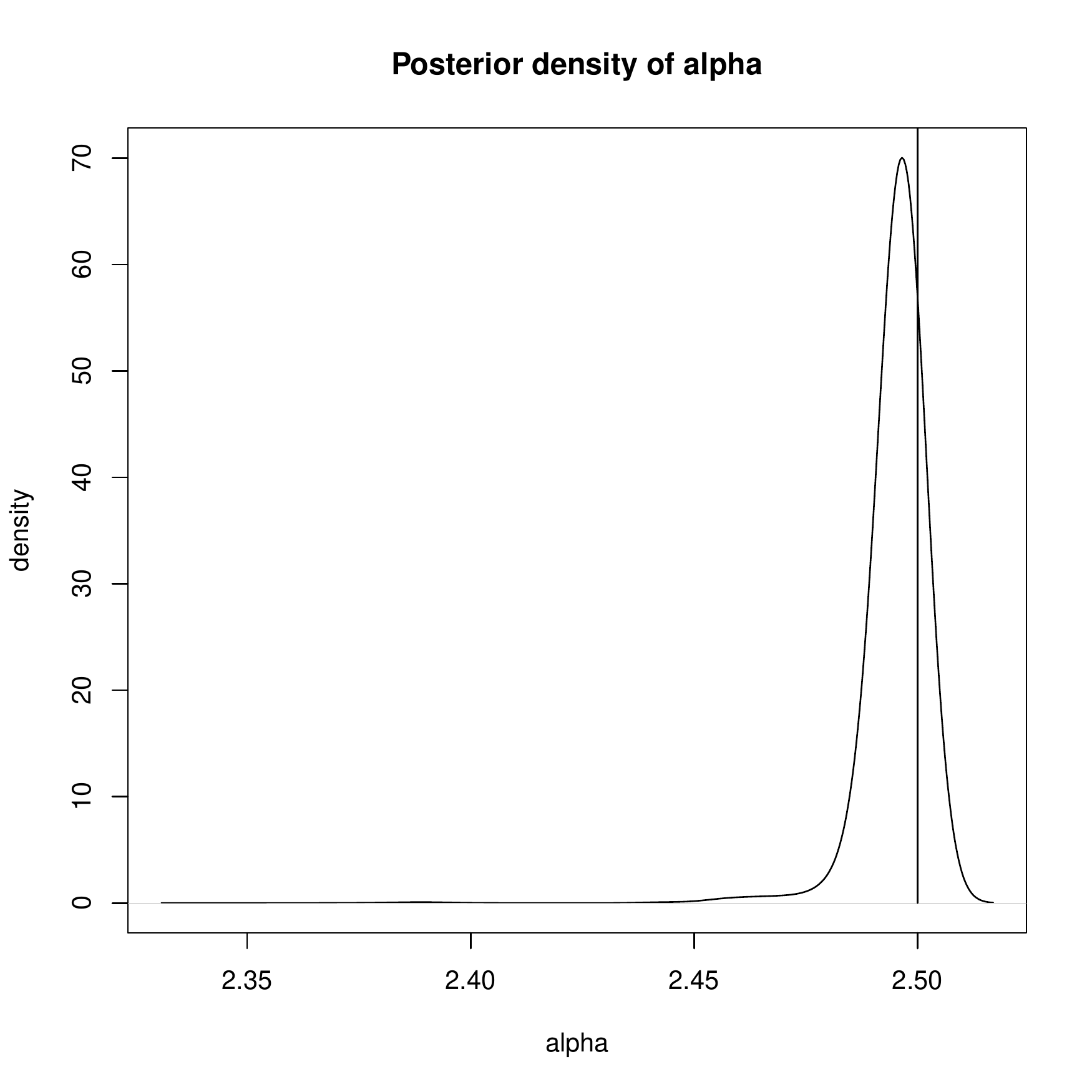}
\includegraphics[height=1.5in,width=1.5in]{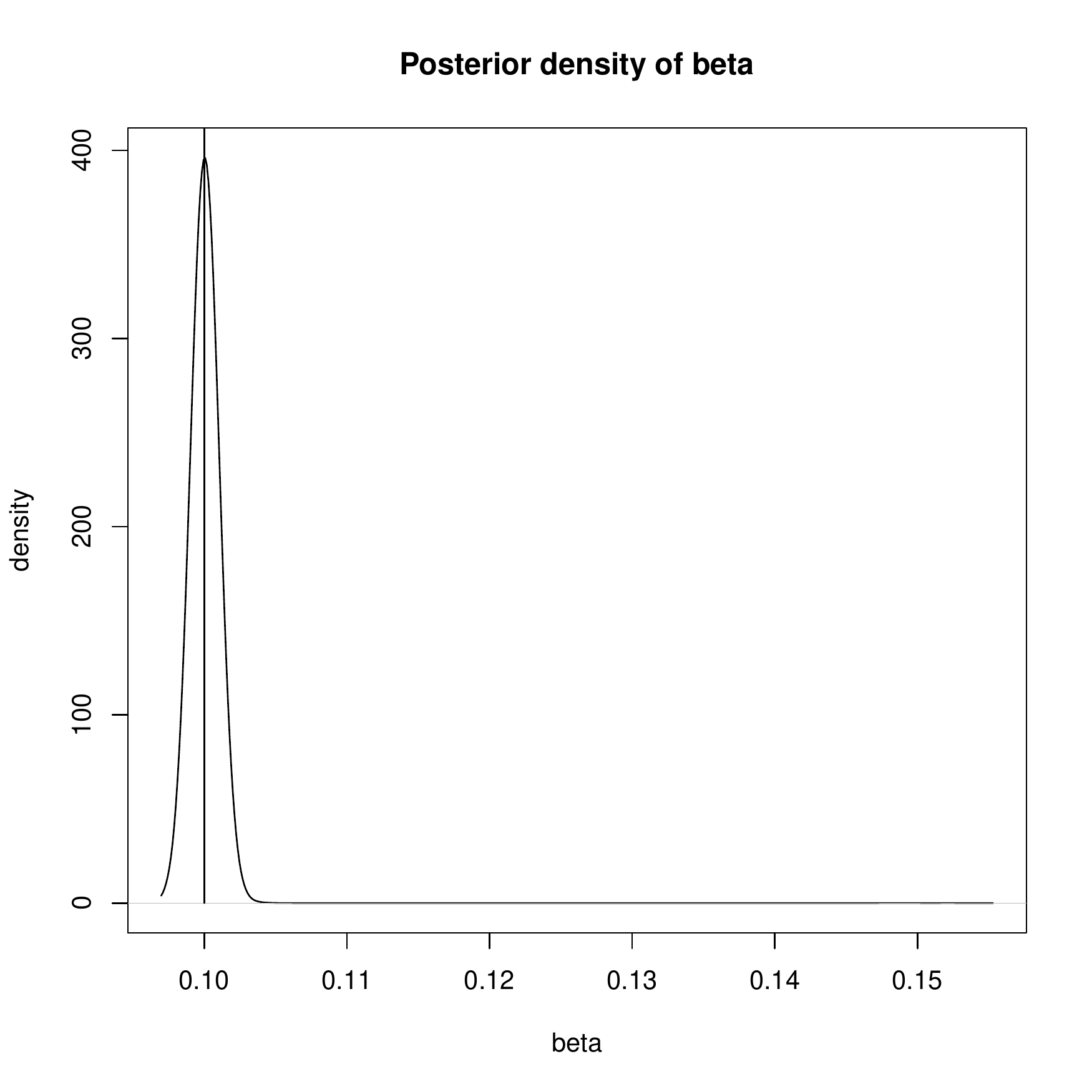}
\includegraphics[height=1.5in,width=1.5in]{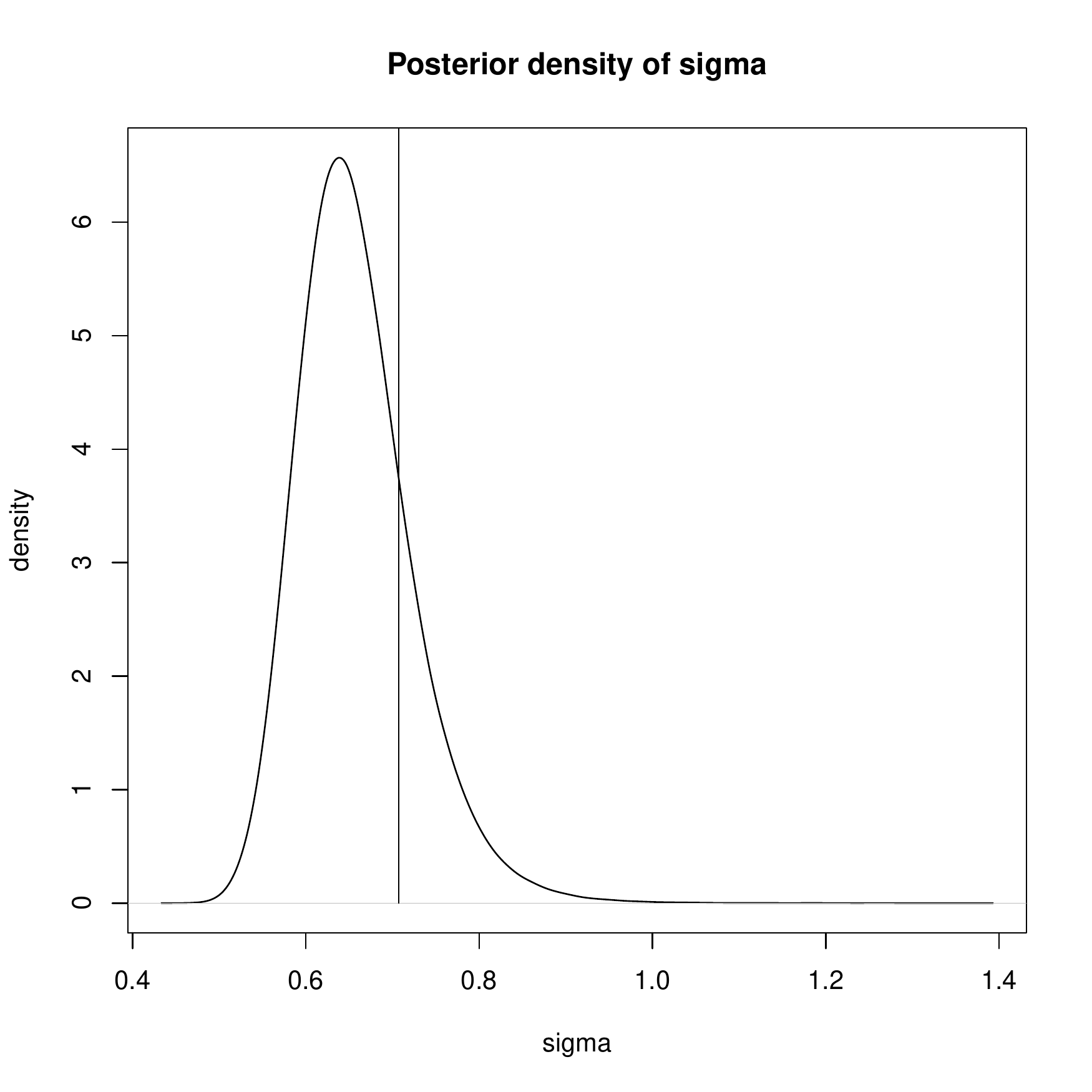}
\includegraphics[height=1.5in,width=1.5in]{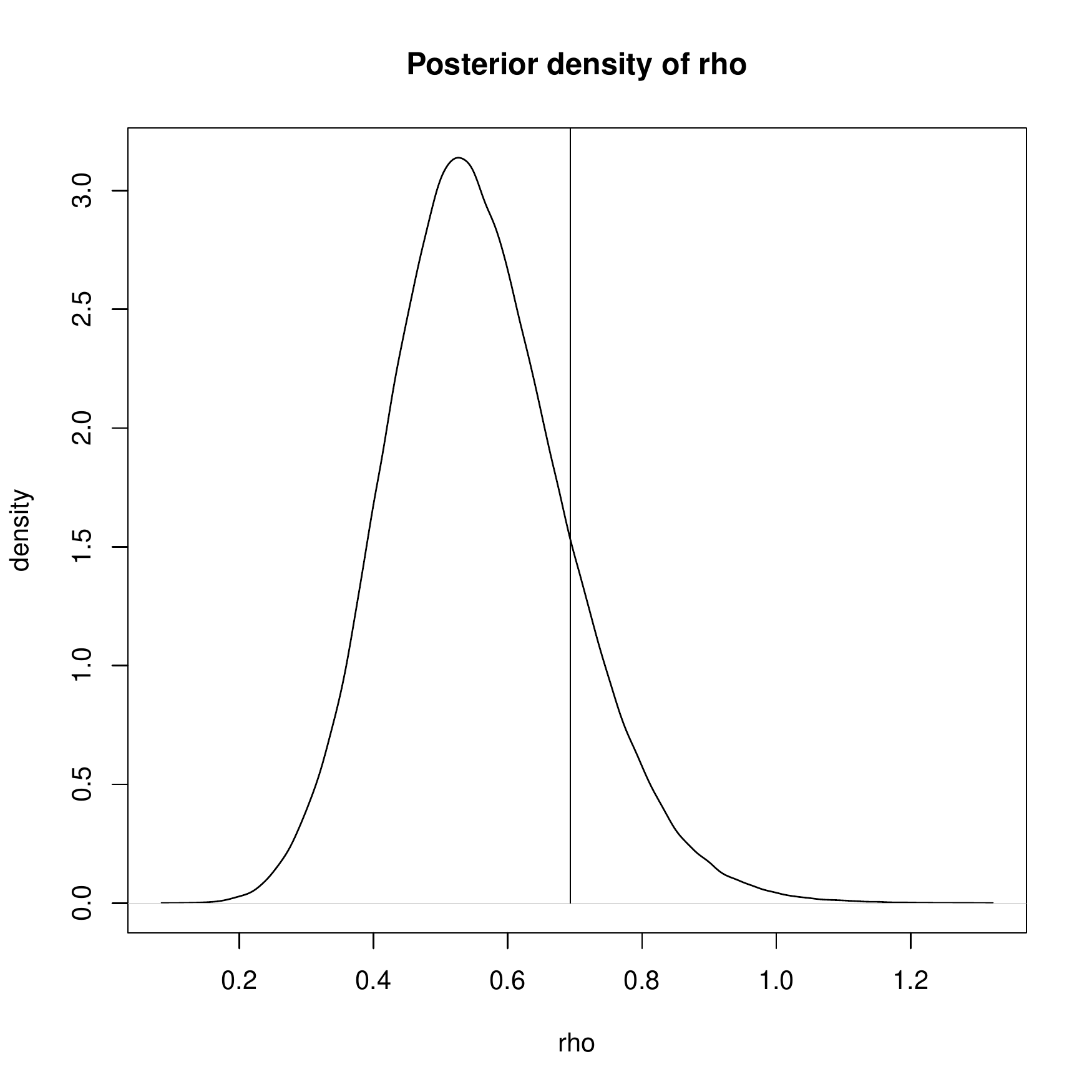}
\caption{Posterior densities of $A$, $B$, $\alpha$, $\beta$, $\sigma$ and $\rho$ for the numerical example of dependent error with exponentially decaying covariance structure, where true values are indicated with vertical lines.}
\label{Fig:Post of A,B,alpha,beta,sigma,rho for dependent sample}
\end{figure}

\begin{figure}[htp]
\includegraphics[height=2.5in,width=2.5in]{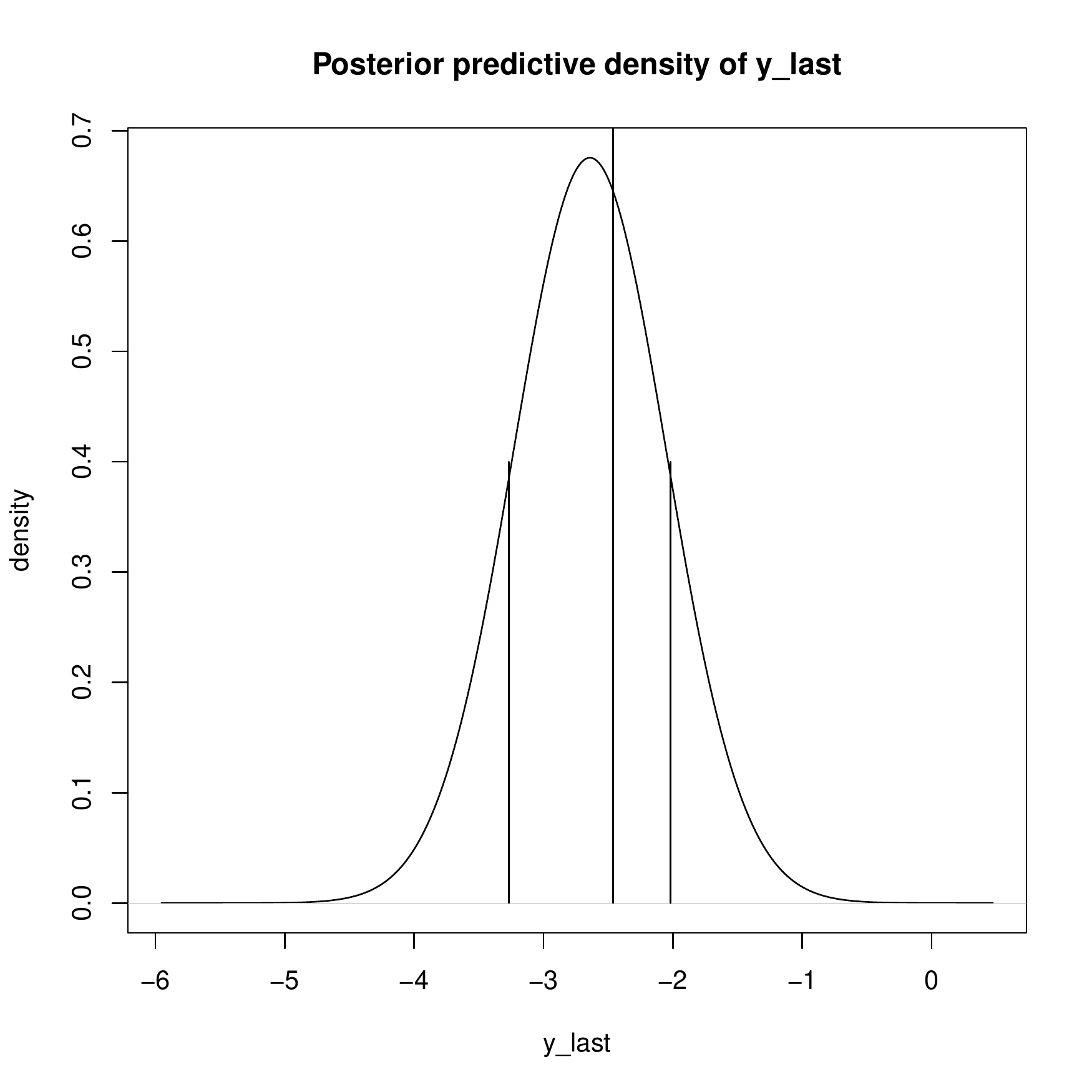}
\includegraphics[height=3.5in,width=3.5in]{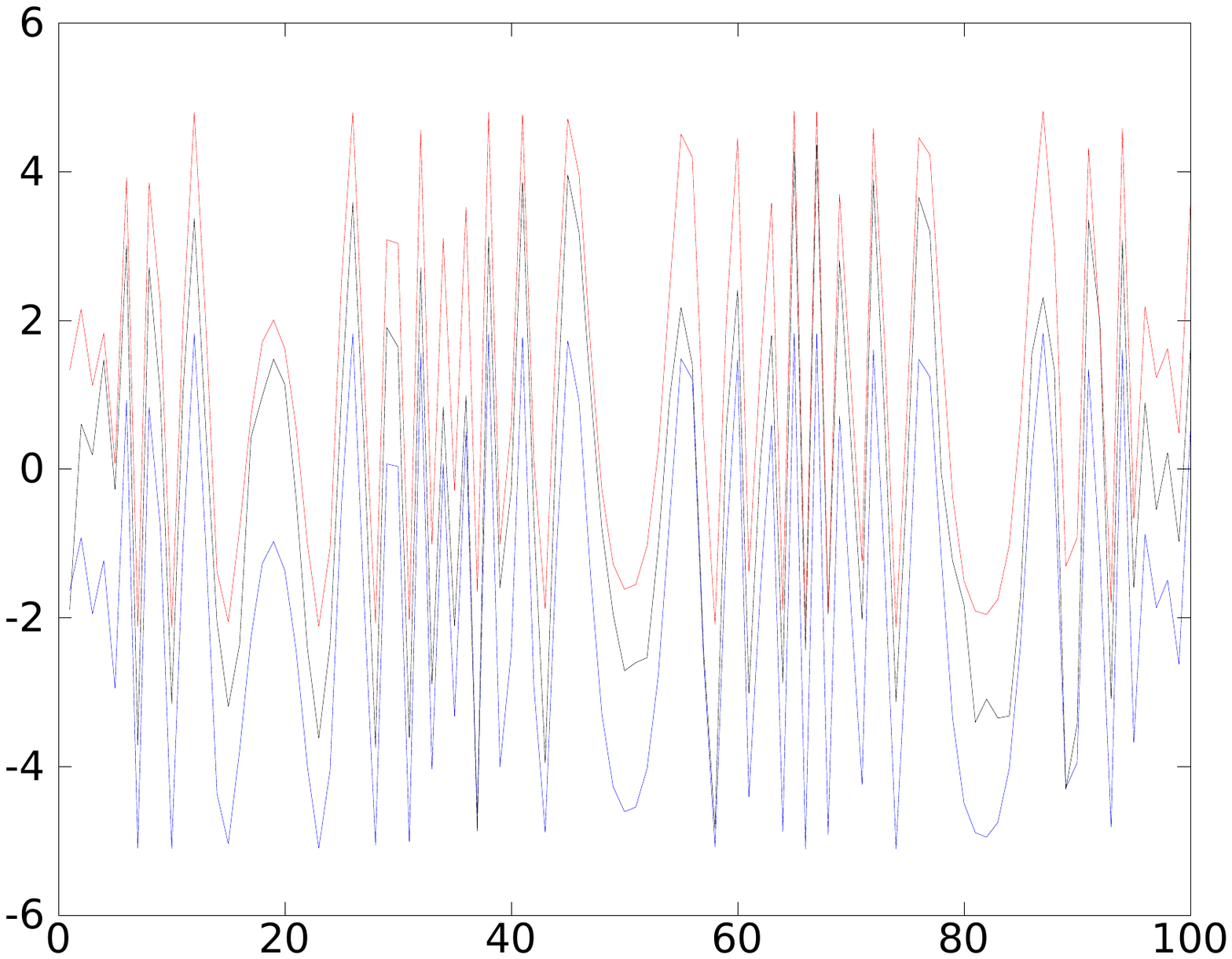}
\caption{Posterior predictive density of 101th observation, long vertical line indicates the true value and the shorter ones denote 95\% credible intervals, and 95\% credible intervals for observed signal (dependent error) with black line showing the observed signal, blue line indicating lower 2.5\% and red line indicating upper 97.5\% signals obtained based on MCMC simulations.}
\label{Fig:posterior predictive and fit for dependent error}
\end{figure}

\end{document}